\documentclass{article}

\usepackage{microtype}
\usepackage{graphicx}
\usepackage{subcaption}
\usepackage{booktabs} 

\usepackage[accepted]{icml2026}

\usepackage[utf8]{inputenc} 
\usepackage[T1]{fontenc}    
\usepackage[hypertexnames=false]{hyperref}

\usepackage{url}            
\usepackage{booktabs}       
\usepackage{amsfonts}       
\usepackage{nicefrac}       
\usepackage{microtype}      
\usepackage{xcolor}         
\usepackage{amssymb}
\usepackage{amsthm}
\usepackage{esvect}
\usepackage{amsmath}
\usepackage{cleveref}
\usepackage{mathtools}

\usepackage{algorithm}
\makeatletter
\renewcommand{\ALG@name}{Procedure} 
\makeatother
\usepackage{algorithmic}
\usepackage[old]{old-arrows}
\usepackage[left,inline,final]{showlabels}
\usepackage{etoc}
\usepackage{setspace}
\usepackage{xspace}

\usepackage{color}

\usepackage{bm}
\usepackage[font=small,labelfont=bf]{caption}
\usepackage{thm-restate}
\usepackage{mathrsfs}
\usepackage{xspace}
\usepackage{enumitem}
\usepackage{tikz} 
\usepackage{comment}
\usepackage{nicefrac}
\usetikzlibrary{calc,patterns,arrows.meta,positioning}

\definecolor{mygreen}{rgb}{0.0, 0.5, 0.0}
\definecolor{myorange}{rgb}{0.55, 0.62, 1}

\definecolor{niceRed}{RGB}{190,38,38}
\definecolor{Red2}{RGB}{219, 50, 54}
\definecolor{mgreen}{RGB}{160, 200, 140}
\definecolor{blueGrotto}{RGB}{5,157,192}
\definecolor{limeGreen}{HTML}{81B622}
\definecolor{myellow}{rgb}{0.88,0.61,0.14}
\definecolor{darkGreen}{HTML}{2E8B57}
\definecolor{navyBlueP}{HTML}{03468F}
\definecolor{Sepia}{HTML}{7F462C}
\definecolor{red2}{HTML}{1F462C}
\definecolor{orange2}{HTML}{FF8000}
\definecolor{mgray}{HTML}{ABB3B8}
\definecolor{lgray}{HTML}{E5E8E9}
\definecolor{myPurple}{RGB}{175,0,124}
\definecolor{mypurple2}{rgb}{0.8,0.62,1}
\definecolor{royalBlue}{HTML}{057DCD}
\definecolor{mpink}{HTML}{FC6C85}
\definecolor{lblue}{RGB}{74,144,226}
\definecolor{peagreen}{RGB}{152,193,39}
\definecolor{typ_navy}{HTML}{001f3f}
\definecolor{typ_blue}{HTML}{0074d9}
\definecolor{typ_aqua}{HTML}{7fdbff}
\definecolor{typ_teal}{HTML}{39cccc}
\definecolor{typ_eastern}{HTML}{239dad}
\definecolor{typ_purple}{HTML}{b10dc9}
\definecolor{typ_fuchsia}{HTML}{f012be}
\definecolor{typ_maroon}{HTML}{85144b}
\definecolor{typ_red}{HTML}{ff4136}
\definecolor{typ_orange}{HTML}{ff851b}
\definecolor{typ_yellow}{HTML}{ffdc00}
\definecolor{typ_olive}{HTML}{3d9970}
\definecolor{typ_green}{HTML}{2ecc40}
\definecolor{typ_lime}{HTML}{01ff70}
\definecolor{newgreen}{HTML}{83c702}
\definecolor{newpurp}{RGB}{97,96,121}
\definecolor{Andrea}{RGB}{204, 0, 153}

\definecolor{axisblue}{RGB}{20,60,140}
\definecolor{diaggreen}{RGB}{40,140,70}
\definecolor{regionred}{RGB}{200,40,40}

\hypersetup{
	colorlinks = true,
	linkcolor = typ_blue,
	citecolor = typ_maroon,
	linktocpage = true,
	urlcolor = darkGreen
}

\theoremstyle{plain}
\newtheorem{theorem}{Theorem}[section]

\newtheorem{lemma}[theorem]{Lemma}
\newtheorem{corollary}[theorem]{Corollary}
\theoremstyle{definition}

\theoremstyle{remark}
\newtheorem{remark}[theorem]{Remark}

\newcommand{\E}{\mathbb{E}}

\newcommand{\Var}{\mathbb{V}\mathrm{ar}}

\newcommand{\supp}{\textup{supp}}

\newcommand{\argmin}{\operatornamewithlimits{argmin}}
\newcommand{\argmax}{\operatornamewithlimits{argmax}}

\newcommand{\Prob}{\mathbb{P}}

\makeatletter
\newcommand{\mapstoto}{\mathpalette\@mapstoto\relax}
\newcommand*{\@mapstoto}[2]{%
    \mathrel{%
      \vcenter{%
         \vbox{%
            \baselineskip\z@skip
            \lineskip\z@
            \ialign{##\cr$#1\mapstochar\varrightarrow$\cr
            $#1\mapstochar\varrightarrow$\cr}%
         }%
      }%
   }%
}
\makeatother

\newcommand{\ie}{\emph{i.e.},\xspace}
\newcommand{\eg}{\emph{e.g.},\xspace}

\crefname{assumption}{Assumption}{Assumptions}

\icmltitlerunning{Optimal Rates for Feasible Payoff Set Estimation in Games}

\begin{document}

\etocdepthtag.toc{mtchapter}
\etocsettagdepth{mtchapter}{subsection}
\etocsettagdepth{mtappendix}{none}

\twocolumn[
  \icmltitle{Optimal Rates for Feasible Payoff Set Estimation in Games}



  \icmlsetsymbol{equal}{*}

  \begin{icmlauthorlist}
    \icmlauthor{Annalisa Barbara}{yyy,eth}
    \icmlauthor{Riccardo Poiani}{yyy}
    \icmlauthor{Martino Bernasconi}{yyy}
    \icmlauthor{Andrea Celli}{yyy}
  \end{icmlauthorlist}

  \icmlaffiliation{yyy}{Bocconi University, Milano, Italy}
  \icmlaffiliation{eth}{ETH Zurich, Switzerland}
%
  \icmlcorrespondingauthor{Riccardo Poiani}{riccardo.poiani@unibocconi.it}

  \icmlkeywords{Machine Learning, ICML}

  \vskip 0.3in
]



\printAffiliationsAndNotice{}  

\begin{abstract}

We study a setting in which two players play a (possibly approximate) Nash equilibrium of a bimatrix game, while a learner observes only their actions and has no knowledge of the equilibrium or the underlying game. A natural question is whether the learner can rationalize the observed behavior by inferring the players' payoff functions. 
Rather than producing a single payoff estimate, inverse game theory aims to identify the entire set of payoffs consistent with observed behavior, enabling downstream use in, \emph{e.g.}, counterfactual analysis and mechanism design across applications like auctions, pricing, and security games.
We focus on the problem of estimating the set of feasible payoffs with high probability and up to precision $\epsilon$ on the Hausdorff metric. 
We provide the first minimax-optimal rates for both exact and approximate equilibrium play, in zero-sum as well as general-sum games. Our results provide learning-theoretic foundations for set-valued payoff inference in multi-agent environments.

\end{abstract}
\section{Introduction}

The central question in \emph{inverse game theory} is the following \citep{waugh2011computational, kuleshov2015inverse}: given observed equilibrium behavior, what utility functions could have generated it? 
%
%
We study the problem of recovering the \emph{entire} set of possible payoffs compatible with the players' observed behavior, given action samples generated by their strategies. We are interested in the minimax sample complexity of approximately identifying such a set.
Even when equilibrium behavior is observed exactly (let alone through samples), the problem of inferring an agent’s utilities from behavior is inherently underdetermined, as multiple payoff functions can rationalize the same observed play~\citep{ng2000algorithms,waugh2011computational}.
This issue is often referred to as \emph{unidentifiability}. For example, the zero payoff matrix is compatible with any equilibrium. This ambiguity is usually addressed by imposing additional structure, such as restricting attention to  parametrized low-dimensional payoff families \citep{waugh2011computational, kuleshov2015inverse}, or by choosing a priori a selection principle such as maximum entropy \citep{vorobeychik2007learning, waugh2011computational}, or bounded-rationality models (such as quantal response) \citep{yu2022inverse}.

Rather than imposing additional modeling assumptions or fixing an \emph{a priori} selection principle, we focus on identifying the entire set of payoff functions consistent with the observed strategies, \emph{i.e.}, the feasible payoff set. The idea of this approach is to defer the selection of a particular payoff to a later stage, thereby making the entire process modular and more adaptable to the requirements of the application at hand.
For example, given the entire set of possible equilibria, one may select equilibria according to a maximum-entropy criterion \citep{waugh2011computational}, add additional structural constraints to the set (\eg if we know that the game is symmetric or has a potential-game structure), or choose a robust representative such as the incenter solution \citep{cui2025inverse}.
This set-valued perspective is standard in structural econometrics and inverse game theory (see, among others, \citet{brown1996testable,manski2003partial,tamer2010partial,waugh2011computational}).
However, the corresponding learning problem remains poorly understood from a theoretical standpoint, and even in the simplest multi-agent settings, we lack a sharp characterization of the \emph{statistical} difficulty of such set-valued identification.

In this work, we address this gap by providing minimax-optimal sample complexity guarantees for the problem of recovering the feasible payoff set from the observed behavior of two players in an unknown bimatrix game playing a Nash equilibrium. In particular, we consider two players repeatedly interacting in an unknown $n$-dimensional normal-form bimatrix game with payoff matrices $(A,B)$.
The players play a (possibly approximate) Nash equilibrium $(x,y)$ of the game, and the learner observes i.i.d.~action pairs $(X_t,Y_t)$ sampled from the product distribution $x \times y$.
The learner's goal is to infer the \emph{feasible payoff set}:
the set of all payoff matrices $(A',B')$ for which the strategies $(x,y)$ constitute a Nash equilibrium (or an $\alpha$-Nash equilibrium) of $(A',B')$. The goal of the paper is to answer the following question: \emph{How many equilibrium samples are required to estimate the feasible payoff set up to Hausdorff error $\epsilon$ with probability at least $1-\delta$?}
Our main contribution is the first minimax-optimal characterization of this question. Our results cover both \emph{general-sum} and \emph{zero-sum} games, and both approximate and exact Nash equilibria. 



\paragraph{Technical challenges} An important technical ingredient in our analysis is the correspondence (a set-valued function) that maps equilibrium strategies to the feasible payoff sets that rationalize them. Regarding the lower bounds, we first show that, when observing exact equilibria (\ie $\alpha=0$), the strong discontinuity of such correspondence makes the problem unlearnable. Therefore, we relax the learning requirement and provide minimax-optimal rates when the strategy profiles have a minimum probability of playing each action that is bounded away from $0$.
Conversely, when $\alpha > 0$, we can show that such discontinuities are somewhat mitigated. Interestingly, our results reveal that small values of $\alpha$ can substantially degrade the rates through an unavoidable dependence on $\alpha^{-1}$. Informally, the cause is that when $(x,y)$ play some actions with probability $\propto \alpha$, small changes in the strategy profiles lead to large changes in the payoff sets.
For the upper bounds, the major technical obstacle in obtaining tight rates is establishing this $\alpha^{-1}$ dependence in the approximate-equilibrium case.
To this end, focusing on a term that depends only on the $x$-player, we first show that the Hausdorff error scales with $\mathcal{O}(\alpha^{-1}(\hat x - x)^\top A y)$, where $\hat x$ is the empirical mean of $x$ after $m$ samples and $A$ is any feasible payoff for $(x,y)$ (an analogous argument holds for the $y$-player).
Hence, to get a tight rate, we need to show that the empirical value of the game $\hat x^\top Ay$ converges to the true value $x^\top Ay$ as ${\mathcal{O}}(\sqrt{\alpha (n+\log(1/\delta))/m})$, where $m$ is the number of samples (while standard concentration of $\hat x$ to $x$ would only yield ${\mathcal{O}}(\sqrt{n+\log(1/\delta)/m})$). We argue that such rates are achievable by studying a fractional knapsack problem that arises in controlling the error. 
\subsection{Related Works}

Early empirical work on payoff learning treats utilities as regression targets when realized utilities are observed for sampled strategy profiles, enabling supervised payoff approximation~\citep{vorobeychik2007learning}.
Closer in spirit to our setting, \citet{kuleshov2015inverse} introduced inverse game theory in the context of succinctly-representable games and studied the computational problem of finding utilities that rationalize an observed equilibrium (primarily via feasibility of linear/convex programs under structural restrictions). \citet{waugh2011computational} study an inverse equilibrium problem (for correlated equilibria) in a parametrized setting, and recover a polytope of compatible payoff functions used to compute the maximum-entropy equilibrium.
Recently, \citet{Goktas} studied how to recover a \emph{single} payoff-parameter/Nash-equilibrium pair consistent with observed behavior, relying on some form of oracle access to payoffs during the inverse optimization. This is fundamentally different from our setting, where we aim to recover the feasible set of payoff matrices from i.i.d.~equilibrium action samples alone.
A related line of works focus on learning adversary models in Stackelberg games, with the main goal of enabling behavioral modeling of attackers in security games \cite{Haghtalab2016three, Sinha2016learning, wu2022inverse}. This problem differs from ours in the fact that the learner can probe the adversary (\emph{i.e.}, the player with unknown payoffs) by observing their responses to ad hoc strategies.
In the context of sponsored search auctions, \citet{nekipelov2015econometrics} studies how to estimate players' valuations from observed contextual data, by assuming each player places bids following a no-regret algorithm.
Finally, \citet{chunkai} attempts to fit parametric games end-to-end by replacing Nash equilibrium with differentiable fixed points, such as quantal response equilibrium, enabling gradient-based estimation of game parameters.

Related to our work is the single-agent version of our problem, usually studied in inverse reinforcement learning (IRL). 
In IRL, the objective is to infer rewards for which a demonstrated policy is optimal \citep{abbeel2004apprenticeship}. It is well known that even in single-agent settings, the problem is fundamentally \emph{ill-posed}, as multiple reward functions may be compatible with the observed behavior \cite{ng2000algorithms}. Analogously to the game theory setting, there have been several attempts to resolve this ambiguity by focusing on specific selection criteria such as maximum margin, maximum entropy, and minimum Hessian eigenvalue \cite{ratliff2006maximum,metelli2017compatible,zeng2022maximum}. In this work, we take inspiration from an alternative recent line of research that attempts to overcome the ambiguity issue by estimating the \emph{feasible reward set}, \emph{i.e.}, the entire set of reward functions compatible with the observed data \cite{metelli2021provably,lindner2022active, metelli2023towards,poiani2024sub}. Remarkably, this approach captures the inherent ambiguity of the problem, supports robust prediction and counterfactual reasoning, and avoids committing to arbitrary equilibrium-selection or regularization choices.

There are many works that extend the single-agent IRL problem to the multi-agent IRL problem (MAIRL), including, \eg \citep{natarajan2010multi, yu2019multi, fu2021evaluating, reddy2012inverse, zhang2019non}. As in inverse game theory, most MAIRL approaches resolve the inherent ambiguity through a priori selection principles or additional modeling assumptions. There are recent exceptions that study the learning problem of returning the entire set of feasible rewards \citep{freihaut2024feasible, tang2024multi}. However, unlike what happens in game-theoretic applications such as mechanism design, auctions, and pricing, the primary goal of (MA)IRL is to serve as a preprocessing step within a larger pipeline: the learned model and payoff are subsequently used to train multi-agent policies, which are then evaluated in the true environment.
%
As noted by \citet{freihaut2024feasible}, in strategic settings this pipeline lacks guarantees without further restrictions (such as uniqueness of equilibria or structural assumptions that enforce it, \eg entropy-regularized Markov games), since the new learned equilibria for the approximated payoffs might perform arbitrarily bad in the original setting.

Finally, a closely related area is \emph{revealed preference} analysis, which asks whether observed price/purchase observations can be rationalized by utility maximization and constructs utility representations consistent with the available observations~\citep{samuelson1948consumption,afriat1967construction,varian1982nonparametric}. There is a related line of works studying the sample complexity of learning to predict optimal bundles of items for a buyer with an unknown demand function from some specific class (e.g., linear, separable piecewise-linear concave, and Leontief) from observed choices~\citep{beigman2006learning,zadimoghaddam2012efficiently,balcan2014learning}.
\section{Preliminaries}

\paragraph{Mathematical notation}
For $n \in \mathbb{N}$, we denote by $\Delta_n$ the $n$-dimensional simplex. Furthermore, $[n]$ denotes the set $\{1, \dots, n\}$. We also denote by $e_i$ the $i$-th standard basis vector of $\mathbb{R}^n$. For any $x \in \Delta_n$, we denote by $\supp(x) = \{i\in[n]: x_i>0\}$. Let $\mathcal{X}_1$, $\mathcal{X}_2$ be two non-empty subsets of a metric space $(\mathcal{X}, d)$. Then, the Hausdorff distance \citep{rockafellar2017stochastic} $H_d(\mathcal{X}_1, \mathcal{X}_2)$ is:
\begin{align*}
    \max \left\{ \sup_{x_1 \in \mathcal{X}_1} \inf_{x_2 \in \mathcal{X}_2} d(x_1, x_2), \sup_{x_2 \in \mathcal{X}_2} \inf_{x_1 \in \mathcal{X}_1} d(x_1, x_2) \right\}.
\end{align*}
We note that $H_d$ is defined based on the metric $d$.
In this work, we will be primarily interested in the case where $d$ is the $\ell_\infty$-norm between (subsets of) Euclidean spaces. From now on, we will drop the dependency on $d$.

\paragraph{General-Sum Games and Nash-Equilibrium}
A general-sum game (GSG) is specified by two matrices $A,B \in [-1,1]^{n \times n}$. For mixed strategies $x, y \in \Delta_n$, the expected losses incurred by the two players are $x^{\top} A y$ and $x^{\top} B y$, respectively. Each player seeks to choose a mixed strategy that minimizes their own expected loss. Given $\alpha \ge 0$, we say that a pair of strategies $x,y$ is an $\alpha$-Nash equilibrium for $\mathcal{G}$ if the following conditions are met: $x^\top A y \le e_i Ay + \alpha$ for all $i \in [n]$, and $x^\top B y \le x^\top B e_j + \alpha$ for all $j \in [n]$.
For any $x,y$ we denote by $\mathcal{G}_{\alpha}(x,y)$ the set of all matrices for which $x,y$ is an $\alpha$-Nash equilibrium, that is, the set of matrices $A,B$ that verify the following constraints for all indices $i,j \in [n]^2$:
\begin{align*}
    x^\top A y \le e_i^\top A y + \alpha \quad\textnormal{ and }\quad  x^\top B y \le x^\top B e_j + \alpha.
\end{align*}
Formally, we can define the following sets:
\begin{align*}
& \mathcal{G}^{x}_\alpha(x ,y) = \{ A \in [-1,1]^{n \times n}: x^\top Ay \le e_i^\top A y + \alpha~~\forall i \} \\
& \mathcal{G}^{y}_\alpha(x ,y) = \{ B \in [-1,1]^{n \times n}: x^\top By \le x^\top B e_j + \alpha~~\forall j \},
\end{align*}
so that we can write $\mathcal{G}_\alpha(x,y) = \mathcal{G}^{x}_\alpha(x ,y) \times \mathcal{G}^{y}_\alpha(x ,y)$.

\paragraph{Zero-Sum Games and Nash-Equilibrium}
A zero-sum game (ZSG), instead, is specified by a single matrix $A \in [-1,1]^{n \times n}$. In this case, $x,y \in \Delta_n \times \Delta_n$ is an $\alpha$-Nash equilibrium if $x^\top A e_j - \alpha\le x^\top A y \le e_i^\top A y + \alpha$ holds for every $i,j \in [n]^2$.
For any pair of mixed strategies $x,y$, we denote by $\mathcal{Z}_{\alpha}(x,y)$ the set of all matrices for which $x,y$ is an $\alpha$-Nash equilibrium, that is:
\begin{align*}
    \mathcal{Z}_{\alpha}(x,y) = & \left\{ A \in [-1,1]^{n\times n}: x^\top A e_j - \alpha\le x^\top A y\right. \\ & \hspace{.2cm}\textnormal{and } \left.x^\top A y \le e_i^\top A y + \alpha, \quad \forall i,j \in [n]^2 \right\}.
\end{align*}

Note that $\mathcal{Z}_\alpha(x,y) \subseteq \mathcal{G}_\alpha(x,y)$. Indeed, it is sufficient to slice $\mathcal{G}_\alpha(x,y)$ on the hyperplane where $A = -B$.

\begin{remark}
    When $\alpha$ is not specified, we refer to the exact equilibrium case. There, we will simply write \eg $\mathcal{G}(x,y)$ and $\mathcal{Z}(x,y)$ instead of $\mathcal{G}_\alpha(x,y)$ and $\mathcal{Z}_\alpha(x,y)$.
\end{remark}

\paragraph{Learning Framework}
We consider a setting in which two players follow a pair of mixed strategies $x, y \in \Delta_n \times \Delta_n$ that constitute an $\alpha$-Nash equilibrium of an unknown general-sum game. The objective of the agent is to recover the feasible payoff set $\mathcal{G}_{\alpha}(x,y)$ of matrices for which $(x,y)$ is an $\alpha$-Nash equilibrium.  In particular, the interaction framework is as follows: during each round $t \in \mathbb{N}$, the learner only observes $X_t, Y_t \sim x, y$. 
The learning algorithm takes in input a precision level $\epsilon >0$ together with a maximum risk parameter $\delta \in (0,1)$ and the parameter $\alpha \in [0,1)$ and is composed of (i) a stopping rule $\tau_\delta$ that controls the end of the data acquisition phase, and (ii) a recommendation rule, which we express as a function $f$ that takes in input a pair of matrices $A,B$ and returns $1$ if the agent believes that $(A,B)$ belongs to $\mathcal{G}_\alpha(x,y)$ and $0$ otherwise. We note that $f$ uniquely describes a set $\hat {\mathcal{G}}_\alpha=\{(A,B):f((A,B))=1\}$, which denotes the guess of the agent for all the pair of matrices belonging to $\mathcal{G}_{\alpha}(x,y)$. Hence, in the following, we will refer to $\hat{\mathcal{G}}_{\alpha}$ as the recommendation rule (similarly, $\hat{\mathcal{Z}}_\alpha$ is the guess of the agent for the set $\mathcal{Z}_{\alpha}(x,y)$).

Among the possible learning algorithms, we are interested in those that are $(\epsilon, \delta)$-correct w.r.t.~the Hausdorff distance. More precisely, the algorithm has to guarantee that:
\begin{align*}
    &\Prob( H(\mathcal{G}_{\alpha}(x,y), \hat{\mathcal{G}}_{\alpha}) \ge \epsilon) \le \delta
\end{align*}
in the general-sum case, with an analogous requirement for $\mathcal{Z}_{\alpha}$ in the zero-sum setting.

%

We make a remark that distinguishes the two learning objectives. At first, one might think that the problem of learning $\mathcal{Z}_\alpha(x,y)$ could be easily solved once we have a set $\hat{\mathcal{G}}_\alpha$ that approximates $\mathcal{G}_\alpha(x,y)$. Nonetheless, we did not find any trivial way to slice a generic set $\hat {\mathcal{G}}_\alpha$ so that it approximates $\mathcal{Z}_\alpha(x,y)$. On one hand, one might be tempted to obtain $\hat{\mathcal{Z}}_\alpha$ by slicing $\hat{\mathcal{G}}_\alpha$ on $A = -B$. However, this only ensures one direction of the Hausdorff distance.\footnote{Note that, in principle, the resulting set could even be empty.} Similarly, one could think of slicing $\hat{\mathcal{G}}_{\alpha}$ by taking all the pair of  matrices $(A,B)$ which are almost zero-sum, \eg $\|A - (-B) \|_{\infty} \le \epsilon$. However, this set might be too large and could fail to satisfy the other direction of the Hausdorff distance. For these reasons, we study zero-sum and general-sum games as separate problems, providing minimax rates for both these settings.

\section{Sample Complexity Lower Bound}\label{sec:lb}

\begin{figure*}[ht]
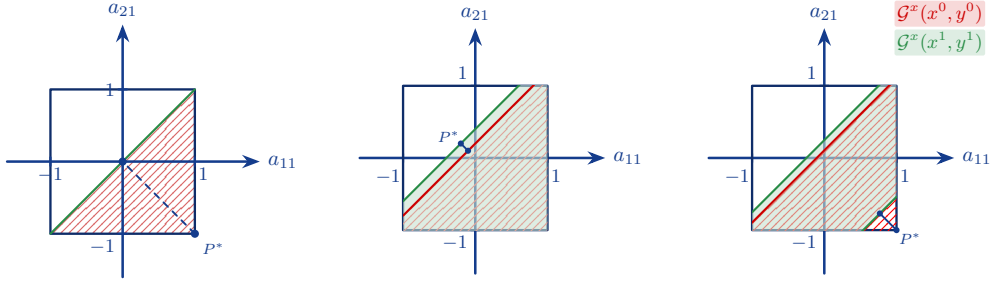

\begin{center}
\scalebox{0.8}{\centerline{\begin{tikzpicture}

\node at (-5.8,0) {\input{img/example1}};
\node at (0,0)  {\input{img/example2}};
\node at (5.8,0)  {\input{img/example3}};

\end{tikzpicture}}}
\caption{Visual representations of payoff sets $\mathcal{G}^x_\alpha(\cdot, \cdot)$ for several pairs of instances $(x^0, y^0)$ (striped \textbf{{\color{regionred} red}} region) and $(x^1, y^1)$ (solid \textbf{{\color{diaggreen} green}} region). In each figure, a point $P^*$ is used to illustrate the Hausdorff distance between the two areas. The $y$ vector is $(1,0)$ in all problems. (\emph{Left}): we set $\alpha=0$, $x^0=(0,1)$, $x^1=(\pi, 1-\pi)$ for any $\pi \in (0,1)$; (\emph{Center}): $\alpha = 0.2$, $x^0=(0,1)$, and $x^1=(\pi, 1-\pi)$ for $\pi = 0.05$; (\emph{Right}): $\alpha=0.1$, $\epsilon=0.1$, $x^0 = (\tfrac{\alpha}{2}, 1-\tfrac{\alpha}{2})$ and $x^1 =(\tfrac{\alpha(1+3\epsilon)}{2}, 1-\tfrac{\alpha(1+3\epsilon)}{2}) $.}
\label{fig:lb-areas}
\end{center}
\vskip -0.2in
\end{figure*}

In this section, we study the sample complexity of any $(\epsilon, \delta)$-correct algorithm that learns, with high probability, the set of payoffs induced by an unknown strategy profile $(x,y)$. In \Cref{sec:lb-exact} we discuss the case of $\alpha = 0$, and in \Cref{sec:lb-approx} we discuss the case when $\alpha>0$. 

\subsection{On Exact Equilibria}\label{sec:lb-exact}

We begin by highlighting a fundamental challenge in learning the feasible sets $\mathcal{G}(x,y)$ and $\mathcal{Z}(x,y)$ from data generated by an exact equilibrium.
To build intuition for the problem, we first consider a simple general-sum example. Let $n=2$ and consider two strategy profiles $(x^0, y^0)$ and $(x^1, y^1)$  such that $y^0 = y^1 = (1,0)$, $x^0=(0,1)$, and $x^1=(\pi, 1-\pi)$ for some $\pi$.\footnote{We use exactly this construction to prove \Cref{thm:lb-alpha-0-impossibility}.} Then, it is easy to see that a matrix $A$ belongs to $\mathcal{G}^x(x^0,y^0)$ if and only if $a_{21} \le a_{11}$, while $A$ belongs to $\mathcal{G}^x(x^1, y^1)$ if and only if $a_{21} = a_{11}$. Hence, as $\pi \to 0$, we have a pair of instances that are statistically indistinguishable (\ie the KL divergence between the two instances approaches $0$ as $\pi \to 0$) while the Hausdorff between the payoff sets is constant and much larger than $\epsilon$. This behavior is illustrated in \Cref{fig:lb-areas} (\emph{left}): for example, the payoff matrix with $a_{21}=-1$ and $a_{11}=1$  (marked with $P^\ast$ in the figure) belongs to $\mathcal{G}(x^0,y^0)$ but lies at a constant distance from every point on the line $a_{21}=a_{11}$. Consequently, when $x$ is unknown and must be learned from samples, one cannot distinguish between $(x^0,y^0)$ or $(x^1,y^1)$, and therefore cannot output an $\epsilon$-correct estimate of the feasible set. In general, we have the following result.
\begin{theorem}\label{thm:lb-alpha-0-impossibility}
    Let $\alpha = 0$ and $\epsilon < 1$. Then, there exists a problem instance $(x, y) \in \Delta_n^2$ such that, for any $\pi \in (0,1)$ and $(\epsilon, \delta)$-correct algorithm, it holds that:
    \begin{align}\label{eq:impossibility}
        \E_{x,y}[\tau_\delta] \ge  \Omega\left( \nicefrac{\log\left(\frac{1}{\delta}\right)}{\log\left( \frac{1}{1-\pi} \right)} \right).
    \end{align}
    This holds both for General and Zero-Sum Games.
\end{theorem}

These results show that there exists an instance in which learning $\mathcal{G}(x,y)$ and/or $\mathcal{Z}(x,y)$ is essentially impossible when $\alpha=0$. Indeed, taking the limit of \Cref{eq:impossibility} for $\pi \to 0$ leads to an infinite complexity lower bound. The cause of this phenomenon is intrinsic in the definition of an exact Nash equilibrium. Indeed, the condition for $A \in \mathcal{G}^x(x,y)$,\footnote{Similar reasoning holds for $\mathcal{G}^y(x,y)$ and $\mathcal{Z}(x,y)$ as well.} \ie $x^\top A y \le e_i^\top A y~\forall i \in [n]$, can equivalently be rewritten as follows:\footnote{This is a well-known result, \eg~\citet[Section~3.6]{gintis2000game}. For completeness, we report a formal statement in \Cref{lemma:nash-support}.}
\begin{subequations}\label{eq:nash-revr}
    \begin{align}
        & x^\top A y = e_i^\top A y \quad \forall i: x_i > 0 \\
        & x^\top A y \le e_i^\top A y \quad \forall i: x_i = 0.
    \end{align}
\end{subequations}
This explains the strong discontinuity of $\mathcal{G}^x(\cdot,y)$ when the strategy profile of the $x$-player changes support, as we observed in the example.

To mitigate this issue, a natural assumption is to restrict attention to problem instances in which each action is played with probability bounded away from zero. This is exactly what was done in the IRL problem \citep{metelli2023towards}, where the authors have found a similar phenomenon for the inverse reinforcement learning problem.  
Specifically, for any $x,y$, let $\pi_{\min}(x,y) \in (0,\tfrac{1}{n}]$ as:
\begin{align*}
    \pi_{\min}(x,y) = \min \left\{ \min_{i: x_i > 0} x_i, \min_{j: y_j > 0} y_j \right\}.
\end{align*}
Then, we assume that the agent is dealing with a set of instances such that $\pi_{\min}(x,y) \ge  \pi_{\min}$ for some $ \pi_{\min} > 0$.\footnote{We observe that in some settings assuming that all actions are played with positive probability could be too demanding. Another possibility, would be to relax the learning requirement, \eg to retrieve the set of matrices for which $(x,y)$ is an $\alpha$-Nash equilibrium with $\alpha > 0$. As we discuss later in this section, this problem is learnable. Furthermore, since $(x,y)$ is an exact equilibrium, the true (unknown) payoffs also belong to $\mathcal{G}_\alpha(x,y)$ or $\mathcal{Z}_\alpha(x,y)$. In this way, one can still guarantee to recover a set that contains an approximation of the true payoffs.} In this case, we derive the following lower bound.

\begin{theorem} \label{thm:lb-alpha-0}
Let $\alpha = 0$, ${\pi}_{\min} > 0$, and $\epsilon < 1/384$. Then, there exists a problem instance $(x,y)$  such that, for any $(\epsilon, \delta)$-correct algorithm, the following holds:
\begin{align}\label{eq:rate-alpha-0-pi-min}
\mathbb{E}_{x,y}[\tau_\delta] \ge \Omega\left( \frac{\log\left(\frac{1}{\delta}\right)}{\log\left( \frac{1}{1- \pi_{\min}} \right)} + \frac{n+\log\left(\frac{1}{\delta}\right)}{\epsilon^2} \right).    
\end{align}
This result holds both for General and Zero-Sum Games.
\end{theorem}

We can interpret the first term as the worst-case complexity of learning the support of $(x,y)$. The second term, instead, arises from the fact that, even after the support is identified, the strategies $(x,y)$ must still be estimated accurately in order to return an $(\epsilon,\delta)$-correct set in the Hausdorff metric. 


\subsection{On Approximate Equilibria}\label{sec:lb-approx}

In this section we provide a lower bound for the case in which players are playing an $\alpha$-approximate Nash equilibrium. A natural question is whether the difficulty of learning the support of the equilibrium strategies again leads to an infinite sample-complexity bound. We argue that, in this setting, the answer is \emph{no}. The presence of $\alpha$ in the definitions of $\mathcal{G}_\alpha(x,y)$ and $\mathcal{Z}_\alpha(x,y)$ prevents these sets from having a representation analogous to \Cref{eq:nash-revr}. Intuitively, this removes the strong discontinuity exhibited by $\mathcal{G}(x,y)$ and $\mathcal{Z}(x,y)$ when $(x,y)$ vary across strategy profiles with different supports. Informally, the parameter $\alpha$ acts as a tolerance that absorbs the error contributed by actions that are played with small probability.
To have an intuitive picture of this phenomenon, it is instructive to consider the simple instance that we analyzed for the $\alpha = 0$ setting. 
We recall that $y=y^0=y^1 = (1,0)$, $x^0=(0,1)$, and $x^1 = (\pi, 1-\pi)$ for some $\pi$. Let $\Delta_A = a_{21} - a_{11}$. Then, we have that
\begin{align*}
& A \in \mathcal{G}^x_\alpha(x^0, y^0) \Leftrightarrow -2 \le \Delta_A \le \alpha; \\
& A \in \mathcal{G}^x_\alpha(x^1, y^1) \Leftrightarrow \Delta_A \in [\max\{-2,-\tfrac{\alpha}{\pi} \} , \min\{2,\tfrac{\alpha}{1-\pi}]\}.
\end{align*}
As $\pi \to 0$ (or, equivalently, as $\pi \ll \alpha$), the two sets collapse. Thus, $\alpha$ is absorbing the error of actions that are played rarely. These sets are depicted in \Cref{fig:lb-areas} (\emph{center}).


While this discussion is reassuring, an important question remains: what is the \emph{worst-case} dependence of the lower bound on the parameter $\alpha$? One can anticipate that, as $\alpha \to 0$, the sample complexity must diverge, since actions played with vanishing probability increasingly influence the geometry of the feasible sets. To address this question, we establish the following result.

\begin{theorem}
\label{thm:lb-any-alpha}
Let $\alpha \in (0, 1/4)$ and $\epsilon \le 1/128$. Then, there exists a problem instance $(x,y)$ such that, for any $(\epsilon, \delta)$-correct algorithm, the following holds:
\begin{align}\label{eq:lb-any-alpha-main-text}
    \mathbb{E}_{x,y}[\tau_\delta] \geq \Omega\!\left(
\frac{n+\log(\frac{1}{\delta})}{\epsilon^{2}\,\alpha}
\right).
\end{align}
This result holds both for General and Zero-Sum Games.
\end{theorem}

An important aspect of this result is that the factor $\alpha^{-1}$ multiplies \emph{all} the other terms in the bound. To understand why this happens, we provide some intuition on how we obtained the term $\log(\delta^{-1}) / \epsilon^2 \alpha$.\footnote{The term $n/(\epsilon^2\alpha)$ follows from a related but more involved argument, which we defer to Appendix~\ref{app:lb-any-alpha-gsg}.} Consider two instances with $n=2$ such that $y^1=y^0= (1,0)$, $x^0=(\tfrac{\alpha}{2}, 1-\tfrac{\alpha}{2})$, and $x^1=(\tfrac{\alpha + 3\alpha\epsilon}{2}, 1-\tfrac{\alpha+3\alpha\epsilon}{2})$. Again, let $\Delta_A = a_{21} - a_{11}$. One can verify that, for $\alpha$ and $\epsilon$ sufficiently small, it holds
\begin{align*}
& A \in \mathcal{G}^x_\alpha(x^0, y^0) \implies \Delta_A \ge -2  ;\\
& A \in \mathcal{G}^x_\alpha(x^1, y^1) \implies \Delta_A \ge -\tfrac{2\alpha}{\alpha + 3\alpha\epsilon} > -2.
\end{align*}
These sets are depicted in \Cref{fig:lb-areas} (\emph{right}). With some algebraic manipulations, one can show that the distance between these sets is at least $2\epsilon$ and, therefore, it is important for any algorithm to distinguish $(x^0, y^0)$ and $(x^1, y^1)$ to be $(\epsilon,\delta)$-correct.
However, for $\epsilon$ and $\alpha$ small, the two instances are very similar as the KL divergence between the two is at most $\epsilon^2 \alpha$.
Therefore, using change-of-measure arguments, we can obtain the term $\log(\delta^{-1}) / \epsilon^2 \alpha$. This shows that, while $\alpha$ provides a form of protection against the adverse geometry encountered when learning from exact equilibria, its effect is limited since all the other parameter the defines the instance are badly impacted by a small $\alpha$.

\section{Algorithm}\label{sec:algo}

\begin{algorithm}[!t]
	\caption{}
	\label{alg:WAS}
	\begin{algorithmic}[1]
        \REQUIRE{Number of samples $m \in \mathbb{N}$, $\alpha, \epsilon, \delta \in \mathbb{R}_{\ge 0}$}
		\WHILE{$t \in \{1, \dots, m \}$}
        \STATE{Observe $X_t, Y_t \sim x,y$ and update $\hat{x}(t), \hat{y}(t)$}
        \ENDWHILE
        \STATE{Return $\mathcal{G}_{\alpha}(\hat{x}(m), \hat{y}(m))$ or $\mathcal{Z}_{\alpha}(\hat{x}(m), \hat y(m))$}
	\end{algorithmic}
\end{algorithm}

The algorithm we developed and analyzed to attain minimax optimality is simple, and its pseudocode can be found in Procedure \ref{alg:WAS}. It takes as input the parameter of the problem, \ie $\alpha, \epsilon$ and $\delta$, and an additional parameter $m \in \mathbb{N}$, that denotes the number of samples to collect. Then it simply gathers $m$ observations for the two players $x$ and $y$, and it stops by returning all the matrices that are an ($\alpha$-)Nash equilibrium for the empirical strategy profiles. In other words, the recommendation rule is $\mathcal{G}_\alpha(\hat x(m), \hat y(m))$ for general-sum games and $\mathcal{Z}_\alpha(\hat x(m), \hat y(m))$ for zero sum games, where $\hat x(m)$ and $\hat y(m)$ are the maximum likelihood estimators of $x$ and $y$ after $m$ samples (\ie empirical means). In the following, we drop the dependency on $m$ in $\hat x(m)$ and $\hat y(m)$, and we simply write $\hat x$ and $\hat y$.

In the rest of this section, we show that, by properly choosing the parameter $m$, Procedure \ref{alg:WAS} is $(\epsilon, \delta)$-correct and, furthermore, this choice of $m$ attains the same rate that we presented in the lower bound section. Specifically, in \Cref{sec:alg-exact} we discuss our results for the exact equilibrium case and in \Cref{sec:alg-approx} for the approximate one. 

\paragraph{Representation of $\mathcal{G}_\alpha(\hat x, \hat y)$ and $\mathcal{Z}_\alpha(\hat x, \hat y)$}
In the introductory section, we noted that the recommendation rule is in one-to-one correspondence with a set of matrices. Here, we note that the sets produced by our procedure can be represented as linear feasibility problems. Indeed, given $\hat x$ and $\hat y$, the set $\mathcal{G}_\alpha(\hat x, \hat y)$ and $\mathcal{Z}_\alpha(\hat x, \hat y)$ can be described by collections of linear inequalities.

\subsection{Optimal Rates for Exact Equilibria}\label{sec:alg-exact}
As discussed in \Cref{sec:lb-exact}, when $\alpha=0$ the problem of learning the payoff sets is in general ill-posed, as the sample complexity may be infinite. Therefore, we adopt the same assumption of \Cref{thm:lb-alpha-0} and restrict our attention to instances in which  $\pi_{\min}(x,y) \ge \pi_{\min} > 0$. In this setting, we can show the following.

\begin{theorem}\label{thm:ub-alpha-0}
    Consider $\alpha = 0$ and let $$m \in \widetilde{{\mathcal{O}}} \left( \frac{\log\left( \frac{n}{\delta} \right)}{\log\left( \frac{1}{1-\pi_{\min}} \right)} + \frac{n+\log\left(\frac{1}{\delta}\right)}{\epsilon^2} \right).$$ Then, Procedure \ref{alg:WAS} is $(\epsilon, \delta)$-correct and its sample complexity $\tau_\delta$ is given by $m$, both for General and Zero-Sum Games.\footnote{Here, $\widetilde{\mathcal{O}}$ hides multiplicative constants and logarithmic factors.}
\end{theorem}

\Cref{thm:ub-alpha-0} shows the same rate of the lower bound of \Cref{thm:lb-alpha-0}, thus showing minimax optimality. As discussed in the proof sketch below, the first term can be interpreted as the number of samples required to learn the support of $(x,y)$, while the second one, is the number of samples needed for the Hausdorff distance between $\mathcal{G}(x,y)$ and $\mathcal{G}(\hat x, \hat y)$ to be bounded by $\epsilon$, once the support has been identified. 

We now present a proof sketch of \Cref{thm:ub-alpha-0} for the general-sum game problem. Then, we will discuss how to extend this reasoning to the zero-sum setting.

\paragraph{Proof sketch of \Cref{thm:ub-alpha-0} (General-Sum Games)} 
The proof works as follows. We fix a generic $m \in \mathbb{N}$, and we derive a high-probability upper bound on the Hausdorff distance between $\mathcal{G}(x,y)$ and $\mathcal{G}(\hat x, \hat y)$. Then, $m$ is chosen so that this upper bound is below $\epsilon$. The line that we follow to derive this upper bound is inspired by the lower bound of \Cref{sec:lb}. In particular, we know that the support of $(x,y)$ and $(\hat x, \hat y)$ should match in order to avoid the degenerate behavior that we identified. Hence, we start with the assumption that $(x,y)$ and $(\hat x, \hat y)$ are such that $\supp(x)=\supp(\hat x)$ and $\supp(y)=\supp(\hat y)$, and we later discuss how to guarantee that this holds in high probability once $m$ is large enough. Under the assumption that $\supp(x)=\supp(\hat x)$ and $\supp(y)=\supp(\hat y)$, we upper bound $H(\mathcal{G}(x,y), \mathcal{G}(\hat x, \hat y))$ by showing that for any $(A,B) \in \mathcal{G}(x,y)$ there exists $(\hat A, \hat B) \in \mathcal{G}(\hat x, \hat y)$ such that:\footnote{We also show that the symmetric claim is true, since we need both directions to upper bound the Hausdorff distance.}
\begin{align}\label{eq:constr-main}
    \|A-\hat A\|_{\infty} \le 4 \| y - \hat y\|_1,\, \|B-\hat B \|_{\infty} \le 4\|x - \hat x \|_1.
\end{align}
Note that this directly gives us an upper bound on the Hausdorff distance of $ 4(\|x - \hat x \|_1+ \|y - \hat y \|_1)$. In other words, the error scales linearly with the $\ell_1$-distance between the true strategy profile and the empirical one. 

Now, we discuss how to obtain \Cref{eq:constr-main} (we do that for $A$, and the construction for $B$ follows from a similar reasoning).
Starting from $A \in \mathcal{G}^x(x,y)$ (and assuming equal support), we know that if $\hat A \in \mathcal{G}^x(\hat x, \hat y)$, then for all $i \in \supp(x)$, $\hat x^\top \hat{A} \hat y = e_i^\top \hat{A} \hat y$. Then, since we want to construct $\hat A$ close to $A$, our goal is to "redistribute" the payoffs of $A$ according to the shift of $\hat x$ and $\hat y$ so that, in all the rows in $\supp(\hat x)$, the value obtained by the x-player is left unchanged. To this end, for any $i \in \supp(\hat x)$, we first define the following quantity:
\begin{align}\label{eq:delta-main}\textstyle
    \Delta_i = \max_{k \in \supp(\hat x)} \sum_{j \in [n]} \hat y_j A_{kj} - \sum_{j \in [n]} \hat y_j A_{ij}.
\end{align}
Here, $\Delta_i \in [0, 2\|y-\hat y \|_1]$ represents the maximum payoff difference (for the $x$-player) in playing the pure strategy $i$ w.r.t.~playing any other strategy $k \in \supp(\hat x)$. Intuitively, this is the shift, for actions in $\supp(\hat x)$, that $\hat A$ should meet so that $\hat A  \in \mathcal{G}^x(\hat x, \hat y)$.
Instead, for all the rows that do not belong to $\supp(\hat x)$, we only need to ensure that their value remains greater than those of actions in $\supp(\hat x)$.
Given these intuitions, we define $\hat A$ as follows: 
\begin{equation}\label{eq:mat-constr-a-gsg-game-main}
\hat A_{ij} :=
\begin{cases}
\dfrac{A_{ij}}{1+2\|y - \hat y\|_1}, & \text{if } i \in S^\star_{\hat x},\\[12pt]
\dfrac{A_{ij} + \Delta_i^x}{1+2\|y - \hat y\|_1}, & \text{if } i \in \supp(\hat x) \setminus S^\star_{\hat x},\\[12pt]
\dfrac{A_{ij} + 2\|y - \hat y\|_1}{1+2\|y - \hat y\|_1}, & \text{if } i \notin \supp(\hat x),
\end{cases}
\end{equation}
where $S^{\star}_{\hat x} = \argmax_{k \in \supp(\hat x)} \sum_{j \in [n]} \hat y_j A_{kj}$. Ignoring the denominator, which is simply a rescaling to guarantee that $\hat A$ remains in $[-1,1]^{n \times n}$, we observe that the construction does not modify values in $S^\star_{\hat x}$, but it increases the values of other actions played in $\hat x$ so that any action in $\supp (\hat x)$ yields the same value. 
Moreover, this construction increases the value of all the actions which are not played by $\hat x$, so that they do not obtain values which are better than actions that are actually played by $\hat x$. It is then possible to verify that $\hat A$ belongs to $\mathcal{G}^x(\hat x, \hat y)$, and that $\|A- \hat A \|_{\infty} \le 4 \|y-\hat y\|_1$.

Finally, to conclude the proof, it is sufficient to upper bound with high-probability $\|x - \hat x \|_1$ and $\|y - \hat y \|_1$, and to guarantee that the support of $(x,y)$ aligns with that of $(\hat x, \hat y)$. Standard concentration tools yield $\mathcal{O}(\sqrt{(n+\log(\delta^{-1})/m})$ for the former and $\widetilde{\mathcal{O}}(\log(\delta^{-1})/\log((1-\pi_{\min})^{-1}))$ for the latter. Putting everything together this gives \Cref{thm:ub-alpha-0}.

\paragraph{What changes in Zero-Sum Games} The proof for the zero-sum game problem follows from a similar argument; however, there is an important difference in how to upper bound the Hausdorff distance once the support of $(x,y)$ and $(\hat x, \hat y)$ matches. Specifically, constructing $\hat A \in \mathcal{Z}(\hat x, \hat y)$ is more challenging as now $\hat A$ should satisfy the constraints for both players simultaneously. In particular, we show in \Cref{app:zsg-alpha0} that a matrix $\hat A$ constructed as in \Cref{eq:mat-constr-a-gsg-game-main} fails to meet the constraints for the $y$-player. To solve this issue (\ie building a matrix that accounts at the same time for changes in $x$ and $y$), we show that we can upper bound $H(\mathcal{Z}(x, y), \mathcal{Z}(\hat x, \hat y))$ as:
\begin{align*}
    H(\mathcal{Z}(x, y), \mathcal{Z}(\hat x, y)) + H(\mathcal{Z}(\hat x, y), \mathcal{Z}(\hat x, \hat y)).
\end{align*}
This error decomposition allows us to apply the technique for general-sum games even in the zero-sum setting. 

\subsection{Optimal Rates for Approximate Equilibria}\label{sec:alg-approx}

Next, we present the results for the case in which the two players are playing an approximate Nash equilibrium. The following result summarizes our findings. 

\begin{theorem}\label{thm:ub-any-alpha}
    Consider $\alpha > 0$ and let 
    \begin{align*}
        m \in {\mathcal{O}} \left( \frac{n+\log\left( \frac{1}{\delta} \right)}{\epsilon^2 \alpha } + \frac{\sqrt{n} \log\left( \frac{1}{\delta} \right)}{ \epsilon \alpha} \right).
    \end{align*}
    Then Procedure \ref{alg:WAS} is $(\epsilon, \delta)$-correct and its sample complexity $\tau_\delta$ is given by $m$,
    both for General and Zero-Sum Games.
\end{theorem}

For $\epsilon$ sufficiently small, \ie $\epsilon \lesssim n^{-1/2}$, \Cref{thm:ub-any-alpha} matches the rate of \Cref{thm:lb-any-alpha}. Hence, Procedure \ref{alg:WAS} is minimax optimal in the relevant regime of small $\epsilon$.


\paragraph{Proof sketch of \Cref{thm:ub-any-alpha} (General-Sum Games)} As for \Cref{thm:ub-alpha-0}, we fix a generic $m \in \mathbb{N}$ and derive a high probability upper bound on the Hausdorff distance between the two sets as a function of $m$. Then, we choose $m$ so that this upper bound is below $\epsilon$. We start by employing the same decomposition that we used for zero-sum games and $\alpha = 0$ in order to upper bound $H(\mathcal{G}_\alpha(x, y), \mathcal{G}_\alpha(\hat x, \hat y))$ with 
\begin{align*}
    H(\mathcal{G}_\alpha(x, y), \mathcal{G}_\alpha(\hat x, y)) + H(\mathcal{G}_\alpha(\hat x, y), \mathcal{G}_\alpha(\hat x, \hat y)).
\end{align*}
This allows us to study only the error components that arise from one of the two strategy profiles, $x$ or $y$. For clarity of exposition, we explain how to upper bound $H(\mathcal{G}_\alpha(x,y),\mathcal{G}_\alpha(\hat x,y))$, and in particular we focus on the component corresponding to the first player, namely the sets $\mathcal{G}_\alpha^x(\cdot,\cdot)$ (as for $\mathcal{G}_\alpha^y(\cdot, \cdot)$ we have specular arguments). 
%

Thus, we start by considering any $A \in \mathcal{G}^x_\alpha(x,y)$ and construct $\hat A \in \mathcal{G}^x_\alpha(\hat x, y)$ which is close to $A$. 
To this end, we follow a reasoning that is similar to the above, \ie we first quantify $A$'s constraint violation of $\mathcal{G}_\alpha^x(\hat x, y)$ and we try to correct that error with an appropriate construction. Specifically, let us denote by $g_A$ the following quantity 
\begin{align*}
    g_A \coloneqq \max_{i \in [n]} \left( \hat x^\top A y - e_i^\top A y\right).
\end{align*}
If $g_A \le \alpha$, then $A \in \mathcal{G}^x_\alpha(\hat x, y)$ as $\alpha$ is absorbing the error that arises from $\hat x \ne x$. If $g_A \ge \alpha$, instead, we construct $\hat A$ as $\hat A = \frac{\alpha}{g_A} A$. This multiplicative construction ensures that $\hat A \in \mathcal{G}_\alpha(\hat x, y)$ since, for all $i \in [n]$, we have that:
\begin{align*}
    \hat x^\top \hat A y - e_i^\top \hat A y = \frac{\alpha}{g_A}(\hat x^\top Ay - e_i^\top A y) \le \alpha. 
\end{align*}
The difference between $A$ and $\hat A$ is then given by:
\begin{align*}
\| A - \hat A \|_{\infty} =  \mathcal{O}\left( \frac{g_A - \alpha}{g_A}\right)
\le \mathcal{O}\left(\frac{(\hat x - x)^\top A y}{\alpha} \right),
\end{align*}
where the second step follows from $g_A \ge \alpha$ and the fact that $A \in \mathcal{G}_\alpha^x(x,y)$.\footnote{We add and subtract $x^\top A y$. Then we use $x^\top A y - e_i^\top A y \le \alpha$.} 
Hence, to control the Hausdorff distance related to $\mathcal{G}_\alpha^x(x, y)$ and $\mathcal{G}_\alpha^x(\hat x, y)$ , we need to control:
\begin{align}\label{eq:ub-any-alpha-eq1-main}
    \frac{1}{\alpha}\sup_{A \in \mathcal{G}_\alpha^x(x,y)} {(\hat x - x)^\top A y}{}.
\end{align}
At first, one might be tempted to bound $(\hat x- x)^\top A y$ with $2 \|x-\hat x \|_1$. However, this naive argument yields sub-optimal rates. Since with high probability $\|x-\hat x \|_1 \le \sqrt{(n+\log(\delta^{-1})/m}$, the number of samples to ensure that \Cref{eq:ub-any-alpha-eq1-main} is below $\epsilon$ would roughly be $(n+\log(\delta^{-1}))/{\epsilon^2 \alpha^2}$, which does not match the lower bound of \Cref{thm:lb-any-alpha}.
The difficulty is that the bound $(\hat x - x)^\top A y \le 2\|x-\hat x\|_1$ fails to exploit the fact that $A\in\mathcal{G}_\alpha^x(x,y)$ and does not leverage its geometric structure.

In the following, we show how to obtain better rates by exploiting the structure of $\mathcal{G}_\alpha^x(x,y)$. Indeed, the properties of $\mathcal{G}_\alpha^x(x,y)$ will ensure that the variance of $(\hat x - x)^\top A y$ is small, thus leading to faster high-probability concentration. To this end, we need to manipulate appropriately \Cref{eq:ub-any-alpha-eq1-main}. First of all, we bound \Cref{eq:ub-any-alpha-eq1-main} as follows
\begin{align}\label{eq:ub-any-alpha-eq4-main}
    \sup_{A \in \mathcal{G}_\alpha^x(x,y)}\frac{\sum_{i: x_i > 0} (\hat x_i - x_i) (Ay)_i}{\alpha} + \frac{2\sum_{i: x_i = 0} \hat x_i}{\alpha},
\end{align}
to get rid of the components where $x_i = 0$ in the optimization problem. Let us ignore the second term for now, and focus on the remaining optimization over $A$ (we return to the second term later in the proof sketch).
We will show that this remaining optimization problem presents a hidden fractional knapsack structure.  First, let us introduce some notation. For any $A \in \mathcal{G}_\alpha^x(x,y)$, let $\star \in [n]$ be any row that belongs to $\argmin_{i \in [n]} e_i^\top A y$. Then, let $c_i = x_i [(Ay)_i - (Ay)_\star]$. Then, we can rewrite that optimization problem as:
\begin{align}\label{eq:ub-any-alpha-eq2-main}
    \frac{1}{\alpha}\sup_{c: c_i \in [0, 2x_i]} \sum_{i: x_i > 0} \frac{\hat x_i - x_i}{x_i} c_i \quad \text{s.t.~} \sum_{i \in [n]} c_i \le \alpha.
\end{align}
This problem is well-known to be solved by a greedy algorithm that sorts the indexes $i$'s according to their value (\ie $\frac{\hat x_i -x_i}{x_i}$) and allocates the maximum value of $c_i$ (\ie $c_i = 2x_i$), according to this order until the budget is fully depleted, \ie is stops once $\sum_{i: x_i > 0} c_i >  \alpha$. \Cref{eq:ub-any-alpha-eq2-main} can be rewritten as (exact steps in \Cref{app:anyalpha-gsg}):
\begin{align*}
    \frac{1}{\alpha}\sup_{S \in \mathcal{S}_\alpha(x)} \sum_{i \in S} (\hat x_i - x_i),
\end{align*}
where $\mathcal{S}_\alpha(x)$ is given by
\begin{align*}\textstyle
\mathcal{S}_\alpha(x) \coloneqq \{S \subseteq [n]: x_i > 0~\forall i \in S \land \sum_{i \in S} x_i \le \alpha/2 \}.    
\end{align*}
Now, the main point is that each set $S \in \mathcal{S}_\alpha(x)$ is such that $\sum_{i \in {S}} x_i \le \alpha /2$. Hence, we can treat $\sum_{i \in S}(\hat x_i - x_i)$ as a Bernoulli random variable with mean bounded by $\alpha /2$ (and variance bounded by $\alpha/2$). Using Bernstein's inequality, we have that $\sum_{i \in S}(\hat x_i - x_i)$ concentrates as $\mathcal{O}(\sqrt{\alpha \log(\delta^{-1})/m})$. Taking an union bound over all the possible subsets $S \in \mathcal{S}_\alpha(x)$, we obtain a rate of $\mathcal{O}(\sqrt{\alpha(n+\log(\delta^{-1}) / m)}$. Plugging this result in \Cref{eq:ub-any-alpha-eq4-main}, we have that only $\mathcal{O}(\frac{n+\log(\delta^{-1})}{\epsilon^2 \alpha})$ samples are needed to have the first error term below $\epsilon$.

Now, we still need to control the second term of \Cref{eq:ub-any-alpha-eq4-main}. It is easy to see that this term is $0$. Indeed, given that $\hat x_i$ is the empirical mean of $x_i$, $x_i = 0 \implies \hat x_i = 0$. However, this only holds for one direction of the argument, \ie when we are finding matrices $\hat A \in \mathcal{G}^x_\alpha(\hat x, y)$ which are close to $A \in \mathcal{G}_\alpha^x(x,y)$. In the opposite case, we obtain $2\sum_{i: \hat x_i = 0} x_i$, which is commonly referred to as \emph{missing mass} \citep{mcallester2003concentration} and is not equal to $0$.\footnote{There is another asymmetry in bounding that direction of the Hausdorff distance. Specifically, when constructing $A \in \mathcal{G}_\alpha(x,y)$ close to some $\hat A \in \mathcal{G}_\alpha(\hat x, y)$, the set $\mathcal{S}_\alpha(\hat x)$ is stochastic. In the proof, we show that we can reduce it to $\mathcal{S}_{2\alpha}(x)$ by only paying a minor order term in the sample complexity.} However, \citet{rajaraman2020toward} has shown that this term concentrates around its expectation fast w.r.t.~to the number of samples, \ie it decays as $\frac{\sqrt{n} \log(\delta^{-1})}{m}$ and hence only $\sqrt{n} \log(\delta^{-1})/\epsilon\alpha$ samples are needed to have that term below $\epsilon$.\footnote{Furthermore, the expectation decays also as $\mathcal{O}({m}^{-1})$).} In this way, we concluded the proof of \Cref{thm:ub-any-alpha}.


\paragraph{What changes in Zero-Sum Games}
For the zero-sum game problem, we follow an analogous argument. 
The main difference is that we define $g_A$ as:
\begin{align*}
    g_A \coloneqq \max \left\{ \max_{i \in [n]} \hat x^\top A y - e_i^\top A y, \max_{j \in [n]} \hat x^\top A e_j - \hat x^\top A y \right\}.
\end{align*}
This is due to the fact that the matrix $\hat A$ now needs to satisfy a larger set of constraints. We show that the errors introduced by this modification are bounded as in the general-sum case.

\paragraph{The challenges for $\alpha > 0$} 
We conclude by highlighting the challenges introduced by the case $\alpha>0$. The main difficulty lies in the fact that when the players are playing an approximate equilibrium, we cannot rely on strong characterizations of their behavior, \ie \Cref{eq:nash-revr}, as we did for the $\alpha = 0$ case. Furthermore, \Cref{thm:lb-any-alpha} shows that the dependence in $\alpha$ is complex, as it multiplies all the parameters of the problem due to geometric arguments that impact how much the Hausdorff distance changes when the strategy profiles have values around $\alpha$.

\section{Conclusions}

In this paper, we study the problem of learning from demonstrations the payoff sets that are compatible with an unknown, possibly approximate, equilibrium. We derived tight minimax rates by providing matching lower and upper bounds. As we highlighted, these results required careful considerations on the geometry of the set of feasible payoffs.  
Our work paves the way for several promising directions for future research. For instance, it would be interesting to extend these results beyond the notion of approximate Nash equilibrium to broader solution concepts, such as correlated equilibria. Another line of possible research, instead, could aim to extend these results beyond bimatrix games with a finite number of actions. Indeed, in several games the set of actions for each player is an infinite set. Furthermore, in this work, we assumed that the strategy profiles $(x,y)$ are fixed. Future research should focus on the case where strategy profiles are evolving, and players are learning the game through regret minimization. Finally, we note that our analysis is instance-independent. In the future, it would be exciting to develop instance-dependent lower and upper bounds that directly depend on the unknown strategy profiles $(x,y)$. Here, one could take inspiration from the optimal instance-dependent results available in the bandit literature \citep[\eg][]{garivier2016optimal,poiani2026pure}. In particular, in \citet{poiani2026pure}, the authors deals with optimal instance-dependent rates for problems that admits an infinite set of (euclidean) answers. In the future, it would be interesting to extend these results to the inverse learning problem that we presented in this work.

\section*{Acknowledgments}
This work was partially funded by the European Union. Views and opinions expressed are however those of the author(s) only and do not necessarily reflect those of the European Union or the European Research Council Executive Agency. Neither the European Union nor the granting authority can be held responsible for them.

This work is supported by an ERC grant (Project 101165466 — PLA-STEER).

\section*{Impact Statement}
This paper presents work aimed at advancing the field of machine learning. There are many potential societal consequences of our work, none of which we feel must be specifically highlighted here.

\bibliographystyle{plainnat}
\bibliography{bibliography}

\newpage
\appendix
\onecolumn

\etocdepthtag.toc{mtappendix}
\etocsettagdepth{mtchapter}{none}
\etocsettagdepth{mtappendix}{subsection}

\begin{spacing}{1}
\tableofcontents
\end{spacing}

\newpage

\begin{table}[t]
\centering
\caption{Table of Symbols}
\label{tab:symbols}
\begin{tabular}{ll}
\toprule
\textbf{Symbol} & \textbf{Description} \\
\midrule
$n$ & Number of actions per player \\
$[n]$ & Set $\{1,\dots,n\}$ \\
$\Delta_n$ & Probability simplex over $[n]$ \\
$e_i$ & $i$-th canonical basis vector \\
$\supp(x)$ & Support of $x \in \Delta_n$ \\
\midrule
$A,B$ & Payoff matrices of a general-sum bimatrix game \\
$x,y$ & Mixed strategies of the two players \\
$\alpha$ & Approximation parameter for $\alpha$-Nash equilibrium \\
\midrule
$\mathcal{G}_\alpha(x,y)$ & Set of payoff matrices $(A,B)$ for which $(x,y)$ is an $\alpha$-Nash equilibrium \\
$\mathcal{G}^x_\alpha(x,y)$ & Player-1 feasible payoff matrices \\
$\mathcal{G}^y_\alpha(x,y)$ & Player-2 feasible payoff matrices \\
$\mathcal{Z}_\alpha(x,y)$ & Zero-sum payoff matrices for which $(x,y)$ is an $\alpha$-Nash equilibrium \\
$\hat{\mathcal{G}}_\alpha, \hat{\mathcal{Z}}_\alpha$ & Learner's estimate of the feasible payoff set \\
\midrule
$(X_t,Y_t)$ & Actions observed at round $t$, sampled from $x \times y$ \\
$\hat x, \hat y$ & Empirical estimates of strategies \\
$m$ & Number of observed samples \\
$\tau_\delta$ & Stopping time of the learning algorithm \\
\midrule
$H(\cdot,\cdot)$ & Hausdorff distance (w.r.t. $\ell_\infty$ norm) \\
$\|\cdot\|_1, \|\cdot\|_\infty$ & $\ell_1$ and $\ell_\infty$ norms \\
$\epsilon$ & Target Hausdorff accuracy \\
$\delta$ & Failure probability \\
$\pi_{\min}(x,y)$ & Minimum positive probability in supports of $x$ and $y$ \\
$\pi_{\min}$ & Known lower bound on $\pi_{\min}(x,y)$ \\
\bottomrule
\end{tabular}
\end{table}

\section{Proof of the lower bounds}\label{app:lower}

\subsection{Proof of \Cref{thm:lb-alpha-0-impossibility}}\label{app:lb-alpha0-impossibility}

\begin{proof}
    The proof is split into several steps. 
    
    \paragraph{Step 1: Instance Construction}
    To prove the result, we consider the case where $n=2$. Then, we
    consider two different instances where the unknown strategy profiles $(x^0, y^0)$ and $(x^1, y^1)$ are defined as follows:
    \begin{align*}
    & x^0=(0,1), y^0=(1,0) \\
    & x^1=(\pi, 1-\pi), y^1=(1,0)   
    \end{align*}
    where $\pi$ is any number in $(0,1)$.
 
    \paragraph{Step 2: Haussdorff Distance}
    Next, we analyze the distances $H(\mathcal{G}(x^0, y^0), \mathcal{G}(x^1, y^1))$ and $H(\mathcal{Z}(x^0, y^0), \mathcal{Z}(x^1, y^1))$. In particular, we will provide lower bounds on these quantities.

    We start from $H(\mathcal{G}(x^0, y^0), \mathcal{G}(x^1, y^1))$. First of all, we note that:
    \begin{align}\label{eq:lb-alpha0-gsg-eq1}
        H(\mathcal{G}(x^0, y^0), \mathcal{G}(x^1, y^1)) \ge \max_{A \in \mathcal{G}^x(x^0, y^0)} \min_{A' \in \mathcal{G}^x(x^1, y^1)} \| A- A' \|_{\infty}.
    \end{align}
    Next, we analyze this lower bound in our simplified case. 
    For any matrix \(A=\begin{bmatrix}a_{11}&a_{12}\\ a_{21}&a_{22}\end{bmatrix} \in [-1,1]^{2\times2}\), using the definition of $\mathcal{G}^x(x,y)$, we have that:
    \begin{align*}
        & \text{(i) }  A \in \mathcal{G}^x(x^0, y^0) \iff a_{21} \le a_{11} \text{ and }  a_{21} \le a_{21} \iff a_{21} \le a_{11}; \\
        & \text{(ii) }A \in \mathcal{G}^x(x^1, y^1) \iff \pi a_{11} + (1-\pi)a_{21} \le \min\{a_{11}, a_{21} \} \iff a_{21} \le a_{11} \text{ and } a_{11} \le a_{21} \iff a_{21} = a_{11} .
    \end{align*}
    

    We now use this result to further lower bound the r.h.s. of \Cref{eq:lb-alpha0-gsg-eq1}. In particular, we have that:
    \begin{align*}
        \max_{A \in \mathcal{G}^x(x^0, y^0)} \min_{A' \in \mathcal{G}^x(x^1, y^1)} \| A- A' \|_{\infty}& =\max_{a_{11}, a_{21} \in [-1,1]^2: a_{11} \ge a_{21}} \min_{t \in [-1,1]} \max \{ |a_{11} - t|, |a_{21}-t| \} \\
        & \ge \min_{t \in [-1,1]} \max \{ |1-t|, |1+t| \} \\
        & = 1,
    \end{align*}
    where the first equality is justified by points (i) and (ii) above, together with the fact that for any $A \in \mathcal{G}(x^0,y^0)$ we can pick $A' \in \mathcal{G}(x^1, y^1)$ such that $a_{22} = a'_{22}$ and $ a_{12} = a'_{12}$, and the inequality follows by setting $a_{11} = 1$ and $a_{21} = -1$. 

    Similarly we can analyze $H(\mathcal{Z}(x^0,y^0), \mathcal{Z}(x^1, y^1))$.    
    In this case, we have:
    \begin{align*}
        & \text{(i) } A \in \mathcal{Z}(x^0, y^0) \iff a_{22} \le a_{21} \le a_{11}; \\
        & \text{(ii) } A \in \mathcal{Z}(x^1, y^1) \iff \pi a_{21} + (1-\pi)a_{22}  \le \pi a_{11} + (1-\pi)a_{21} \le \min\{a_{11}, a_{21} \}.
    \end{align*}
    Hence, we have that:
    \begin{align*}
        H(\mathcal{Z}(x^0, y^0), \mathcal{Z}(x^1, y^1)) & \ge \max_{A \in \mathcal{Z}(x^0, y^0)} \min_{A' \in \mathcal{Z}(x^1, y^1)} \| A- A' \|_{\infty} \\
        & \ge  \min_{A' \in \mathcal{Z}(x^1, y^1)} \max \{ |1-a'_{11}|, |-1-a'_{21}|  \} \tag{Pick $A \in \mathcal{Z}(x^0,y^0): a_{11}=1, a_{21}=-1$ }\\
        & \ge \min_{t \in [-1,1]} \max \{ |1-t|, |1+t| \} \tag{$A' \in \mathcal{Z}(x^1, y^1) \implies a'_{11} = a'_{22}$}\\
        & = 1.
    \end{align*}

    Therefore, the following holds:
    \begin{align}
        &  H(\mathcal{G}(x^0, y^0), \mathcal{G}(x^1, y^1)) \ge 1 \label{eq:lb-alpha0-gsg-eq2} \\
        &  H(\mathcal{Z}(x^0, y^0), \mathcal{Z}(x^1, y^1)) \ge 1. \label{eq:lb-alpha0-gsg-eq2-new} 
    \end{align}
    
    \paragraph{Step 3: Change of measure}
    We first set some notation. For any $x,y \in \Delta_n$, we have that $\mathbb{P}_{x, y} = \prod_{t=1}^{\tau_\delta} p_{x}(X_t) p_y(Y_t)$,
    where $p_q(\cdot)$ denotes the density function of $q \in \Delta_n$. Then, since \Cref{eq:lb-alpha0-gsg-eq2,eq:lb-alpha0-gsg-eq2-new} hold (and since $\epsilon < 1$), we can apply \Cref{lemma:ours-change-of-measure}. Thus for any $(\epsilon,\delta)$-correct algorithm, for both GSGs and ZSGs we obtain
    \begin{align*}
        \delta \ge \frac{1}{4} \exp\left( -\text{KL}(\mathbb{P}_{x^0, y^0}, \mathbb{P}_{x^1, y^1}) \right).
    \end{align*}
    We proceed by analyzing the r.h.s.~of this equation.
    In the case we are considering we have set $n=2$. Hence, $x,y$ are Bernoulli distributions, and we obtain that:
    \begin{align*}
        \text{KL}\left(\mathbb{P}_{x^0, y^0}, \mathbb{P}_{x^1, y^1}\right) & = \mathbb{E}_{x^0, y^0}\left[ \sum_{t=1}^{\tau_\delta} \text{KL}(p_{x^0}(X_t), p_{x^1}(X_t)) \right] \tag{Since $y^0=y^1$} \\
        & = \mathbb{E}_{x^0, y^0}\left[ \tau_\delta  \right] \text{KL}(p_{x^0}, p_{x^1})\tag{Wald's identity} \\
        & = \mathbb{E}_{x^0, y^0}\left[ \tau_\delta  \right] \log\left( \frac{1}{1-\pi} \right).\tag{By definition of $x^0$, $x^1$}
    \end{align*}
    Then, we obtained that:
    \begin{align*}
        \delta \ge  \frac{1}{4} \exp\left( -\mathbb{E}_{x^0, y^0}\left[ \tau_\delta  \right] \log\left( \frac{1}{1-\pi} \right) \right),
    \end{align*}
    which yields:
    \begin{align*}
        \mathbb{E}_{x^0, y^0}\left[ \tau_\delta  \right]  \ge \frac{\log\left( \frac{1}{4\delta} \right)}{\log\left( \frac{1}{1-\pi} \right)},
    \end{align*}
    thus concluding the proof.
\end{proof}

\subsection{Proof \Cref{thm:lb-alpha-0}} \label{app:lb-alpha-0-gsg}

\paragraph{Proof outline} \Cref{thm:lb-alpha-0} is obtained as a combination of three distinct lower bounds, \ie \Cref{thm:alpha0-pimin,thm:alpha0-eps,thm:alpha0-n}. In these theorems, we prove that there exists instances where the sample complexity of any $(\epsilon, \delta)$-correct algorithm is lower bounded by:
\begin{align*}
    \text{(i)} \,\,\frac{\log(\frac{1}{\delta})}{\log\left(\frac{1}{1-\pi_{\min}} \right)}, \,\, \text{ (ii)} \,\,\frac{\log\left( \frac{1}{\delta} \right)}{\epsilon^2}, \,\, \text{ and (iii)}\,\, \frac{n}{\epsilon^2},
\end{align*}
respectively (here, we neglected the exact multiplicative constants, which are listed below).
Importantly, these results hold both for General-Sum and Zero-Sum Games.  Then, \Cref{thm:lb-alpha-0} follows by the following argument.

\begin{proof}[Proof of \Cref{thm:lb-alpha-0}]
    It is sufficient to consider the instances presented in \Cref{thm:alpha0-pimin,thm:alpha0-eps,thm:alpha0-n}. The results then follows from applying $\max\{ a, b, c\} \ge \frac{1}{3} (a+b+c)$. 
\end{proof}

Now, we proceed by proving \Cref{thm:alpha0-pimin,thm:alpha0-eps,thm:alpha0-n}. We start from \Cref{thm:alpha0-pimin}, which is a simple corollary of \Cref{thm:lb-alpha-0-impossibility}.

\begin{theorem}\label{thm:alpha0-pimin}
    Let $\alpha = 0$ and $\epsilon < 1$. Then, there exists a problem instance such that, for any $(\epsilon, \delta)$-correct algorithm, the following holds: 
    \begin{align*}
        \E[\tau_\delta] \ge \frac{\log(\frac{1}{4\delta})}{\log\left(\frac{1}{1-\pi_{\min}} \right)}.
    \end{align*}
    This holds for both Zero-Sum and General-Sum Games.
\end{theorem}
\begin{proof}
    The results follow from applying \Cref{thm:lb-alpha-0-impossibility} with $\pi = \pi_{\min}$.
\end{proof}

We now continue with \Cref{thm:alpha0-eps}.

\begin{theorem}\label{thm:alpha0-eps}
    Let $\alpha = 0$ and $\epsilon < \tfrac{1}{\sqrt{32}}$. Then, there exists a problem instance such that, for any $(\epsilon, \delta)$-correct algorithm, it holds that
    \begin{align*}
        \E[\tau_\delta] \ge \frac{\log\left( \frac{1}{4\delta} \right)}{16 \epsilon^2}.
    \end{align*}
    This holds both for Zero-Sum and General-Sum Games.
\end{theorem}
\begin{proof}
    The proof is split into three steps

    \paragraph{Step 1: Instance Construction} We consider the setting with $n=2$. Then, we consider two different instances where the unknown strategy profiles $(x^0, y^0)$ and $(x^1, y^1)$ are defined as follows:
    \begin{align*}
        &x^0=\left(\frac{1}{2}, \frac{1}{2}\right),~y^0=\left(\frac{1}{2}, \frac{1}{2}\right) \\
        &x^1=\left(\frac{1}{2}+\gamma, \frac{1}{2}-\gamma\right),~y^1= \left(\frac{1}{2}, \frac{1}{2}\right),
    \end{align*}
    where $\gamma \in (0,1/2)$ is a real number that will specified later in the proof.

    \paragraph{Step 2: Haussdorff Distance} Next, we analyze the distances $H(\mathcal{G}(x^0, y^0), \mathcal{G}(x^1, y^1))$ and $H(\mathcal{Z}(x^0, y^0), \mathcal{Z}(x^1, y^1))$. In particular, we will provide lower bounds on these quantities.

    We start from $H(\mathcal{G}(x^0, y^0), \mathcal{G}(x^1, y^1))$. We have that:
    \begin{align*}
        H(\mathcal{G}(x^0, y^0), \mathcal{G}(x^1, y^1)) & \ge \max_{B \in \mathcal{G}^y(x^0,y^0)} \min_{B ' \in \mathcal{G}^y(x^1, y^1)} \| B - B'\|_{\infty} \\
        & \ge \min_{B ' \in \mathcal{G}^y(x^1, y^1)} \max \{ |1-b'_{11}|, |b'_{12}|, |b'_{21}|, |1-b'_{22}   | \},
    \end{align*}
    where, in the second step, we have set \(B=\begin{bmatrix}1&0\\ 0&1\end{bmatrix} \in \mathcal{G}^y(x^0, y^0)\) and \(B'=\begin{bmatrix}b'_{11}&b'_{12}\\ b'_{21}&b'_{22}\end{bmatrix}\).\footnote{Indeed, $(x^0)^\top B y^0 = \frac{1}{2}$ and $(x^0)^\top B e_j = \frac{1}{2}$ for all $j \in [n]$. Hence, all the constraints of $\mathcal{G}^y(x^0, y^0)$ are satisfied for such $B$.} Now, using the values of $(x^1, y^1)$, we can rewrite the last optimization problem as follows:
    \begin{align*}
        & \min_{B'}\max \{ |1-b'_{11}|, |b'_{12}|, |b'_{21}|, |1-b'_{22}   | \} \\
        & \text{s.t.\quad} b'_{ij} \in [-1, 1] \quad \forall i,j \in [n] \\
        & \phantom{\text{s.t.}\quad} \tfrac12\Bigl((\tfrac12+\gamma)(b'_{11}+b'_{12})
+(\tfrac12-\gamma)(b'_{21}+b'_{22})\Bigr) \le b'_{11}(\tfrac12+\gamma)+b'_{21}(\tfrac12-\gamma),\\
& \phantom{\text{s.t.}\quad}\tfrac12\Bigl((\tfrac12+\gamma)(b'_{11}+b'_{12})
+(\tfrac12-\gamma)(b'_{21}+b'_{22})\Bigr) \le b'_{12}(\tfrac12+\gamma)+b'_{22}(\tfrac12-\gamma).
    \end{align*}
    Introducing an additional variable $t \coloneqq \|B-B' \|_{\infty} \ge 0$ and simplifying the constraints, we can rewrite this optimization problem as follows:
    \begin{equation}\label{eq:lb-alpha0-gsg-eq1-eps}
    \begin{alignedat}{2}
        & \min_{B' \in [-1,1]^{2\times 2},t \ge 0} t \\
        & \text{s.t.\quad} |1-b'_{11}|\le t,~~~|b'_{12}| \le t,~~~ |b'_{21}| \le t,~~~ |1-b'_{22}|\le t   \\
        & \phantom{\text{s.t.}\quad}(b'_{11}-b'_{12})(\tfrac12+\gamma)+(b'_{21}-b'_{22})(\tfrac12-\gamma)=0.
    \end{alignedat}
    \end{equation}
    Let $t(B')$ be the optimal value of Problem \eqref{eq:lb-alpha0-gsg-eq1-eps} for any feasible matrix $B'$.  In the following, we will prove that $t(B') \ge \gamma$ for any feasible matrix $B'$. To this end, let us introduce the following vector $g^\top = [(\tfrac12+\gamma), (-\tfrac12-\gamma), (\tfrac12-\gamma), (-\tfrac12+\gamma) ]$. Then, the last linear constraints in \eqref{eq:lb-alpha0-gsg-eq1-eps} can be rewritten as $g^\top z(B')=0$, where $z(B')$ denotes the flattened version of matrix $B'$, i.e., $z(B')=(b'_{11}, b'_{12}, b'_{21}, b'_{22})$. 
    Next, let $z_0 = (1,0,0,1)$. Note that $g^\top z_0 = 2 \gamma$. Thus, for any feasible $B'$, $g^\top (z(B') - z_0) = -2\gamma$. Therefore,
    \begin{align*}
        2\gamma & = |g^\top (z(B') - z_0)| \\
        & \le \| g \|_{1} \| z(B') - z_0 \|_{\infty} \tag{H\"{o}lder's inequality}\\
        & =\| g \|_1 t(B') \tag{Since $z_0 = z(B)$ and $t(B') = \| B'-B\|_{\infty} $} \\
        & = 2t(B') \tag{Definition of $g$ for $\gamma \in (0,\tfrac12)$} 
    \end{align*}
    Thus, we obtained that, for any feasible matrix $B'$, $t(B') \ge \gamma$. Choosing $\gamma = 2\epsilon$,\footnote{This choice is valid since $\gamma \in (0, \tfrac12)$ and $\epsilon < 1/4$ by assumption.} we obtained that 
    \begin{align}\label{eq:lb-alpha0-gsg-eq2-eps}
    H(\mathcal{G}(x^0, y^0), \mathcal{G}(x^1, y^1)) \ge 2\epsilon    .
    \end{align}

    Next, we analyze \(H(\mathcal{Z}(x^0, y^0), \mathcal{Z}(x^1, y^1))\). We have that:
    \begin{align*}
    H(\mathcal{Z}(x^0, y^0), \mathcal{Z}(x^1, y^1))
    &\ge \max_{A \in \mathcal{Z}(x^0,y^0)} 
        \min_{A' \in \mathcal{Z}(x^1,y^1)} \|A-A'\|_\infty \\
    &\ge \min_{A' \in \mathcal{Z}(x^1,y^1)}
    \max\bigl\{|1-a'_{11}|,\ |a'_{12}|,\ |a'_{21}|,\ |1-a'_{22}|\bigr\},
    \end{align*}
    where, in the second step, we have set
    \[
    A=\begin{bmatrix}1&0\\0&1\end{bmatrix}\in \mathcal{Z}(x^0,y^0).
    \footnote{Indeed, $(x^0)^\top A y^0 = \frac{1}{2}$, $e_j^\top A y^0 = \frac{1}{2}$  and $(x^0)^\top A e_j = \frac{1}{2}$ for all $j \in [n]$. Hence, all the constraints of $\mathcal{Z}(x^0, y^0)$ are satisfied.}
    \]
    Now, using the values of \((x^1,y^1)\), we can rewrite the last optimization problem as follows:
    \begin{align*}
    & \min_{A'}\max \{ |1-a'_{11}|,\ |a'_{12}|,\ |a'_{21}|,\ |1-a'_{22}| \} \\
    & \text{s.t.\quad} a'_{ij} \in [-1, 1] \quad \forall i,j \in \{1,2\} \\
    & \phantom{\text{s.t.}\quad}(\tfrac12+\gamma)a'_{11}+(\tfrac12-\gamma)a'_{21}
    \le \tfrac12\Bigl((\tfrac12+\gamma)(a'_{11}+a'_{12})+(\tfrac12-\gamma)(a'_{21}+a'_{22})\Bigr) \\
    & \phantom{\text{s.t.}\quad}(\tfrac12+\gamma)a'_{12}+(\tfrac12-\gamma)a'_{22}
    \le \tfrac12\Bigl((\tfrac12+\gamma)(a'_{11}+a'_{12})+(\tfrac12-\gamma)(a'_{21}+a'_{22})\Bigr) \\
    & \phantom{\text{s.t.}\quad}\tfrac12\Bigl((\tfrac12+\gamma)(a'_{11}+a'_{12})+(\tfrac12-\gamma)(a'_{21}+a'_{22})\Bigr)
    \le \tfrac12(a'_{11}+a'_{12}) \\
    & \phantom{\text{s.t.}\quad}\tfrac12\Bigl((\tfrac12+\gamma)(a'_{11}+a'_{12})+(\tfrac12-\gamma)(a'_{21}+a'_{22})\Bigr)
    \le \tfrac12(a'_{21}+a'_{22}).
    \end{align*}

    Introducing the additional variable \(t\coloneqq \|A-A'\|_\infty\ge 0\) and simplifying the constraints, we can rewrite the above minimization as the following optimization problem:
    \begin{equation}\label{eq:lb-alpha0-zsg-eq1-eps}
    \begin{alignedat}{2}
    & \min_{A'\in[-1,1]^{2\times 2},~t\ge 0}\quad t\\
    & \text{s.t.}\quad |1-a'_{11}|\le t,\quad |a'_{12}|\le t,\quad |a'_{21}|\le t,\quad |1-a'_{22}|\le t,\\
    & \phantom{\text{s.t.}\quad}
    (\tfrac12+\gamma)(a'_{11}-a'_{12})+(\tfrac12-\gamma)(a'_{21}-a'_{22})=0,\\
    & \phantom{\text{s.t.}\quad}\tfrac12\Bigl((\tfrac12+\gamma)(a'_{11}+a'_{12})+(\tfrac12-\gamma)(a'_{21}+a'_{22})\Bigr)
    \le \tfrac12(a'_{11}+a'_{12}) \\
    & \phantom{\text{s.t.}\quad}\tfrac12\Bigl((\tfrac12+\gamma)(a'_{11}+a'_{12})+(\tfrac12-\gamma)(a'_{21}+a'_{22})\Bigr)
    \le \tfrac12(a'_{21}+a'_{22}).
    \end{alignedat}
    \end{equation}

    Now, as we did for General-Sum Games, let $t(A')$ be the optimal value of Problem \eqref{eq:lb-alpha0-zsg-eq1-eps} for any feasible matrix $A'$.  In the following, we will prove that $t(A') \ge \gamma$ for any feasible matrix $A'$. To this end, let us introduce the following vector $g^\top = [(\tfrac12+\gamma), (-\tfrac12-\gamma), (\tfrac12-\gamma), (-\tfrac12+\gamma) ]$. Then, the second constraint in \eqref{eq:lb-alpha0-zsg-eq1-eps} can be rewritten as $g^\top z(A')=0$, where $z(A')$ denotes the flattened version of matrix $A'$, i.e., $z(A')=(a'_{11}, a'_{12}, a'_{21}, a'_{22})$. 
    Next, let $z_0 = (1,0,0,1)$. Note that $g^\top z_0 = 2 \gamma$. Thus, for any feasible $A'$, $g^\top (z(A') - z_0) = -2\gamma$. Therefore,
    \begin{align*}
        2\gamma & = |g^\top (z(A') - z_0)| \\
        & \le \| g \|_{1} \| z(A') - z_0 \|_{\infty} \tag{H\"{o}lder's inequality}\\
        & =\| g \|_1 t(A') \tag{Since $z_0 = z(A)$ and $t(A') = \| A'-A\|_{\infty} $} \\
        & = 2t(A'). \tag{Definition of $g$ for $\gamma \in (0,\tfrac12)$} 
    \end{align*}
    Thus, we obtained that, for any feasible matrix $A'$, $t(A') \ge \gamma$. Choosing $\gamma = 2\epsilon$,\footnote{This choice is valid since $\gamma \in (0, \tfrac12)$ and $\epsilon < 1/4$ by assumption.} we obtained that 
    \begin{align}\label{eq:lb-alpha0-zsg-eq2-eps}
    H(\mathcal{Z}(x^0, y^0), \mathcal{Z}(x^1, y^1)) \ge 2\epsilon.
    \end{align}

    \paragraph{Step 3: Change of measure} For any $x,y \in \Delta_n$, let $\mathbb{P}_{x, y} = \prod_{t=1}^{\tau_\delta} p_{x}(X_t) p_y(Y_t)$,
    where $p_q(\cdot)$ denotes the density function of $q \in \Delta_n$. Then, since \Cref{eq:lb-alpha0-gsg-eq2-eps,eq:lb-alpha0-zsg-eq2-eps} holds, we can apply \Cref{lemma:ours-change-of-measure}. Thus, both for GSGs and ZSGs, we have that:
    \begin{align*}
        \delta \ge \frac{1}{4} \exp\left( -\text{KL}(\mathbb{P}_{x^0, y^0}, \mathbb{P}_{x^1, y^1}) \right).
    \end{align*}
    We proceed by analyzing the r.h.s. of this equation.
    In the case we are considering we have set $n=2$. Hence, $x,y$ are Bernoulli distributions, and we obtain that:
    \begin{align*}
        \text{KL}\left(\mathbb{P}_{x^0, y^0}, \mathbb{P}_{x^1, y^1}\right) & = \mathbb{E}_{x^0, y^0}\left[ \sum_{t=1}^{\tau_\delta} \text{KL}(p_{x^0}(X_t), p_{x^1}(X_t) \right] \tag{Since $y^0=y^1$} \\
        & = \mathbb{E}_{x^0, y^0}\left[ \tau_\delta  \right] \text{KL}(p_{x^0}, p_{x^1})\tag{Wald's identity} \\
        & =  \frac12\mathbb{E}_{x^0, y^0}\left[ \tau_\delta  \right]\log\!\left(\frac{1}{1-16\epsilon^2}\right) \tag{By def. of $x^0,x^1$ and setting $\gamma=2\epsilon$}\\
        & \le \mathbb{E}_{x^0, y^0}\left[ \tau_\delta  \right] \frac{8 \epsilon^2}{1-16\epsilon^2} \tag{$-\log(1-x) \le \frac{x}{1-x}$ for $x \le 1$} \\
        & \le 16 \mathbb{E}_{x^0, y^0}\left[ \tau_\delta  \right] \epsilon^2 \tag{$\epsilon \le \sqrt{\tfrac{1}{32}} \implies (1-16\epsilon^2) \ge \tfrac12$}
    \end{align*}
    Hence, we obtained that:
    \begin{align*}
        \delta \ge  \frac{1}{4} \exp\left( -16\mathbb{E}_{x^0, y^0}\left[ \tau_\delta  \right] \epsilon^2 \right)
    \end{align*}
    Rearranging the terms yields:
    \begin{align*}
        \mathbb{E}_{x^0, y^0}\left[ \tau_\delta  \right]  \ge \frac{\log\left( \frac{1}{4\delta} \right)}{16 \epsilon^2},
    \end{align*}
    which concludes the proof.
\end{proof}

Finally, we prove \Cref{thm:alpha0-n}.

\begin{theorem}\label{thm:alpha0-n}
    Let $\alpha = 0$, $n \ge 20$, $n \mod 16 = 0$ and $\epsilon \le \frac{1}{384}$. Then, there exists a problem instance such that, for any $(\epsilon, \delta)$-correct algorithm, the following holds:
    \begin{align*}
        \E[\tau_\delta] \ge \frac{\log(2)n}{40(192)^2\epsilon^2}.
    \end{align*}
    This holds both for Zero-Sum and General-Sum Games.
\end{theorem}
\begin{proof}As before, the proof is split into three main steps.
    \paragraph{Step 1: Instance Construction} In the following, taking inspiration from the work of \citet{jin2020reward}, we construct a set of different $n$-dimensional instances. We assume that $n \ge 16$ and that $n \mod 16 = 0$. These instances are parameterized with an $n$-dimensional vector $v \in \{-1,1 \}^n$ such that $\sum_{k=1}^n v_k = 0$ and are constructed as follows. We denote by $\mathcal{V}$ the set of vectors that satisfies these two requirements. Now, let $\gamma \in (0, \tfrac12)$; we define $(x^v, y^v)$ as follows:
    \begin{align*}
        x^v = (1, 0, \dots 0) \text{ and } y^v = \left( \frac{1+\gamma v_1}{n}, \frac{1+ \gamma v_2}{n}, \dots, \frac{1+\gamma v_n}{n}  \right).
    \end{align*}


Observe that $n \ge 2, \gamma \in (0, \tfrac12), \text{and~} v \in \{-1,1 \}^n $. Therefore, by construction, we have that $y^v\in \Delta_n$.

    In the rest of the proof, we restrict our attention to vectors $v$ belonging to a subset $\bar{\mathcal{V}}$ of $\mathcal{V}$. In particular, consider a pair $v,w \in \bar{\mathcal{V}}$ and let 
    \begin{align*}
        D^+_{v,w} = \{ k \in [n]: v_k = 1 \text{ and } w_k = -1\} \quad \quad D^-_{v,w} = \{k \in [n]: v_k=-1 \text{ and }w_k = 1 \},
    \end{align*}
    where we will drop the explicit dependence on $v,w$ when clear from context. 
    Then, we require that for any pair $v,w\in\bar{\mathcal{V}}$ it holds
    \begin{align}
        |D^+_{v,w}| = |D^-_{v,w}| \ge \frac{n}{64}.\label{eq:lb-alpha0-n-dep-eq1} 
    \end{align}
    It is easy to see that $|D^+_{v,w}| = |D^-_{v,w}|$ is granted for any pair $v,w \in \mathcal{V}$. The requirement that, \eg $|D^+_{v,w}| \ge \tfrac{n}{64}$, is instead equivalent to impose that $\sum_{k=1}^n |v_k - w_k| \ge \tfrac{n}{16}$. Thanks to \Cref{lemma:pack-arg}, we know that:
    \begin{align}\label{eq:pack-arg}
        \exists \bar{\mathcal{V}} \subseteq \mathcal{V} \,\text{ s.t. }\, \forall v,w \in \bar{\mathcal{V}},  \quad \sum_{k=1}^n |v_k - w_k| \ge \frac{n}{16} \text{~and~} |\bar{\mathcal{V}}| > 2^{\frac{n}{5}}.
    \end{align}
    In the following step, we will consider two generic profiles $(x^v,y^v)$ and $(x^w,y^w)$ such that $v,w \in \bar{\mathcal{V}}$.

    \paragraph{Step 2: Haussdorff Distance} We analyze the distances $H(\mathcal{G}(x^v, y^v), \mathcal{G}(x^w, y^w))$ and $H(\mathcal{Z}(x^v, y^v), \mathcal{Z}(x^w, y^w))$. In particular, we will provide lower bounds on these quantities.

    We start from $H(\mathcal{G}(x^v, y^v), \mathcal{G}(x^w, y^w))$. We have that:
    \begin{align*}
        H(\mathcal{G}(x^v, y^v), \mathcal{G}(x^w, y^w)) & \ge \max_{A \in \mathcal{G}^x(x^v,y^v)} \min_{A' \in \mathcal{G}^x(x^w, y^w)} \| A - A'\|_{\infty} \\
        & \ge \min_{A' \in \mathcal{G}^x(x^w, y^w)} \|A- A' \|_{\infty} ,
    \end{align*}
    where we define the matrix $A \in \mathcal{G}(x^v, y^v)$ as
    \begin{equation}\label{eq:lb-alpha0-n-dep-eq2}
    A_{ij}:=
    \begin{cases} 
    0, & i=1, j \in [n]\\[2pt]
    \tfrac13, & i>1,j\in D^+,\\[2pt]
    -\tfrac13 \kappa , & i> 1,j\in D^-,\\[2pt]
    0, & i>1,j\in[n]\setminus(D^+\cup D^-),
    \end{cases}
    \qquad\text{where}\qquad
    \kappa \coloneqq \frac{\sum_{j\in D^+} y^v_j}{\sum_{j\in D^-} y^v_j}.
    \end{equation}
    Before continuing, we prove that $A$ constructed as in \Cref{eq:lb-alpha0-n-dep-eq2} satisfies $A \in \mathcal{G}(x^v, y^v)$. To this end, we first verify that $A \in [-1,1]^{n \times n}$. It suffices to check that $\kappa \le 3$. With some simple manipulations, we have that:
    \begin{align*}
        \kappa = \frac{\sum_{j\in D^+} y^v_j}{\sum_{j\in D^-} y^v_j} = \frac{\frac{|D^+|(1+\gamma)}{n}}{\frac{|D^-|(1-\gamma)}{n}} = \frac{1+\gamma}{1-\gamma} \le 3,
    \end{align*}
    where the second equality is due to \Cref{{eq:lb-alpha0-n-dep-eq1}} and the last step is due to $\gamma \in (0, \tfrac12)$. Next, we verify that $A$ satisfies the Nash constraints. Since $x^v = e_1$, those constraints are equivalent to $(e_1 - e_i)^\top A y^v \le 0$ for all $i \in [n]$. Using the definition of $A$ and $y^v$ we have that:
    \begin{align*}
        (e_1 - e_i)^\top A y^v & = \sum_{j=1}^n y_j^v (A_{1j}-A_{ij}) \\ 
        & = -\sum_{j=1}^n y_j^v A_{ij} \tag{$A_{1j} = 0$ for all $j \in [n]$} \\
        & = - \frac13 \left(  \sum_{j \in D^+} y_j^v  - \kappa\sum_{j \in D^-} y_j^v \right) \tag{Definition of $A$} \\
        & = 0. \tag{Definition of $\kappa$}
    \end{align*}
    Hence, all constraints are tight and $A \in \mathcal{G}^x(x^v, y^v)$.

    We can now continue by lower bounding $\min_{A' \in \mathcal{G}^x(x^w, y^w)} \|A- A' \|_{\infty}$. In the following, we will prove that:
    \begin{align}\label{eq:lb-alpha0-n-dep-eq3}
        \min_{A' \in \mathcal{G}^x(x^w, y^w)} \|A- A' \|_{\infty} \ge \frac{2|D^+|\gamma}{3n(1-\gamma)}.
    \end{align}
    
    To this end, consider any matrix $A' \in \mathcal{G}^x(x^w, y^w)$. Let $\Delta(A') = A'-A$ so that $A' = \Delta(A')+A$. 
    Since $x^w = e_1$, then $A'$ satisfies $(e_1 - e_i)^\top A' y^w \le 0$ for all $i \in [n]$. Moreover, for all $A' \in \mathcal{G}^x(x^w, y^w)$ it holds that:
    \begin{align*}
        (e_1 - e_i)^\top A' y^w & = (e_1 - e_i)^\top (\Delta(A')+A) y^w  \\
        & = (e_1 - e_i)^\top \Delta(A') y^w + (e_1 - e_i)^\top A y^w \\
        & = (e_1 - e_i)^\top \Delta(A') y^w - \sum_{j=1}^n y^w_jA_{ij} \tag{$A_{1j} = 0$ for all $j \in [n]$} \\
        & = (e_1 - e_i)^\top \Delta(A') y^w - \frac13 \left( \frac{|D^+|(1-\gamma)}{n} -  \kappa \frac{|D^-|(1+\gamma)}{n} \right) \tag{Def. of $A,y^w, D^+$ and $D^-$.} \\
        & = (e_1 - e_i)^\top \Delta(A') y^w - \frac13 \left( \frac{|D^+|(1-\gamma)}{n} -  \frac{|D^+|(1+\gamma)^2}{n(1-\gamma)} \right) \tag{Def. of $\kappa$} \\
        & = (e_1 - e_i)^\top \Delta(A') y^w + \frac{4|D^+|\gamma}{3n(1-\gamma)}  \\
        & \le 0 \tag{$A' \in \mathcal{G}^x(x^w, y^w)$}
    \end{align*}
    Then, it follows that any feasible $A' \in \mathcal{G}^x(x^w, y^w)$ satisfies $(e_i - e_1)^\top \Delta(A') y^w \ge \frac{4|D^+|\gamma}{3n(1-\gamma)}$. This yields the following:
    \begin{align*}
        \frac{4|D^+|\gamma}{3n(1-\gamma)} &  \le (e_i - e_1)^\top \Delta(A') y^w \\
        & \le \sum_{j=1}^n y^w_j \left( |\Delta(A')_{ij}| + |\Delta(A')_{1j}| \right) \\
        & \le 2 \sum_{j=1}^n y^w_j  \| \Delta(A') \|_{\infty} \\
        & = 2 \|A- A' \|_{\infty},
    \end{align*}
    which proves \Cref{eq:lb-alpha0-n-dep-eq3}.

    Combining this result with \Cref{eq:lb-alpha0-n-dep-eq1} gives $H(\mathcal{G}(x^v, y^v), \mathcal{G}(x^w, y^w)) \ge \frac{\gamma}{96}$. Picking $\gamma = 192 \epsilon$,\footnote{This choice is valid. Indeed, $\epsilon \le \frac{1}{384}$ by assumption and we imposed $\gamma < \frac12$.} we obtained
    \begin{align}\label{eq:lb-alpha0-n-dep-eq4}
        H(\mathcal{G}(x^v, y^v), \mathcal{G}(x^w, y^w)) \ge 2 \epsilon.
    \end{align}

    Next we analyze \(H\!\left(\mathcal Z(x^v,y^v),\mathcal Z(x^w,y^w)\right)\). We have that:
    \begin{align*}
        H\!\left(\mathcal Z(x^v,y^v),\mathcal Z(x^w,y^w)\right) & \ge \max_{A \in \mathcal{Z}(x^v,y^v)} \min_{A' \in \mathcal{Z}(x^w, y^w)} \| A - A'\|_{\infty} \\
        & \ge \min_{A' \in \mathcal{Z}(x^w, y^w)} \|A- A' \|_{\infty},
    \end{align*}
    where the matrix $A \in \mathcal{Z}(x^v, y^v)$ is defined as in \Cref{eq:lb-alpha0-n-dep-eq2}. We already proved above that $A \in [-1,1]^{n \times n}$. It remains to prove that it satisfies the Nash constraints for zero-sum games, \ie that $A$ satisfies:
    \begin{align*}
        e_1^\top A e_j\le e_1^\top A y^v \le e_i^\top A y^v \quad \forall i,j \in [n].
    \end{align*}
    We already proved above that $e_1^\top A y^v \le e_i^\top A y^v$ holds for all $i \in [n]$. We only need to check that $e_1^\top A e_j\le e_1^\top A y^v$ holds for all $j \in [n]$. Since $A_{1j} = 0$ for all $j \in [n]$, we have that:
    \begin{align*}
        e_1^\top A e_j = A_{1j} = 0 = \sum_{j \in [n]} y^v_j A_{1j} = e_1^\top A y^v,
    \end{align*}
    thus showing that $A \in \mathcal{Z}(x^v, y^v)$. 

    We now continue by lower bounding $\min_{A' \in \mathcal{Z}(x^w, y^w)} \|A- A' \|_{\infty}$. We can follow the same reasoning that we used for General-Sum Games to prove that:
    \begin{align}\label{eq:lb-alpha0-n-dep-eq6}
        H(\mathcal{Z}(x^v, y^v), \mathcal{Z}(x^w, y^w)) \ge 2 \epsilon.
    \end{align}
    
    \paragraph{Step 3: Change of measure}
    For any $x,y \in \Delta_n$, let $\mathbb{P}_{x, y} = \prod_{t=1}^{\tau_\delta} p_{x}(X_t) p_y(Y_t)$,
    where $p_q(\cdot)$ denotes the density function of $q \in \Delta_n$. Then, since \Cref{eq:lb-alpha0-n-dep-eq4,eq:lb-alpha0-n-dep-eq6} holds, we can apply \Cref{lemma:ours-change-of-measure-several-instances} using the set of instances $\{(x^v, y^v)\}_{v \in \bar{\mathcal{V}}}$ and $(x^0, y^0) = \!(e_1, (\tfrac{1}{n}, \dots \tfrac{1}{n})\!)$ as base instance. Thus, for both GSGs and ZSGs we have that  for any $(\epsilon, \delta)$-correct algorithm it holds
    \begin{align*}
        \delta \ge 1-\frac{1}{\log (|\bar{\mathcal{V}}|)} \left( \frac{1}{|\bar{\mathcal{V}}|}\sum_{v \in \bar{\mathcal{V}}} \textup{KL}(\mathbb{P}_{x^v, y^v}, \mathbb{P}_{x^0, y^0}) + \log2  \right).
    \end{align*}
    We proceed by analyzing the kl-divergence term in the r.h.s. of this equation. Specifically, we have that:
    \begin{align*}
        \textup{KL}(\mathbb{P}_{x^v, y^v}, \mathbb{P}_{x^0, y^0}) & = \mathbb{E}_{(x^v,y^v)}\left[ \sum_{t=1}^{\tau_\delta} \text{KL}(p_{y^v}(Y_t),p_{y^0}(Y_t)) \tag{Since $x^v = x^0$} \right]  \\
        & = \mathbb{E}_{(x^v,y^v)}\left[ \tau_\delta \right] \text{KL}(p_{y^v},p_{y^0}) \tag{Wald's identity} \\
        & \le 2 (192^2)\epsilon^2 \mathbb{E}_{(x^v,y^v)}\left[ \tau_\delta \right] \tag{\Cref{lemma:kl-divergence-aux-lemma}}.
    \end{align*}
    Plugging this result in the previous equation and rearranging the terms, we obtain:
    \begin{align*}
        \frac{1}{|\bar{\mathcal{V}}|} \sum_{v \in \bar{\mathcal{V}}} \mathbb{E}_{(x^v,y^v)}\left[ \tau_\delta \right] \ge \frac{(1-\delta) \log(|\bar{\mathcal{V}}|) + \log 2}{2(192)^2 \epsilon^2}.
    \end{align*}
    Since $\mathbb{E}_{(x^v,y^v)}\left[ \tau_\delta \right] \le \max_{w \in \bar{\mathcal{V}}} \mathbb{E}_{(x^w,y^w)}\left[ \tau_\delta \right]$ for all $v \in \bar{\mathcal{V}}$, the last equation implies that there exists $\bar{w} \in \bar{\mathcal{V}}$ such that:
    \begin{align*}
        \mathbb{E}_{(x^{\bar{w}},y^{\bar{w}})}\left[ \tau_\delta \right]  \ge \frac{(1-\delta) \log(|\bar{\mathcal{V}}|) + \log 2}{2(192)^2 \epsilon^2}
    \end{align*}
    Finally, we have
    \begin{align*}
        \mathbb{E}_{(x^{\bar{w}},y^{\bar{w}})}\left[ \tau_\delta \right] &  \ge \frac{(1-\delta) \log(|\bar{\mathcal{V}}|) + \log 2}{2(192)^2 \epsilon^2} \\
        & \ge \frac{\log(2) ((1-\delta)\frac{n}{5} + 1) }{2(192)^2 \epsilon^2} \tag{\Cref{eq:pack-arg}} \\
        & \ge \frac{\log(2)n}{40(192)^2\epsilon^2}  \tag{$\delta \le \tfrac12$ and $n \ge 20$},
    \end{align*}
    which concludes the proof.
\end{proof}

\subsection{Proof of \Cref{thm:lb-any-alpha}}\label{app:lb-any-alpha-gsg}

Similarly to \Cref{thm:lb-alpha-0}, the proof of \Cref{thm:lb-any-alpha} is obtained by combining two distinct lower bounds, \ie \Cref{thm:lb-any-alpha-log,thm:lb-any-alpha-n}. In these theorems, we prove that there exist instances where the sample complexity of any $(\epsilon, \delta)$-correct algorithm is lower bounded by:
\begin{align*}
    \text{ (i)} \,\,\frac{\log\left( \frac{1}{\delta} \right)}{ \epsilon^2\alpha}, \,\, \text{ and (ii)}\,\, \frac{n}{\epsilon^2\alpha},
\end{align*}
here, we only dropped the multiplicative constants.
These results hold both for General-Sum Games and Zero-Sum Games. \Cref{thm:lb-any-alpha} follows by the following argument.

\begin{proof}[Proof of \Cref{thm:lb-any-alpha}]
It is sufficient to consider the case where the agent observes the instances presented in \Cref{thm:lb-any-alpha-log,thm:lb-any-alpha-n} with probability $(\tfrac12, \tfrac12)$. The result is then straightforward.
\end{proof}

We now proceed by proving \Cref{thm:lb-any-alpha-log,thm:lb-any-alpha-n}.

\begin{theorem}
\label{thm:lb-any-alpha-log}
Let $\alpha \in (0, \tfrac14)$ and $\epsilon \le \tfrac{1}{8}$. Then, there exists a problem instance such that, for any $(\epsilon, \delta)$-correct algorithm, the following holds:
\[
\mathbb{E}[\tau_\delta] \geq 
\frac{\log(\frac{1}{4\delta})}{18\epsilon^{2}\,\alpha}.
\]
This holds for both Zero-Sum and General-Sum Games.
\end{theorem}

\begin{proof}
The proof is divided into three steps.
\paragraph{Step 1: Instance Construction}

To prove the result, we consider the case in which $n=2$, and we have two different instances where the unknown strategy profiles $(x^0, y^0)$ and $(x^1, y^1)$ are defined as follows:
\begin{align*}
    &x^0=\left(\tfrac{\alpha}{2}, 1- \tfrac{\alpha}{2} \right),~y^0=\left(1, 0\right) \\
    &x^1=\left( \tfrac{\alpha}{2} + \tfrac{\beta}{2}, 1-\tfrac{\alpha}{2}-\tfrac{\beta}{2} \right),~y^1= \left(1, 0\right),
\end{align*}
where $\beta \in (0, \tfrac14)$ is a parameter that will be tuned later on.
Since $\alpha \in (0, \tfrac14)$ and $\beta \in (0, \tfrac14)$, $x^0$ and $x^1$ both belong to the simplex.

\paragraph{Step 2: Haussdorff Distance} Next, we lower bound $H(\mathcal{G}_\alpha(x^0, y^0), \mathcal{G}_\alpha(x^1, y^1))$ and $H(\mathcal{Z}_\alpha(x^0, y^0), \mathcal{Z}_\alpha(x^1, y^1))$. 

We start from $H(\mathcal{G}_\alpha(x^0, y^0), \mathcal{G}_\alpha(x^1, y^1))$. We have that:
\begin{align*}
    H(\mathcal{G}_\alpha(x^0, y^0), \mathcal{G}_\alpha(x^1, y^1)) & \ge \max_{A \in \mathcal{G}_\alpha^x(x^0, y^0)} \min_{A' \in \mathcal{G}_\alpha^x(x^1, y^1)} \| A - A'\|_{\infty} \\
    & \ge \min_{A' \in \mathcal{G}_\alpha(x^1, y^1)} \max \{ |1-a'_{11}|, |1+a'_{12}|, |1+a'_{21}|, |1+a'_{22}| \},
\end{align*}
where, in the second step, we have set $A=\begin{bmatrix}1&-1\\ -1&-1\end{bmatrix} \in \mathcal{G}_\alpha^x(x^0, y^0)$.
\footnote{Indeed, $(x^0)^\top A y^0 = -1+\alpha$, $e_1^\top A y^0 = \tfrac\alpha2$ and $e_2^\top A y^0 = -1$. Hence, $(x^0)^\top A y^0 \le e_i^\top A y^0 + \alpha$ for all $i \in \{1,2\}$.}

Now, using the values of $x^1$ and $y^1$, we are interested in solving the following problem:
\begin{equation*}
    \begin{alignedat}{2}
     \min_{A' \in [-1,1]^{2 \times 2}}\max &\,\,\{ |1-a'_{11}|, |1+a'_{12}|, |1+a'_{21}|, |1+a'_{22}| \} \\
     \text{s.t.~} & (\tfrac\alpha2+\tfrac\beta2) a'_{11}+(1-\tfrac\alpha2 - \tfrac \beta 2) a'_{21} \le a'_{11}+\alpha\\
    & (\tfrac\alpha2+\tfrac\beta2) a'_{11}+(1-\tfrac\alpha2 - \tfrac \beta 2) a'_{21} \le a'_{21}+\alpha.
    \end{alignedat}
\end{equation*}
Simplifying the constraints we can rewrite this optimization problem as follows:
\begin{equation}\label{eq:lb-anyalpha-eq1}
    \begin{alignedat}{2}
     \min_{A' \in [-1,1]^{2 \times 2}} & \max \,\, \{ |1-a'_{11}|, |1+a'_{21}| \} \ \\
     \text{s.t.~} \,\,\quad & a'_{21}-a'_{11} \le  \tfrac{2\alpha}{2-\alpha-\beta} \\
    & a'_{11} - a'_{21} \le \tfrac{2\alpha}{\alpha + \beta}.
    \end{alignedat}
\end{equation}
Now, we note that:
\begin{align*}
    \max \{ |1-a'_{11}|, |1+a'_{21}| \} & \ge \frac{1}{2} \left( |1-a'_{11}|+ |1+a'_{21}|\right) \tag{$\max\{a,b\} \ge \tfrac12 (a+b)$}    \\
    & \ge \frac{1}{2} |2 - (a'_{11}-a'_{21}) |. \tag{Triangular inequality}
\end{align*}
Let $\Delta = a'_{11} - a'_{21}$ and observe that $\Delta \in [-2, 2]$. We can then lower bound the value of Problem \eqref{eq:lb-anyalpha-eq1} with the value of the following optimization problem:
\begin{equation}\label{eq:lb-anyalpha-eq2}
    \begin{alignedat}{2}
    \min_{\Delta \in [-2,2]} & \frac12 |2-\Delta| \ \\
     \text{s.t.~} &  -\frac{2\alpha}{2-\alpha-\beta}\le \Delta \le \frac{2\alpha}{\alpha+\beta}.
    \end{alignedat}
\end{equation}
This can be easily solved in closed form, yielding $\Delta^\star = 2\frac{\alpha}{\alpha+\beta}$. Consequently, we have obtained the following lower bound on the Hausdorff distance: $H(\mathcal{G}_\alpha(x^0, y^0), \mathcal{G}_\alpha(x^1, y^1)) \ge \frac{\beta}{\alpha+\beta}$.

Set $\beta = \tfrac83\epsilon \alpha$. We have obtained that:\footnote{In this step, we make use of $\epsilon < \tfrac18$.}
\begin{align}\label{eq:lb-anyalpha-eq3}
    H(\mathcal{G}_\alpha(x^0, y^0), \mathcal{G}_\alpha(x^1, y^1)) \ge 2\epsilon. 
\end{align}

We continue by lower bounding $H(\mathcal{Z}_\alpha(x^0, y^0), \mathcal{Z}_\alpha(x^1, y^1))$. We have that:
\begin{align*}
    H(\mathcal{Z}_\alpha(x^0, y^0), \mathcal{Z}_\alpha(x^1, y^1)) & \ge \max_{A \in \mathcal{Z}_\alpha(x^0, y^0)} \min_{A' \in \mathcal{Z}_\alpha(x^1, y^1)} \|A - A' \|_\infty \\
    & \ge \min_{A' \in \mathcal{Z}_\alpha(x^1, y^1)}  \max \{ |1-a'_{11}|, |1+a'_{12}|, |1+a'_{21}|, |1+a'_{22}| \},
\end{align*}
where, in the second step, we have set $A=\begin{bmatrix}1&-1\\ -1&-1\end{bmatrix} \in \mathcal{Z}_\alpha(x^0, y^0)$. Indeed, we already proved above that $A$ satisfies $-1+\alpha = (x^0)^\top A y^0 \le e_i^\top A y^0 + \alpha $ for all $i \in \{1,2\}$. It is also easy to see that $(x^0)^\top A e_1 = -1+\alpha$  and $(x^0)^\top Ae_2 = -1$ for all $j \in \{ 1,2\}$. Hence, we obtained $A \in \mathcal{Z}_\alpha(x^0, y^0)$. Then, one can proceed as above to obtain:\footnote{It is sufficient to drop all the constraints that arise from $(x^1)^\top A' y^1 \ge (x^1)^\top A' e_j - \alpha $ in the optimization problems.}
\begin{align}\label{eq:lb-anyalpha-eq4}
    H(\mathcal{Z}_\alpha(x^0, y^0), \mathcal{Z}_\alpha(x^1, y^1)) \ge 2\epsilon. 
\end{align}

\paragraph{Step 3: Change of measure}

For any $x,y \in \Delta_n$, let $\mathbb{P}_{x, y} = \prod_{t=1}^{\tau_\delta} p_{x}(X_t) p_y(Y_t)$,
where $p_q(\cdot)$ denotes the density function of $q \in \Delta_n$. Then, since \Cref{eq:lb-anyalpha-eq3,eq:lb-anyalpha-eq4} hold, we can apply \Cref{lemma:ours-change-of-measure}. For both GSGs and ZSGs we have 
\begin{align*}
    \delta \ge \frac{1}{4} \exp\left( -\text{KL}(\mathbb{P}_{x^0, y^0}, \mathbb{P}_{x^1, y^1}) \right).
\end{align*}

We proceed by analyzing the r.h.s. of this equation.
In the case we are considering we have set $n=2$. Hence, $x,y$ are Bernoulli distributions, and we obtain that:
\begin{align*}
    \text{KL}\left(\mathbb{P}_{x^0, y^0}, \mathbb{P}_{x^1, y^1}\right) & = \mathbb{E}_{x^0, y^0}\left[ \sum_{t=1}^{\tau_\delta} \text{KL}(p_{x^0}(X_t), p_{x^1}(X_t) )\right] \tag{Since $y^0=y^1$} \\
    & = \mathbb{E}_{x^0, y^0}\left[ \tau_\delta  \right] \text{KL}(p_{x^0}, p_{x^1})\tag{Wald's identity} \\
    & \le 2\mathbb{E}_{x^0, y^0}\left[ \tau_\delta  \right] \frac{\text{TV}(p_{x^0}, p_{x^1})^2}{\alpha} \tag{Reverse Pinsker's Inequality} \\
    & \le 2\mathbb{E}_{x^0, y^0}\left[ \tau_\delta  \right] \frac{(\tfrac83 \epsilon \alpha)^2}{\alpha}\\
    & \le 18 \mathbb{E}_{x^0, y^0}\left[ \tau_\delta  \right] {\alpha \epsilon^2}.
\end{align*}
Hence, we obtained that:
\begin{align*}
    \delta \ge  \frac{1}{4} \exp\left( -18\mathbb{E}_{x^0, y^0}\left[ \tau_\delta  \right] \alpha \epsilon^2 \right).
\end{align*}
Rearranging the terms yields:
\begin{align*}
    \mathbb{E}_{x^0, y^0}\left[ \tau_\delta  \right]  \ge \frac{\log\left( \frac{1}{4\delta} \right)}{18 \alpha \epsilon^2}.
\end{align*}
which concludes the proof.
\end{proof}

Finally, it remains to prove the last result that composes \Cref{thm:lb-any-alpha}.

\begin{theorem}
\label{thm:lb-any-alpha-n}
Let $\alpha \in (0, \tfrac12)$ and $\epsilon \le \tfrac{1}{128}$. Then, there exists a problem instance such that, for any $(\epsilon, \delta)$-correct algorithm, the following holds:
\[
\mathbb{E}[\tau_\delta] \geq 
\frac{\log(2)n}{20(128)^2\alpha\epsilon^2}. 
\]
\end{theorem}
This holds for both Zero-Sum and General-Sum Games.

\begin{proof}
The proof is split into three main steps.

\paragraph{Step 1: Instance Construction} 
    Similar to \Cref{thm:alpha0-n}, we construct a large set of different $n$-dimensional instances. We assume that $n\ge 22$ and that $n -1 \mod 16 = 0$. These instances are parameterized with an $n-1$-dimensional vector $v \in \{-1,1 \}^{n-1}$ such that $\sum_{k=1}^{n-1} v_k = 0$ and are constructed as follows. We denote by $\mathcal{V}$ the set of vectors that satisfies these two requirements. Now, let $\gamma \in (0, \alpha)$; we define $(x^v, y^v)$ as follows:
    \begin{align*}
        y^v = (1, 0, \dots 0) \text{ and } x^v = \left( \frac{\alpha+\gamma v_1}{n-1}, \frac{\alpha+ \gamma v_2}{n-1}, \dots, \frac{\alpha+\gamma v_{n-1}}{n-1}, 1-\alpha\right).
    \end{align*}

    Observe that $n -1 \ge 16, \gamma \in (0, \alpha), \alpha \in (0,\tfrac12)~  \text{and~} v \in \{-1,1 \}^{n-1} \implies x^v \in [0,1]$. Furthermore
    $\sum_{k=1}^{n-1} v_k = 0 \implies \sum_{j=1}^n x^v_k=1$. Therefore, we have that $x^v\in \Delta_n$. 

    In the rest of the proof, we restrict our attention to vectors belonging to a subset $\bar{\mathcal{V}}$ of $\mathcal{V}$. Precisely, consider any pair $v,w \in \bar{\mathcal{V}}$ and let 
    \begin{align*}
        D^+_{v,w} = \{ k \in [n-1]: v_k = 1 \text{ and } w_k = -1\} \quad \quad D^-_{v,w} = \{k \in [n-1]: v_k=-1 \text{ and }w_k = 1 \},
    \end{align*}
    where we will drop the explicit dependence on $v,w$ when clear from context. 
    Then, we require that for any $v,w\in\bar{ \mathcal{V}}$ it holds
    \begin{align}
        |D^+_{v,w}| = |D^-_{v,w}| \ge \frac{n-1}{16}.\label{eq:lb-any-alpha-n-eq1} 
    \end{align}
    It is easy to see that $|D^+_{v,w}| = |D^-_{v,w}|$ is satisfied for any pair $v,w \in \mathcal{V}$. The requirement that, \eg $|D^+| \ge \tfrac{n-1}{16}$, is instead equivalent to imposing that $\sum_{k=1}^{n-1} |v_k - w_k| \ge \tfrac{n-1}{16}$. Thanks to \Cref{lemma:pack-arg}, we know that:
    \begin{align}\label{eq:lb-any-alpha-n-eq2}
        \exists \bar{\mathcal{V}} \subseteq \mathcal{V}: \forall v,w \in \bar{\mathcal{V}},  \quad \sum_{k=1}^{n-1} |v_k - w_k| \ge \frac{n-1}{16} \text{~and~} |\bar{\mathcal{V}}| > 2^{\frac{n-1}{5}}.
    \end{align}
    In the following step, we will consider two generic profiles $(x^v,y^v)$ and $(x^w,y^w)$ such that $v,w \in \bar{\mathcal{V}}$.

\paragraph{Step 2: Hausdorff Distance}

    We analyze the distances $H(\mathcal{G}_\alpha(x^v, y^v), \mathcal{G}_\alpha(x^w, y^w))$ and $H(\mathcal{Z}_\alpha(x^v, y^v), \mathcal{Z}_\alpha(x^w, y^w))$. In particular, we will provide lower bounds on these quantities.

    We start from $H(\mathcal{G}_\alpha(x^v, y^v), \mathcal{G}_\alpha(x^w, y^w))$. We have that:
    \begin{align*}
        H(\mathcal{G}_\alpha(x^v, y^v), \mathcal{G}_\alpha(x^w, y^w)) & \ge \max_{A \in \mathcal{G}_\alpha^x(x^v,y^v)} \min_{A' \in \mathcal{G}_\alpha^x(x^w, y^w)} \| A - A'\|_{\infty} \\
        & \ge \min_{A' \in \mathcal{G}_\alpha^x(x^w, y^w)} \|A- A' \|_{\infty} ,
    \end{align*}
    where the matrix $A \in \mathcal{G}^x_\alpha(x^v, y^v)$ that we used in the second inequality is constructed as follows:

    \begin{equation}\label{eq:lb-any-alpha-n-eq3}
    A_{ij}:=
    \begin{cases} 
    -1 & i \in [n], j > 1, \\[2pt]
    - \kappa & i \in D^+,j=1,\\[2pt]
    1   & i \in D^- ,j=1,\\[2pt]
    0 & i \in [n-1] \setminus (D^+ \cup D^-),j=1, \\
    -1 & i=n, j =1.
    \end{cases}
    \qquad\text{where}\qquad
    \kappa \coloneqq \frac{\sum_{j\in D^-} x^v_j}{\sum_{i\in D^+} x^v_i}.
    \end{equation}
    Before continuing, we prove that $A$ constructed as in \Cref{eq:lb-any-alpha-n-eq3} satisfies $A \in \mathcal{G}^x_\alpha(x^v, y^v)$. To this end, we first verify that $A \in [-1,1]^{n \times n}$. It suffices to check $|\kappa| \le 1$. With some simple manipulations, we have that:
    \begin{align*}
        |\kappa| = \frac{\sum_{j\in D^-} x^v_j}{\sum_{i\in D^+} x^v_i} = \frac{\frac{|D^-|(\alpha-\gamma)}{n-1}}{\frac{|D^+|(\alpha+\gamma)}{n-1}} = \frac{\alpha-\gamma}{\alpha+\gamma} \le 1,
    \end{align*}
    where the second equality is due to \Cref{eq:lb-any-alpha-n-eq1} and the definition of $x^v$. Next, we verify that $A$ satisfy the $\alpha$-Nash constraints. To this end, we need to verify that $(x^v)^\top A y \le -1 +\alpha$. We note that:
    \begin{align*}
        (x^v)^\top A y & = \sum_{i \in D^+} x^v_i A_{i1} + \sum_{i \in D^-} x_i^v A_{i1} + x_n A_{n1} \\
        & =  -\kappa \sum_{i \in D^+} x^v_i + \sum_{i \in D^-} x_i^v+(\alpha - 1)  \tag{Def. of $A$ and $x^v$} \\
        & = -1+\alpha \tag{Def. of $\kappa$}. 
    \end{align*}
    Hence, all constraints are satisfied and $A \in \mathcal{G}_\alpha^x(x^v, y^v)$.

    We can now continue by lower bounding $\min_{A' \in \mathcal{G}_\alpha^x(x^w, y^w)} \|A- A' \|_{\infty}$. Specifically, let us introduce the auxiliary variable $t = \| A - A' \|_{\infty}$. Then, using the definition of $y^w$, we can rewrite this optimization problem as follows:
    \begin{equation}\label{eq:lb-any-alpha-n-eq4}
    \begin{alignedat}{2}
         \min_{A' \in [-1,1]^{n\times n},t \ge 0} & t \\
         \text{s.t.\quad} & |A'_{ij}-A'_{ij}| \le t,~\forall i,j \in [n]^2\\
        & \sum_{i \in [n]} x^w_i A'_{i1} \le  \min_{i \in [n]} A'_{i1} + \alpha.
    \end{alignedat}
    \end{equation}
    Consider any feasible matrix $A'$. Then, it must hold that:
    \begin{align*}
        A'_{1n} + \alpha & \ge (x^w)^\top A' y \tag{Last constraint of \Cref{eq:lb-any-alpha-n-eq4}}\\
        & = \sum_{i \ne n}x^w_i A'_{i1} + (1-\alpha) A'_{n1} \tag{Def. of $x^w$}\\
        & = \sum_{i \ne n} x^w_i A_{i1} + \sum_{i \ne n} (x_i^w (A'_{i1} - A_{i1})) + (1-\alpha) A'_{n1}  \\
        & \ge  \sum_{i \ne n} x^w_i A_{i1} - t \alpha + (1-\alpha) A'_{n1} \tag{$|A'_{ij} - A_{ij}| \le t$ and $\sum_{i \ne n}x^w_i = \alpha$}\\
        & = \sum_{i \ne n} (x^w_i - x^v_i) A_{i1} - t \alpha + (1-\alpha) A'_{n1} \tag{$\sum_{i \ne n}x_i^v A_{i1} = 0$} \\
        & = \frac{2\kappa |D^+| \gamma}{n-1} + \frac{2|D^-| \gamma}{n-1} - t \alpha + (1-\alpha) A'_{n1} \tag{Def. of $x$ and $A$} \\
        & = \frac{4|D^+|\gamma}{n-1} \left( \frac{\alpha}{\alpha+\gamma} \right)  - t \alpha + (1-\alpha) A'_{n1} \tag{Def. of $\kappa$} \\
        & \ge \frac{\gamma \alpha}{16(\alpha + \gamma)} - t \alpha + (1-\alpha) A'_{n1} \tag{\Cref{eq:lb-any-alpha-n-eq1}}.
    \end{align*}
    Hence, rearranging the terms, and manipulating the equations, we obtain that any feasible $A'$ satisfies:
    \begin{align*}
         \frac{\gamma \alpha}{16(\alpha + \gamma)} & \le \alpha (1+ A'_{n1}) + t \alpha \\
         & \le 2 t \alpha  \tag{$|A_{n1} - A'_{n1}| = |1+A'_{n1}| \le t$}
    \end{align*}

    Therefore, we arrived at $H(\mathcal{G}_\alpha(x^v, y^v), \mathcal{G}_\alpha(x^w, y^w)) \ge \frac{\gamma}{32(\alpha + \gamma)}$. Picking $\gamma = 128 \alpha \epsilon$, we obtained that:\footnote{In this step, we used that $\epsilon \le \tfrac{1}{128}$.}
    \begin{align}\label{eq:lb-any-alpha-n-eq5}
        H(\mathcal{G}_\alpha(x^v, y^v), \mathcal{G}_\alpha(x^w, y^w)) \ge 2 \epsilon.
    \end{align}

    Next we analyze \(H\!\left(\mathcal Z_\alpha(x^v,y^v),\mathcal Z_\alpha(x^w,y^w)\right)\). We have that:
    \begin{align*}
        H\!\left(\mathcal Z_\alpha(x^v,y^v),\mathcal Z_\alpha(x^w,y^w)\right) & \ge \max_{A \in \mathcal{Z}_\alpha(x^v,y^v)} \min_{A' \in \mathcal{Z}_\alpha(x^w, y^w)} \| A - A'\|_{\infty} \\
        & \ge \min_{A' \in \mathcal{Z}_\alpha(x^w, y^w)} \|A- A' \|_{\infty},
    \end{align*}
    where the matrix $A \in \mathcal{Z}_\alpha(x^v, y^v)$ is defined as in \Cref{eq:lb-any-alpha-n-eq3}. We already proved above that $A \in [-1,1]^{n \times n}$. It remains to prove that it satisfies the $\alpha$-Nash constraints for zero-sum games, \ie that $A$ satisfies:
    \begin{align*}
        (x^v)^\top A e_j - \alpha \le (x^v)^\top A y \le e_i^\top A y + \alpha, \quad \forall i,j \in [n]^2.
    \end{align*}
    Above, we have shown that $-1+\alpha=(x^v)^\top A y \le e_i^\top A y + \alpha, \quad \forall i,j \in [n]^2$. Now, it is sufficient to note that $(x^v)^\top A e_1 = (x^v)^\top A y$ and $(x^v)^\top A e_j = -1$ for all $j > 1$. Hence, $A \in \mathcal{Z}_\alpha(x^v, y^v)$.

    Then, one can proceed as above to obtain:
    \begin{align}\label{eq:lb-anyalpha-n-eq6}
        H(\mathcal{Z}_\alpha(x^v, y^v), \mathcal{Z}_\alpha(x^w, y^w)) \ge 2\epsilon. 
    \end{align}

\paragraph{Step 3: Change of measure}
    For any $x,y \in \Delta_n$, let $\mathbb{P}_{x, y} = \prod_{t=1}^{\tau_\delta} p_{x}(X_t) p_y(Y_t)$,
    where $p_q(\cdot)$ denotes the density function of $q \in \Delta_n$. Then, since \Cref{eq:lb-anyalpha-n-eq6,eq:lb-anyalpha-n-eq6} holds, we can apply \Cref{lemma:ours-change-of-measure-several-instances} using the set of instances $\{(x^v, y^v)\}_{v \in \bar{\mathcal{V}}}$ and $(x^0, y^0) = \!((\tfrac{\alpha}{n-1}, \dots \tfrac{\alpha}{n-1}, 1-\alpha), e_1\!)$ as base instance. Thus, both for GSGs and ZSGs, we have that, for any $(\epsilon, \delta)$-correct algorithm, it holds that:
    \begin{align*}
        \delta \ge 1-\frac{1}{\log (|\bar{\mathcal{V}}|)} \left( \frac{1}{|\bar{\mathcal{V}}|}\sum_{v \in \bar{\mathcal{V}}} \textup{KL}(\mathbb{P}_{x^v, y^v}, \mathbb{P}_{x^0, y^0}) + \log2  \right).
    \end{align*}
    We proceed by analyzing the kl-divergence term in the r.h.s. of this equation. Specifically, we have that:
    \begin{align*}
        \textup{KL}(\mathbb{P}_{x^v, y^v}, \mathbb{P}_{x^0, y^0}) & = \mathbb{E}_{(x^v,y^v)}\left[ \sum_{t=1}^{\tau_\delta} \text{KL}(p_{x^v}(X_t),p_{x^0}(X_t)) \tag{Since $y^v = y^0$} \right]  \\
        & = \mathbb{E}_{(x^v,y^v)}\left[ \tau_\delta \right] \text{KL}(p_{x^v},p_{x^0}) \tag{Wald's identity} \\
        & \le 2 (128^2) \alpha \epsilon^2 \mathbb{E}_{(x^v,y^v)}\left[ \tau_\delta \right] \tag{\Cref{lemma:kl-divergence-aux-lemma}}.
    \end{align*}
    Plugging this result in the previous equation and rearranging the terms, we obtain:
    \begin{align*}
        \frac{1}{|\bar{\mathcal{V}}|} \sum_{v \in \bar{\mathcal{V}}} \mathbb{E}_{(x^v,y^v)}\left[ \tau_\delta \right] \ge \frac{(1-\delta) \log(|\bar{\mathcal{V}}|) + \log 2}{2(128)^2 \alpha \epsilon^2}.
    \end{align*}
    Since $\mathbb{E}_{(x^v,y^v)}\left[ \tau_\delta \right] \le \max_{w \in \bar{\mathcal{V}}} \mathbb{E}_{(x^w,y^w)}\left[ \tau_\delta \right]$ for all $v \in \bar{\mathcal{V}}$, the last equation implies that there exists $\bar{w} \in \bar{\mathcal{V}}$ such that:
    \begin{align*}
        \mathbb{E}_{(x^{\bar{w}},y^{\bar{w}})}\left[ \tau_\delta \right]  \ge \frac{(1-\delta) \log(|\bar{\mathcal{V}}|) + \log 2}{ 2 (128^2) \alpha \epsilon^2}
    \end{align*}
    With some final algebraic steps, we can now conclude the proof:
    \begin{align*}
        \mathbb{E}_{(x^{\bar{w}},y^{\bar{w}})}\left[ \tau_\delta \right] &  \ge \frac{(1-\delta) \log(|\bar{\mathcal{V}}|) + \log 2}{2(128)^2 \epsilon^2} \\
        & \ge \frac{\log(2) ((1-\delta)\frac{(n-1)}{5} + 1) }{2(128)^2 \alpha\epsilon^2} \tag{\Cref{eq:lb-any-alpha-n-eq2}} \\
        & \ge \frac{\log(2)n}{20(128)^2\alpha\epsilon^2}  \tag{$\delta \le \tfrac12$ and $n \ge 22$},
    \end{align*}
    This concludes the proof.
\end{proof}

\section{Proof of \Cref{thm:ub-alpha-0}}\label{app:ub1}

This section is structured as follows: 
\begin{itemize}
    \item In \Cref{app:alpha0-gsg} we prove the result for the inverse GSG problem (\Cref{thm:gsg-alpha0}) 
    \item In \Cref{app:zsg-alpha0}, we prove the result of the inverse ZSG problem (\Cref{thm:zsg-alpha0})
\end{itemize}
The proof of \Cref{thm:ub-alpha-0} then follows directly by combining \Cref{thm:gsg-alpha0,thm:zsg-alpha0}.

\subsection{Proof of \Cref{thm:gsg-alpha0} ($\alpha=0$, General-Sum Games)}\label{app:alpha0-gsg}

\paragraph{Proof Outline} First, in \Cref{lemma:matconstr}, we show that, for any pair of strategies $(x,y)$ and $(\hat x, \hat y)$ such that $\supp(x)=\supp(\hat x)$ and $\supp(y)=\supp(\hat y)$, then $\mathcal{G}(x,y)$ is close to $\mathcal{G}(\hat x, \hat y)$ by construction. Specifically, for any point $(A, B) \in \mathcal{G}(x, y)$, there exist $(\hat A, \hat B) \in \mathcal{G}(\hat x, \hat y)$ such that $A \approx \hat A$ and $B \approx \hat B$. Importantly the approximation errors scale linearly with $\|x - \hat x\|_1$ and $\|y- \hat y\|_1$. Then, in \Cref{lemma:matrix-to-h}, we show how \Cref{lemma:matconstr} implies that the Hausdorff distance between the two sets is directly proportional to $\|x - \hat x\|_1$ and $\|y- \hat y\|_1$. In \Cref{lemma:high-prob-alpha0-gsg} we combine this result together with probabilistic argument. In this way, we are able to provide an high probability error upper bound on $H(\mathcal{G}(x,y), \mathcal{G}(\hat x, \hat y))$ that is explicitly related to the number of data $m$ collected by \Cref{alg:WAS}. Finally, to prove \Cref{thm:gsg-alpha0} we only need to select $m$ appropriately to guarantee that the final error of \Cref{alg:WAS} is bounded by $\epsilon$ with high probability.  

We start with \Cref{lemma:matconstr}. Note that here, $x,y,\hat x, \hat y$ are any strategy on the simplex.

\begin{lemma}[Alternative Matrix Construction; $\alpha=0$ and General-Sum Games]\label{lemma:matconstr}
Let $\alpha = 0$, $x, y, \hat x, \hat y \in \Delta_n$.
$A \in \mathcal{G}^x(x,y)$ and $B \in \mathcal{G}^y(x,y)$ and $\hat{x},\hat{y} \in \Delta_n$.

Then if $\supp(x) = \supp(\hat x)$ then there exists $\hat A \in \mathcal{G}^x(\hat x, \hat y)$ such that
\begin{align}
    &\| A -  \hat A \|_{\infty} \le 4 \| y - \hat y \|_1 \label{eq:lem_mat_constr_eq-1},
\end{align}
and similarly, if $\supp(y) = \supp(\hat y)$ then there exists $\hat B \in \mathcal{G}^y(\hat x, \hat y)$ such that
\begin{align}
    &\| B - \hat B \|_{\infty} \le 4 \| x - \hat x\|_1 \label{eq:lem_mat_constr_eq-2}.
\end{align}

Furthermore, $\hat A$ and $\hat B$ are defined as: 
\[
\hat A_{ij} :=
\begin{cases}
\dfrac{A_{ij}}{1+2\|y - \hat y\|_1}, & \text{if } i \in \argmax _{k \in \supp(x)} \sum_{j \in [n]} \hat{y}_j A_{kj},\\[12pt]
\dfrac{A_{ij} + \Delta_i^x}{1+2\|y - \hat y\|_1}, & \text{if } i \in \supp(x) \setminus \argmax _{k \in \supp(x)} \sum_{j \in [n]} \hat{y}_j A_{kj},\\[12pt]
\dfrac{A_{ij} + 2\|y - \hat y\|_1}{1+2\|y - \hat y\|_1}, & \text{if } i \notin \supp(x)
\end{cases}
\]

\[
\hat B_{ij} := 
\begin{cases}
\dfrac{B_{ij}}{1+2\|x - \hat x\|_1}, & \text{if } j \in \argmax _{k \in \supp(y)} \sum_{i \in [n]} \hat{x}_i B_{ik},\\[12pt]
\dfrac{B_{ij} + \Delta_i^y}{1+2\|x - \hat x\|_1}, & \text{if } j \in \supp(y) \setminus \argmax _{k \in \supp(y)} \sum_{i \in [n]} \hat{x}_i B_{ik},\\[12pt]
\dfrac{B_{ij} + 2\|x - \hat x\|_1}{1+2\|x - \hat x\|_1}, & \text{if } j \notin \supp(y),
\end{cases}
\]
where, for all $i \in \supp(x)$ and all $j \in \supp(y)$, $\Delta_i^x$ and $\Delta_j^y$ are defines as follows:
\begin{align*}
    &\Delta_i^x =\max_{k \in \supp(x)} \sum_{j \in [n]} \hat{y}_j A_{kj} - \sum_{j \in [n]} \hat{y}_j A_{ij} \\
    &\Delta_j^y = \max_{k \in \supp(y)} \sum_{i \in [n]} \hat x_i B_{ik} - \sum_{i \in[n]} \hat x_i B_{ij}.
\end{align*}

\end{lemma}

\begin{proof}
We first prove \Cref{eq:lem_mat_constr_eq-1} by construction. 
For brevity, let $S_x:=\supp(x) \subseteq [n]$ and define the set of maximizers of $\sum_{j \in [n]} \hat{y}_j A_{ij}$ over the different rows of the support of $x$ as $S^\star_x$, namely $S^{\star}_x \coloneqq \argmax_{k\in S_x} \sum_{j \in [n]} \hat{y}_j A_{kj}$.
Note that, since $A_{ij} \in [-1,1]$ and $\hat y \in \Delta_n$, we have that $\Delta_i^x \in [0,2]$. Furthermore, for all $i \in S_x$
\begin{align*}
\Delta_i^x & = \max_{k \in S_x} \sum_{j \in [n]} \hat{y}_j A_{kj} - \sum_{j \in [n]} \hat y_j A_{ij} \\
& = \sum_{j \in [n]} \hat y_j (A_{kj} - A_{ij}) \tag{For any $k \in S^\star_x$} \\
& = \sum_{j \in [n]} y_j (A_{kj} - A_{ij}) + \sum_{j \in [n]} (\hat y_j - y_j) (A_{kj}-A_{ij}) \\
& = \sum_{j \in [n]} (\hat y_j - y_j) (A_{kj}-A_{ij}) \tag{\Cref{lemma:nash-support} and $i, k \in S_x$} \\
& \le 2 \| y - \hat y \|_1 \tag{$-1 \le A \le 1$} \\
& \coloneqq \Delta,
\end{align*}
where in the last step we have introduced $\Delta \coloneqq 2 \|y - \hat y \|_1$.

Now, we are ready to define a matrix $\hat A$ that we will use to prove \Cref{eq:lem_mat_constr_eq-1}. Specifically, $\hat A$ is defined as follows:
\[
\hat A_{ij}  := 
\begin{cases}
\dfrac{A_{ij}}{1+\Delta}, & \text{if } i \in S^{\star}_x,\\
\dfrac{A_{ij} + \Delta_i^x}{1+\Delta}, & \text{if } i \in S_x \setminus S^{\star}_x,\\
\dfrac{A_{ij} + \Delta}{1+\Delta}, & \text{if } i \notin S_x.
\end{cases}
\]
First of all, note that since $A \in [-1,1]^{n \times n}$ and $\Delta_i^x \le \Delta$, we have that $\hat A \in [-1,1]^{n \times n}$. 

At this point, we continue by showing that $\hat A \in \mathcal G^x(\hat x,\hat y)$. To this end, since $\alpha=0$ and $\supp(x) = \supp(\hat x) = S_x$, due to \Cref{lemma:nash-support} we only need to show that
\begin{subequations}\label{eq:support}
\begin{align}
e_i^{\top}\hat A \hat y &= e_j^{\top}\hat A \hat y \quad \forall\, i,j\in S_x, \label{eq:support-eq}\\
e_i^{\top}\hat A \hat y &\le e_j^{\top}\hat A \hat y \quad \forall\, i\in S_x,\ j\notin S_x. \label{eq:support-ineq}
\end{align}
\end{subequations}

\paragraph{Condition \eqref{eq:support-eq}}
To prove that Condition \eqref{eq:support-eq} holds, we first show that for any $i_1, i_2 \in S^*_x$, we have $e_{i_1}^{\top}\hat A \hat y = e_{i_2}^{\top}\hat A \hat y$. Specifically, let $i_1,i_2 \in S^\star_x$. Then,
\begin{align*}
    e_{i_1}^\top \hat A \hat y & = \frac{1}{1+\Delta} \sum_{j \in [n]} A_{i_1 j} \hat y_j \\
    & = \frac{1}{1+\Delta} \sum_{j \in [n]} A_{i_2 j} \hat y_j \tag{$S^\star_x = \argmax_{i \in S_x} \sum_{j \in [n]} \hat y_j A_{i j}$ and $i_1, i_2 \in S^\star_x$} \\
    & = e_{i_2}^\top \hat A \hat y.
\end{align*}
To conclude the proof of Condition \eqref{eq:support-eq}, we need to consider  $i_1\in S^{\star}_x$, and $i_2\in S_x \setminus S^{\star}_x$ and prove that $e_{i_1}^\top \hat A \hat y = e_{i_2}^\top \hat A \hat y$. Using the definition of $\hat A$, we have that:
\begin{align*}
    e_{i_1}^\top \hat A \hat y = e_{i_2}^\top \hat A \hat y \iff \sum_{j \in [n]} \hat y_j A_{i_1 j} = \sum_{j \in [n]} \hat y_j A_{i_2 j} + \Delta_{i_2}^x.
\end{align*}
However, by definition of $\Delta_{i_2}^x$, we have that,
\begin{align*}
    \Delta_{i_2}^x = \max_{k \in S} \sum_{j \in [n]} \hat y_j A_{k j} - \sum_{k \in [n]} \hat y_j A_{i_2 k} = \sum_{j \in [n]} \hat y_j A_{i_1 j} - \sum_{k \in [n]} \hat y_j A_{i_2 k},
\end{align*}
hence, $\sum_{j \in [n]} \hat y_j A_{i_1 j} = \sum_{j \in [n]} \hat y_j A_{i_2 j} + \Delta_{i_2}$ holds, and $e_{i_1}^\top \hat A \hat y = e_{i_2}^\top \hat A \hat y$ holds as well. Finally, we note that the above argument also covers the case where
$i_1, i_2 \in S_x \setminus S_x^\star$. Indeed, for any $i \in S_x \setminus S_x^\star$, we have
\[
e_i^\top \hat A \hat y = e_{i^\star}^\top \hat A \hat y
\]
for any $i^\star \in S_x^\star$. Therefore, for any
$i_1, i_2 \in S_x \setminus S_x^\star$,
\[
e_{i_1}^\top \hat A \hat y
= e_{i^\star}^\top \hat A \hat y
= e_{i_2}^\top \hat A \hat y,
\]
which shows that Condition~\eqref{eq:support-eq} holds for all pairs
$i_1,i_2 \in S_x$.

\paragraph{Condition \eqref{eq:support-ineq}.}
To prove that Condition \eqref{eq:support-ineq} holds, since we already proven Condition \eqref{eq:support-eq}, it is sufficient to show that, for $i_1 \in S^\star_x$ and $i_2 \notin S_x$, we have that $e_{i_1}^\top \hat A \hat y \le e_{i_2}^\top \hat A \hat y$. 
Consider the following inequalities:
\begin{align*}
    e_{i_1}^\top \hat A \hat y & = \frac{1}{1+\Delta} \left( \sum_{j \in [n]} \hat y_j A_{i_1 j}\right) \\
    & = \frac{1}{1+\Delta} \left( \sum_{j \in [n]} y_j A_{i_1 j} + \sum_{j \in [n]} (\hat y_j - y_j) A_{i_1 j} \right) \\
    & \le \frac{1}{1+\Delta} \left( \sum_{j \in [n]} y_j A_{i_2 j} + \sum_{j \in [n]} (\hat y_j - y_j) A_{i_1 j} \right) \tag{Due to $A \in \mathcal{G}^x(x,y)$ and \Cref{lemma:nash-support}} \\
    & = \frac{1}{1+\Delta} \left( \sum_{j \in [n]} \hat y_j A_{i_2 j} + \sum_{j \in [n]} (\hat y_j - y_j) (A_{i_1 j}  - A_{i_2 j})\right) \\
    & \le \frac{1}{1+\Delta} \left( \sum_{j \in [n]} \hat y_j A_{i_2 j} + 2\| y - \hat y\|_1 )\right) \\
    & = e_{i_2}^\top \hat A \hat y.
\end{align*}
Hence, we proved that $\hat A \in\mathcal G^x(\hat x,\hat y)$.

Finally, to conclude the proof of \Cref{eq:lem_mat_constr_eq-1}, it remains to bound $\|A - \hat{A}\|_\infty$.
Fix any pair of $i,j \in [n]$. Then, if $i \in S^\star_x$, we have that:
\[
|A_{ij}-\hat A_{ij}|
= \left|A_{ij}-\frac{A_{ij}}{1+\Delta}\right|
= |A_{ij}|\,\frac{\Delta}{1+\Delta}
\le \frac{\Delta}{1+\Delta} \le 2 \Delta.
\]
If, instead, $i\in S_x\setminus S^{\star}_x$:
\[
|A_{ij}-\hat A_{ij}|
= \left|A_{ij}-\frac{A_{ij}+\Delta_i^x}{1+\Delta}\right|
= \frac{|\,\Delta A_{ij}-\Delta_i^x\,|}{1+\Delta}
\le \frac{|A_{ij}|\Delta + \Delta_i^x}{1+\Delta}
\le \frac{2\Delta}{1+\Delta} \le 2\Delta,
\]
Finally, if $i\notin S_x$, we have that:
\[
|A_{ij}-\hat A_{ij}|
= \left|A_{ij}-\frac{A_{ij}+\Delta}{1+\Delta}\right|
= \frac{\Delta|1-A_{ij}|}{1+\Delta}
\le \frac{2\Delta}{1+\Delta} \le 2\Delta.
\]
Hence, $\|A-\hat A\|_\infty \le 2\Delta$, thus concluding the proof of \Cref{eq:lem_mat_constr_eq-1}.

We continue by proving \Cref{eq:lem_mat_constr_eq-2}. In particular, the proof of \Cref{eq:lem_mat_constr_eq-2} follows from arguments that are symmetrical to the ones that we presented above. More precisely, consider $B \in \mathcal{G}^y(x,y)$. Let $S_y \coloneqq \supp(y)$ and for all $j \in [n]$ and define:
\[
\Delta_j^y = \max_{k \in S_y} \sum_{i \in [n]} \hat x_i B_{ik} - \sum_{i \in [n]} \hat x_i B_{ij}.
\]
Furthermore, let $S^{\star}_y = \argmax_{j \in [n]} \sum_{i \in [n]} \hat x_i B_{ij}$. Then, we define $\hat B$ as follows:
\[
\hat B_{ij}  := 
\begin{cases}
\dfrac{B_{ij}}{1+\Delta}, & \text{if } j \in S^{\star}_y,\\
\dfrac{B_{ij} + \Delta_i^x}{1+\Delta}, & \text{if } j \in S_y \setminus S^{\star}_y,\\
\dfrac{B_{ij} + \Delta}{1+\Delta}, & \text{if } j \notin S_y.
\end{cases}
\]
where $\Delta = 2 \|x - \hat x \|_1$. The proof of the results than follows from arguments that are analogous to the ones that we presented for \Cref{eq:lem_mat_constr_eq-1}.
\end{proof}

We now leverage \Cref{lemma:matconstr} to provide an upper bound on the Hausdorfff distance.

\begin{lemma}[Matrix Construction $\to$ Hausdorfff Distance]\label{lemma:matrix-to-h}
    Let $\alpha = 0$, $x,y,\hat x, \hat y \in \Delta_n$. Suppose that $\supp(x) = \supp(\hat x)$ and $\supp(y) = \supp(\hat y)$. Then, it holds that:
    \begin{align*}
        H(\mathcal{G}(x,y), \mathcal{G}(\hat x, \hat y)) \le 4 (\| x - \hat x\|_1 + \|y - \hat y \|_1).
    \end{align*}
\end{lemma}
\begin{proof}
    The proof is a direct application of \Cref{lemma:matconstr}. Precisely:
    \begin{align*}
        H(\mathcal{G}(x,y), \mathcal{G}(\hat x, \hat y)) & = \max \Bigg \{ \max_{(A,B) \in \mathcal{G}(x,y)}  \min_{(\hat A, \hat B) \in \mathcal{G}(\hat x, \hat y)} \max\{ \|A - \hat{A} \|_{\infty}, \|B - \hat{B}\|_\infty , \\&\hspace{2cm} \max_{(\hat A, \hat B) \in \mathcal{G}(\hat x, \hat y)}  \min_{(A, B) \in \mathcal{G}(x, y)} \max\{ \|A - \hat{A} \|_{\infty}, \|B - \hat{B}\|_\infty   \Bigg\}  \\
        & \le 4\max \left\{ \|x-\hat x \|_1, \|y - \hat y \|_1  \right\} \tag{\Cref{lemma:matconstr}} \\
        & \le 4 (\|x-\hat x \|_1 + \|y-\hat y \|_1).
    \end{align*}
    which concludes the proof. Here, we simply remark that, in order to apply \Cref{lemma:matconstr} in the first inequality, it is important to recall that $\mathcal{G}(x,y) = \mathcal{G}^x(x,y) \times \mathcal{G}^{y}(x,y)$.
\end{proof}

At this point, we can combine \Cref{lemma:matrix-to-h} with probabilistic arguments to obtain a high probability error bound on the Hausdorff distance that is directly related to the number of samples observed by our algorithm.

\begin{lemma}[High-Probability Hausdorff Bound]\label{lemma:high-prob-alpha0-gsg}
    Let $\alpha = 0$ and  $\hat x, \hat y \in \Delta_n$ be the maximum likelihood estimators of $x,y$ after $m$ samples. Suppose that $m$ satisfies:
    \begin{align}\label{eq:good-event-hp}
        m \ge \frac{\log\left( \frac{4n}{\delta} \right)}{\log\left( \frac{1}{1-\pi_{\min}} \right)},
    \end{align}
    
    then, with probability at least $1-\delta$, it holds that:
    \begin{align*}
        H(\mathcal{G}(x,y), \mathcal{G}(\hat x, \hat y) \le 8 \sqrt{\frac{2 \log\left(\frac{4}{\delta}\right)+2n\log(6m)}{m}}. 
    \end{align*}
\end{lemma}
\begin{proof}
    We first define the following good event $\mathcal{E}$ and shows that it holds with probability at least $1-\delta$. We define $\mathcal{E} = \mathcal{E}_1 \cap \mathcal{E}_2$, where
    \begin{align*}
        & \mathcal{E}_1 = \left\{ \| x - \hat x \|_1 \le \sqrt{\frac{2 \log\left(\frac{4}{\delta}\right)+2n\log(6m)}{m}}  \right\} \cap \left\{ \|y - \hat y \|_1 \le \sqrt{\frac{2 \log\left(\frac{4}{\delta}\right)+2n\log(6m)}{m}} \right\} \\
        & \mathcal{E}_2 = \{ \supp(x) = \supp(\hat x) \} \cap \{ \supp(y) = \supp(\hat y) \}. 
    \end{align*}
    Then, we have that $\mathbb{P}(\mathcal{E}) = 1 - \Prob(\mathcal{E}^\complement) \ge 1 - (\mathbb{P}(\mathcal{E}_1^\complement) + \Prob(\mathcal{E}^{\complement}_2))$.
    Hence, we simply upper bound $\Prob(\mathcal{E}_1^\complement)$ and $\Prob(\mathcal{E}^\complement_2)$. 
    \begin{align*}
        \Prob ( \mathcal{E}_1^\complement) & \le \Prob\left( \| x - \hat x \|_1 > \sqrt{\frac{2 \log\left(\frac{4}{\delta}\right)+2n\log(6m)}{m}}   \right) + \Prob \left( \|y - \hat y \|_1 > \sqrt{\frac{2 \log\left(\frac{4}{\delta}\right)+2n\log(6m)}{m}} \right) \\
        & \le \frac{\delta}{2}  \tag{\Cref{lemma:l1-high-prob}}.
    \end{align*}
    Furthermore, 
    \begin{align*}
        \Prob(\mathcal{E}_2^\complement) & \le \Prob(\supp(x) \ne \supp(\hat x)) + \Prob(\supp(y) \ne \supp(\hat y)) \\
        & \le \frac{\delta}{2} \tag{\Cref{eq:good-event-hp} and \Cref{lemma:support}}.
    \end{align*}
    Thus, we obtained that $\Prob(\mathcal{E}) \ge 1-\delta$.

    Therefore, due to the definition of $\mathcal{E}$ and by using \Cref{lemma:matrix-to-h}, we have that, with probability at least $1-\delta$:
    \begin{align*}
        H(\mathcal{G}(x,y), \mathcal{G}(\hat x, \hat y)) & \le 4 (\| x - \hat x\|_1 + \|y - \hat y \|_1) \\
        & \le  8 \sqrt{\frac{2 \log\left(\frac{4}{\delta}\right)+2n\log(6m)}{m}},
    \end{align*}
    which concludes the proof.
\end{proof}

We are now ready to show that \Cref{alg:WAS} is an $(\epsilon,\delta)$-correct method for learning $\mathcal{G}(x,y)$ with minimax optimal sample complexity.

\begin{theorem}\label{thm:gsg-alpha0}
    Consider $\alpha = 0$ and let $$m \in \widetilde{\mathcal{O}} \left( \frac{\log\left( \frac{1}{\delta} \right)}{\log(1/(1-\pi_{\min}))} + \frac{n+\log\left(\frac{1}{\delta}\right)}{\epsilon^2} \right).$$ Then \Cref{alg:WAS} is $(\epsilon, \delta)$-correct for General-Sum Games and its sample complexity $\tau_\delta$ is given by $m$.    
\end{theorem}
\begin{proof}
First, we give a precise expression of the number of samples $m$ needed by \Cref{alg:WAS}.

\begin{align}\label{eq:thm-41-alpha0}
    m & = \max \left\{ \frac{\log\left( \frac{4n}{\delta} \right)}{\log\left( \frac{1}{1-\pi_{\min}} \right)}, 2 \left( 2 + \frac{4\log(4/\delta)}{(\epsilon/\sqrt{128})^2} + \frac{4n}{(\epsilon/\sqrt{128})^2} \log\left( \frac{12n}{(\epsilon/\sqrt{128})^2} \right)\right)\right\} \\
    & \in \widetilde{\mathcal{O}}\left( \frac{\log\left( \frac{1}{\delta} \right)}{\log\left( \frac{1}{1-\pi_{\min}} \right)}+ \frac{\log\left( \frac{1}{\delta} \right)+n}{\epsilon^2} \right). 
\end{align}
Now, since $m \ge \frac{\log\left( \frac{4n}{\delta} \right)}{\log\left( \frac{1}{1-\pi_{\min}} \right)}$, by \Cref{lemma:high-prob-alpha0-gsg}, we have that, with probability at least $1-\delta$:
\begin{align*}
    H(\mathcal{G}(x,y), \mathcal{G}(\hat x, \hat y) & \le 8 \sqrt{\frac{2 \log\left(\frac{4}{\delta}\right)+2n\log(6m)}{m}}
\end{align*}

Thus, we need to show that, for $m$ chosen as in \Cref{eq:thm-41-alpha0}, $8 \sqrt{\frac{2 \log\left(\frac{4}{\delta}\right)+2n\log(6m)}{m}} \le \epsilon$. That holds if and only if
\begin{align*}
    \frac{ \log\left(\frac{4}{\delta}\right)+n\log(6m)}{m} \le \left( \frac{\epsilon}{\sqrt{128}} \right)^2.
\end{align*}

This, however, is granted by the definition of $m$. To this end, it is sufficient to apply \Cref{lemma:tech-lemma} with $c_1 = 4, c_2=6, K=(\epsilon/\sqrt{128})$.

Hence, after $m$ samples, \Cref{alg:WAS} is $(\epsilon, \delta)$-correct for $\alpha=0$ and General-Sum Games.
\end{proof}

\subsection{Proof of \Cref{thm:zsg-alpha0} ($\alpha=0$, Zero-Sum Games)}\label{app:zsg-alpha0}

\paragraph{Proof outline} 
The proof of \Cref{thm:zsg-alpha0} is similar in spirit to the one that we presented for \Cref{thm:gsg-alpha0}. There are, however, a couple of important differences that complicate the analysis. Importantly, these differences arise from the additional constraints that are present in the definition of $\mathcal{Z}(x,y)$, which require a single matrix to verify constraints both for the $x$ and the $y$ players. Consider, indeed, any pair of strategies $(x,y)$ and $(\hat x, \hat y)$ such that $\supp(x)=\supp(\hat x)$ and $\supp(y)=\supp(\hat y)$. Then, for any $A \in \mathcal{Z}(x, y)$, the construction of $\hat A$ that we provided in \Cref{lemma:matconstr} for general-sum games fails to satisfy the constraints of $\mathcal{Z}(\hat x, \hat y)$. To this end, consider the following simple numerical example. Let:
\begin{align*}
   A= \begin{pmatrix}
1 & -1 \\
-1 & 1 
\end{pmatrix} \quad x=y=(0.5, 0.5), \quad \hat x=(0.6, 0.4), \hat y =(0.4, 0.6)
\end{align*}
In this case, \Cref{lemma:matconstr} yields:
\begin{align*}
       \hat A= \begin{pmatrix}
1 & -0.42857143 \\
-0.71428571 & 0.71428571 
\end{pmatrix}.
\end{align*}
\begin{align*}
       \hat A\approx \begin{pmatrix}
1 & -0.43 \\
-0.71 & 0.71 
\end{pmatrix}.
\end{align*}
Since $\hat A \notin \mathcal{Z}(\hat x, \hat y)$, we cannot rely on the previous argument for this more complicated case. More generally, we found particularly challenging building analytically a matrix $\hat A$ that takes into account variations both in $x$ and $y$ strategy profiles. To this end, however, we developed a trick to take into account the single variations of the $x$ and the $y$ vector. Specifically, we start our analysis with \Cref{corollary:separation-x-y-err} which states that
\begin{align*}
    H(\mathcal{Z}(x,y), \mathcal{Z}(\hat{x}, \hat{y}))\le H(\mathcal{Z}(x,y), \mathcal{Z}(\hat x,y)) + H(\mathcal{Z}(\hat x,y), \mathcal{Z}(\hat x, \hat y)).
\end{align*}
This result allows us to construct matrices to upper bound the Hausdorff distance again by construction. Specifically, we can fix two matrices $A_1 \in \mathcal{Z}(x,y)$ and $A_2 \in \mathcal{Z}(\hat x, y)$ and find two matrices $A_3 \in \mathcal{Z}(\hat x, y)$ and $A_4 \in \mathcal{Z}(\hat x, \hat y)$ such that $A_1 \approx A_3$ and $A_2 \approx A_4$ (\Cref{lemma:matconstr-1-zsg} and \Cref{lemma:matconstr-2-zsg}). Importantly, note that $A_1$ and $A_3$ belong to sets that ``vary only across the $x$ player's strategy'' (\ie $y$ is fixed), while $A_2$ and $A_4$ belong to sets that ``vary only across the $y$ player's strategy'' (\ie $\hat x$ is fixed).
Importantly the approximation errors of both matrices scale linearly with $\|x - \hat x\|_1$ and $\|y- \hat y\|_1$. 
Then, once this is done, we proceed identically to \Cref{thm:gsg-alpha0}.

We are now ready to start the proof of \Cref{thm:zsg-alpha0}.
As anticipated, given \Cref{corollary:separation-x-y-err}, we start by fixing the $x$-vector and, for any matrix $A \in \mathcal{Z}(x, y)$ we build an alternative matrix $\hat A \in \mathcal{Z}(x, \hat y)$. A careful reader might note that the matrix constructed in the following lemma is similar to the one presented in \Cref{lemma:matconstr} for general-sum games. However, here both the statement and the proofs are different. Indeed, (i) to construct $\hat A$ we require that the support of $y$ and $\hat y$ coincide and (ii) we need to ensure that our construction satisfies a larger number of constraints.

We remark that in the following lemma, $x,y, \hat y$ are generic strategies on the simplex.

\begin{lemma}[Alternative Matrix Construction (fix x); $\alpha=0$; Zero-Sum Games]\label{lemma:matconstr-1-zsg}
Let $\alpha = 0$, and $x,y, \hat y \in \Delta_n$. Let $A \in \mathcal{Z}(x,y)$. It holds that:
\begin{align*}
    & \supp(y) = \supp(\hat y) \implies \exists \hat A \in \mathcal{Z}(x, \hat y): \| A -  \hat A \|_{\infty} \le 4 \| y - \hat y \|_1 \\
\end{align*}
Furthermore, $\hat A$ is defined as follows:
\[
\hat A_{ij}  := 
\begin{cases}
\dfrac{A_{ij}}{1+2\|y - \hat y\|_1}, & \text{if } i \in \argmax _{k \in \supp(x)} \sum_{j \in [n]} \hat{y}_j A_{kj},\\[12pt]
\dfrac{A_{ij} + \Delta_i}{1+2\|y - \hat y\|_1}, & \text{if } i \in \supp(x) \setminus \argmax _{k \in \supp(x)} \sum_{j \in [n]} \hat{y}_j A_{kj},\\[12pt]
\dfrac{A_{ij} + 2\|y - \hat y\|_1}{1+2\|y - \hat y\|_1}, & \text{if } i \notin \supp(x),
\end{cases}
\]
where, for all $i \in \supp(x)$, $\Delta_i$ is given by
\begin{align*}
    &\Delta_i =\max_{k \in \supp(x)} \sum_{j \in [n]} \hat{y}_j A_{kj} - \sum_{j \in [n]} \hat{y}_j A_{ij} \\
\end{align*}
\end{lemma}
\begin{proof}
    Given the definition of $\hat A$, we already know from \Cref{lemma:matconstr} that all the following statements are true:
    \begin{align*}
        & \text{(i)}~\|A - \hat A \|_{\infty} \le 4 \|y - \hat y \|_1 \\
        & \text{(ii)}~e_i^\top \hat A \hat y = e_j^\top \hat A \hat y \quad \quad \quad \quad~ \forall i,j \in \supp(x) \\
        & \text{(iii)}~e_i^\top \hat A \hat y \le e_j^\top \hat A \hat y \quad \quad \quad \quad \forall i \in \supp(x), \forall j \notin \supp(x). 
    \end{align*}
    Hence, from \Cref{lemma:nash-support} it only remains to verify that $\hat A$ satisfies:
    \begin{subequations}
        \begin{align}
        & x^\top \hat A e_j = x^\top \hat A \hat y~\quad \forall j \in \supp(\hat y) \label{eq:zsg-alpha0-cnd1} \\
        & x^\top \hat Ae_j \le x^\top \hat A \hat y ~\quad\forall j \notin \supp(\hat y) \label{eq:zsg-alpha0-cnd2}.
    \end{align}
    \end{subequations}
    Before proving \Cref{eq:zsg-alpha0-cnd1,eq:zsg-alpha0-cnd2} we make some preliminary considerations. First, we recall from \Cref{lemma:nash-support} that, since $A \in \mathcal{Z}(x, y)$, it holds that:
    \begin{align}
        & x^{\top}Ay = x^{\top}Ae_j =e_i^TAy, \quad \forall\, i \in \supp(x),j\in \supp(y),\label{eq:mat-constr-zero-sum-alpha0-helper}\\
        & x^\top A e_j \le x^{\top}Ay \le e_i^\top A y, \quad \forall\, i\notin \supp(x),\ \forall\, j\notin \supp(y).
    \end{align}
    Now, let $\Delta = 2 \|y-\hat y \|_1$ and let $j \in [n]$. Then, we have that:
    \begin{align*}
        x^\top \hat A e_j & = \sum_{i \in [n]} x_i \hat A_{ij} \\ & 
        = \frac{1}{1+\Delta} \left( x^\top A e_j  + \sum_{i: x_i > 0}x_i \Delta_i \right) \tag{Definition of $\hat A$}.
    \end{align*}
    Now, let us focus in detail on the term $\sum_{i: x_i > 0} x_i \Delta_i$. We have that:
    \begin{align*}
        \sum_{i: x_i > 0} x_i \Delta_i & = \sum_{i: x_i > 0} x_i \sum_{j \in [n]} \hat y_j (A_{kj} - A_{ij}) \tag{Definition of $\Delta_i$ and $k \in \argmax_{\star \in \supp(x)} \sum_{j \in [n]} \hat y_j A_{\star,j}$} \\
        & = \sum_{j \in [n]} \hat y_j \left( A_{kj}-\sum_{i: x_i>0} x_iA_{ij} \right) \\
        & = \sum_{j \in [n]} \hat y_j \left( A_{kj}-x^\top A e_j \right) \\
        & = \sum_{j \in [n]} \hat y_j A_{kj}- \sum_{j: y_j > 0} \hat y_j (x^\top Ae_j) \tag{$\supp(y)=\supp(\hat y)$} \\
        & = \sum_{j \in [n]} \hat y_j A_{kj}- x^\top A y \tag{\Cref{eq:mat-constr-zero-sum-alpha0-helper}}
    \end{align*}
    Hence, we obtained that, for all $j \in [n]$:
    \begin{align}\label{eq:mat-constr-zero-sum-alpha0-helper2}
        x^\top \hat A e_j = \frac{1}{1+\Delta} \left( x^\top A e_j - x^\top Ay + \max_{k \in \supp(x)} e_k^\top A \hat y  \right) 
    \end{align}

    We are now ready to conclude the proof. Consider $j_1, j_2 \in [n]$ such that $y_{j_1} > 0$ and $y_{j_2} = 0$. Then, combining \Cref{eq:mat-constr-zero-sum-alpha0-helper,eq:mat-constr-zero-sum-alpha0-helper2}, we have that:
    \begin{align*}
        & x^\top \hat A e_{j_1} = \frac{\max_{k \in \supp(x)} e_k^\top A\hat y}{1+\Delta} \\
        & x^\top \hat A e_{j_2} \le \frac{\max_{k \in \supp(x)} e_k^\top A\hat y}{1+\Delta},
    \end{align*}
    thus proving Conditions \eqref{eq:zsg-alpha0-cnd1} and \eqref{eq:zsg-alpha0-cnd2}. Hence, $\hat A \in \mathcal{Z}(x, \hat y)$.
\end{proof}

We now fix $(x,y)$ and for $A \in \mathcal{Z}(x, y)$, we build $\hat A \in \mathcal{Z}(\hat x, y)$. The proof of the following lemma is similar to that of \Cref{lemma:matconstr-1-zsg}, and the proof is reported mainly for the sake of completeness.

Here, $x, \hat x, y$ are again generic strategies on the simplex.

\begin{lemma}[Alternative Matrix Construction (fix y); $\alpha=0$; Zero-Sum Games]\label{lemma:matconstr-2-zsg}
Let $\alpha = 0$, $A \in \mathcal{Z}(x,y)$. It holds that:
\begin{align*}
    & \supp(x) = \supp(\hat x) \implies \exists \hat A \in \mathcal{Z}(\hat x, y): \| A -  \hat A \|_{\infty} \le 4 \| x - \hat x \|_1 \\
\end{align*}
Furthermore, $\hat A$ is defined as follows:
\[
\hat A_{ij} := 
\begin{cases}
\dfrac{A_{ij}}{1+2\|x - \hat x\|_1}, & \text{if } j \in \argmax _{k \in \supp(y)} \sum_{i \in [n]} \hat{x}_i A_{ik},\\[12pt]
\dfrac{A_{ij} + \Delta_j}{1+2\|x - \hat x\|_1}, & \text{if } j \in \supp(y) \setminus \argmax _{k \in \supp(y)} \sum_{i \in [n]} \hat{x}_i A_{ik},\\[12pt]
\dfrac{A_{ij} -2\|x-\hat x\|_1}{1+2\|x - \hat x\|_1}, & \text{if } j \notin \supp(y),
\end{cases}
\]
where, for all $j \in \supp(x)$, $\Delta_j$ is defined as
\begin{align*}
    &\Delta_j = \max_{k \in \supp(y)} \sum_{i \in [n]} \hat x_i A_{ik} - \sum_{i \in[n]} \hat x_i A_{ij}.
\end{align*}
\end{lemma}
\begin{proof}
    First, as in \Cref{lemma:matconstr}, one can trivially verify that, since $A \in [-1,1]^{n \times n}$, $\hat A \in [-1,1]^{n \times n}$ as well. Furthermore, by following analogous arguments, it is also easy to see that $\| A - \hat A \|_{\infty} \le 4 \|x - \hat x \|_1$

    By applying \Cref{lemma:nash-support} it thus remains to verify that:
    \begin{subequations}
    \begin{align}
        & \hat x^\top \hat A e_i = \hat x^\top \hat A e_j,\quad \forall i,j \in \supp(y) \label{eq:mat-constr-zsg-fix-y-eq-1}\\
        & \hat x^\top \hat A e_i \ge \hat x^\top \hat A e_j, \quad \forall i \in \supp(y), j \notin \supp(y)\label{eq:mat-constr-zsg-fix-y-eq-2} \\
        & e_i^\top \hat A y = e_i^\top \hat A e_j, \quad \forall i,j \in \supp(x) \label{eq:mat-constr-zsg-fix-y-eq-3} \\
        & e_i^\top \hat A y \le e_j^\top \hat A y, \quad \forall i \in \supp(x), j \notin \supp(x) \label{eq:mat-constr-zsg-fix-y-eq-4}.
    \end{align}    
    \end{subequations}

    Before continuing, we introduce some additional notation for the sake of brevity. Let $\Delta = 2 \|x - \hat x \|_1$ and define $S=\supp(y)$ and $S^{\star}_y = \argmax_{k \in \supp(y)} \sum_{i \in [n]} \hat x_i A_{ik}$.

    \paragraph{Conditions \eqref{eq:mat-constr-zsg-fix-y-eq-1} and \eqref{eq:mat-constr-zsg-fix-y-eq-2}}
    Given the expression of $\hat A$, \Cref{eq:mat-constr-zsg-fix-y-eq-1} is a direct consequence of \Cref{lemma:matconstr}. We now prove \Cref{eq:mat-constr-zsg-fix-y-eq-2}. Since \Cref{eq:mat-constr-zsg-fix-y-eq-1} holds, we only need to consider $j_1 \in \supp(y)$ and $j_2 \notin \supp(y)$ and show that $\hat x^\top \hat A e_{j_1} \ge \hat x^\top \hat A e_{j_2}$. Specifically, we consider $j_1 \in S^\star_y$. Then, it holds that:

    \begin{align*}
        \hat x^\top \hat A e_{j_1} & = \frac{1}{1+\Delta} \sum_{i \in [n]} \hat x_i A_{i j_1} \\
        & = \frac{1}{1+\Delta} \left( x^\top A e_{j_1} + \sum_{i \in [n]} (\hat x_i - x_i) A_{i j_1}\right) \\
        & \ge \frac{1}{1+\Delta} \left( x^\top A e_{j_2} + \sum_{i \in [n]} (\hat x_i - x_i) A_{i j_1}\right) \tag{$A \in \mathcal{Z}(x,y)$, \Cref{lemma:nash-support} and $j_1 \in \supp(y)$} \\
        & = \frac{1}{1+\Delta} \left( \sum_{i \in [n]} \hat x_i A_{ij_2} + \sum_{i \in [n]} (\hat x_i - x_i) (A_{i j_1} - A_{i j_2}) \right) \\
        & \ge \frac{1}{1+\Delta} \left( \sum_{i \in [n]} \hat x_i A_{i j_2} - 2\|x-\hat x \|_1  \right) \\
        & = \hat x^\top \hat A e_{j_2},
    \end{align*}
    which concludes the proof.

    \paragraph{Conditions \eqref{eq:mat-constr-zsg-fix-y-eq-3} and \eqref{eq:mat-constr-zsg-fix-y-eq-4}} Here, we proceed in a way that is analogous to \Cref{lemma:matconstr-1-zsg}. Specifically, consider any $i \in [n]$. Then, it holds that:
    \begin{align*}
        e_i^\top \hat A y & = \sum_{j \in [n]} y_j \hat A_{ij} \\
        & = \frac{1}{1+\Delta} \left( e_i^\top A y + \sum_{j: y_j > 0} y_j \Delta_j  \right) \tag{Definition of $\hat A $}.
    \end{align*}
    Let us focus on $\sum_{j: y_j > 0} y_j \Delta_j$. We have that:
    \begin{align*}
        \sum_{j: y_j > 0} y_j \Delta_j & = \sum_{j: y_j > 0} y_j \sum_{i \in [n]} \hat x_i (A_{ik}- A_{ij}) \tag{Definition of $\Delta_j$ and $k \in \argmax_{\star \in \supp(y)} \sum_{i \in [n]} \hat x_i A_{i\star}$} \\
        & = \sum_{i \in [n]} \hat x_i \left( A_{ik} - \sum_{j: y_j > 0} y_j A_{ij} \right) \\
        & = \sum_{i \in [n]} \hat x_i (A_{ik} - e_i^\top A y) \\
        & = \sum_{i \in [n]} \hat x_i A_{ik} - \sum_{i: x_i > 0} \hat x_i (e_i^\top A y)) \tag{$\supp(x) = \supp(\hat x)$} \\
        & = \sum_{i \in [n]} \hat x_i A_{ik} - x^\top A y \tag{$A \in \mathcal{Z}(x,y)$}
    \end{align*}
    Therefore, we obtained that, for all $i \in [n]$:
    \begin{align*}
        e_i^\top \hat A y = \frac{1}{1+\Delta} \left( e_i^\top A y - x^\top A y + \max_{k \in \supp(y)} \hat x^\top A e_k \right).
    \end{align*}
    Now, let $i_1\in \supp(\hat x)$ and $i_2 \notin \supp(\hat x)$. We have that:
    \begin{align*}
        & e_{i_1}^\top \hat A y = \frac{1}{1+\Delta} \max_{k \in \supp(y)} \hat x^\top A e_k \\
        & e_{i_2}^\top \hat A y \ge \frac{1}{1+\Delta} \max_{k \in \supp(y)} \hat x^\top A e_k 
    \end{align*}
    which proves \eqref{eq:mat-constr-zsg-fix-y-eq-3} and \eqref{eq:mat-constr-zsg-fix-y-eq-4}, thus showing that $\hat A \in \mathcal{Z}( \hat x, y)$.
\end{proof}

Combining \Cref{lemma:matconstr-1-zsg,lemma:matconstr-2-zsg} with \Cref{corollary:separation-x-y-err}, we have that, whenever $\supp(x)=\supp(\hat x)$ and $\supp(y)=\supp( \hat y)$, then:
\begin{align}\label{eq:zsg-alpha-0-eq-helper}
    H(\mathcal{Z}(x,y), \mathcal{Z}(\hat{x}, \hat{y}))\le 8 (\| x - \hat x\|_1 + \|y - \hat y \|_1)
\end{align}

Hence, following arguments that are identical to the ones presented for \Cref{thm:gsg-alpha0}, we obtain the following result.

\begin{theorem}\label{thm:zsg-alpha0}
    Consider $\alpha = 0$ and let $$m \in \widetilde{\mathcal{O}} \left( \frac{\log\left( \frac{1}{\delta} \right)}{\log(1/(1-\pi_{\min}))} + \frac{n+\log\left(\frac{1}{\delta}\right)}{\epsilon^2} \right).$$ Then \Cref{alg:WAS} is $(\epsilon, \delta)$-correct for Zero-Sum Games and its sample complexity $\tau_\delta$ is given by $m$.    
\end{theorem}
\begin{proof}
Use \Cref{eq:zsg-alpha-0-eq-helper}. Then, the proof follows from the same arguments of \Cref{thm:gsg-alpha0}.    
\end{proof}

\section[Proof of \Cref{thm:ub-any-alpha}]{Proof of \Cref{thm:ub-any-alpha}}\label{app:ub2}

This section is structured as follows: 
\begin{itemize}
    \item In \Cref{app:anyalpha-gsg} we prove the result for the inverse GSG problem (\Cref{thm:ub-any-alpha-gsg-new}) 
    \item In \Cref{app:anyalpha-zsg}, we prove the result of the inverse ZSG problem (\Cref{thm:ub-any-alpha-zsg-new})
\end{itemize}
The proof of \Cref{thm:ub-any-alpha} then follows directly by combining \Cref{thm:ub-any-alpha-gsg-new,thm:ub-any-alpha-zsg-new}.\\

\subsection{Proof of \Cref{thm:ub-any-alpha-gsg-new} ($\alpha > 0$, General-Sum Games)}\label{app:anyalpha-gsg}

\paragraph{Proof outline} We begin our proof by considering two pair of strategies $(x,y)$ and $(\hat x, \hat y)$ and, for any pair of matrix $(A,B) \in \mathcal{G}_\alpha(x,y)$ we construct:
\begin{itemize}
    \item $(\hat A, \hat B) \in \mathcal{G}_\alpha(\hat x, y)$ and $\mathcal{G}_\alpha(x, \hat y)$ such that $\hat A \approx A$ and $\hat B \approx B$
    \item $(\hat A, \hat B) \in \mathcal{G}_\alpha(x, \hat y)$ and $\mathcal{G}_\alpha(x, \hat y)$ such that $\hat A \approx A$ and $\hat B \approx B$
\end{itemize}
This first result is presented in \Cref{lemma:mat-constr-new-any-alpha}. In \Cref{lemma:key-intermediate-step}, we discuss how well these matrices that we constructed approximate the original pair of matrices $(A,B)$. This lemma is the key result that we discussed in \Cref{sec:alg-approx}. Specifically, it shows how these errors scales as, \eg:
\begin{align*}
    \mathcal{O}\left( \frac{1}{\alpha} \sup_{S \in \mathcal{S}(\alpha)} \sum_{i \in S} (\hat x_i - x_i) + \frac{1}{\alpha}\sum_{i: \hat x_i=0} x_i \right).
\end{align*}
Next, using \Cref{corollary:separation-x-y-err}, we will use this result to formally upper bound, with high probability, the Hausdorff distance between $\mathcal{G}_\alpha(x,y)$ and $\mathcal{G}_\alpha(\hat x, \hat y)$. This is done in \Cref{lemma:any-alpha-new-det-h,lemma:any-alpha-new-hauss-high-p}. Finally, \Cref{thm:ub-any-alpha-gsg-new} simply follows by selecting the number of samples in a way such that the high-probability Hausdorff distance upper bound is below $\epsilon$.

We now start with matrix construction argument. 
Note that in the following lemma $x,y,\hat x,\hat y$ are generic strategies on the simplex. 

\begin{lemma}[Matrix construction, $\alpha > 0$, General-Sum Games]\label{lemma:mat-constr-new-any-alpha}
    Let $\alpha > 0$ and $x, y, \hat x, \hat y \in \Delta_n$. Consider $(A,B) \in \mathcal{G}_\alpha(x,y)$. It holds that:
    \begin{align}
        & \exists (\hat A, \hat B) \in \mathcal{G}_\alpha(\hat x, y): \| A- \hat A \|_{\infty} \le 2 \left( \frac{(\hat x-x)^\top A y}{\alpha} \right)_+ \textup{ and } \| B - \hat B\|_{\infty} \le   2\left( \frac{\max_{j \in [n]}(\hat x-x)^\top B(y-e_j)}{\alpha} \right)_+\label{eq:mat-constr-new-any-alpha-eq1} \\
        & \exists (\hat A, \hat B) \in \mathcal{G}_\alpha(x, \hat y): \| A- \hat A \|_{\infty} \le 2 \left( \frac{\max_{i \in [n]} (x - e_i)^\top A(\hat y - y)}{\alpha} \right)_+\textup{ and } \| B - \hat B\|_{\infty} \le 2 \left(\frac{x^\top B(\hat y -y)}{\alpha} \right)_+. \label{eq:mat-constr-new-any-alpha-eq2}
    \end{align}
\end{lemma}

\begin{proof}
    We first prove \Cref{eq:mat-constr-new-any-alpha-eq1}. 

    Specifically, for $(A,B) \in \mathcal{G}_\alpha(x,y)$ we define $(\hat A, \hat B)$ as $(\hat A, \hat B)=(\lambda_A A, \lambda_B B)$, where $\lambda_A, \lambda_B \in \mathbb{R}$ are defined as follows
    \begin{align*}
        \lambda_A = \min \left\{1, \frac{\alpha}{\max_{i \in [n]} \hat x^\top A y - e_i^\top Ay} \right\}, \quad \lambda_B = \min \left\{ 1,  \frac{\alpha}{\max_{j \in [n]} \hat x By - \hat xBe_j}\right\}.
    \end{align*}
    From now on, w.l.o.g., we asume that $\lambda_A, \lambda_B > 0$. Otherwise, we could simply pick $\hat A = A$ and $\hat B = B$ and all the Nash constraints are verified.

    Next, we verify that $(A,B) \in \mathcal{G}_\alpha(\hat x,y)$. First, under the assumption that $\lambda_A, \lambda_B > 0$, it follows that $\lambda_A, \lambda_B \in (0,1]$, and hence $\hat A, \hat B \in [-1,1]^{n \times n}$. Secondly, for any $i \in [n]$:
    \begin{align*}
        \hat x^\top \hat A y - e_i^\top \hat A y &  = \lambda_A \left( \hat x^\top A y - e_i^\top A y \right) \\
        & = \min \left\{1, \frac{\alpha}{\max_{i \in [n]} \hat x^\top A y - e_i^\top Ay} \right\} \left( \hat x^\top A y - e_i^\top A y \right) \\
        & \le \alpha,
    \end{align*}
    where, (i) if $\lambda_A < 1$, the last step is trivial, and (ii) if $\lambda_A = 1$, then, $\alpha > \max_{i \in [n]} \hat x^\top A y - e_i^\top A y$. Thus, we have proved that $\hat A \in \mathcal{G}^x_\alpha(\hat x,y)$. The proof for showing that $\hat B \in \mathcal{G}^y_\alpha(\hat x,y)$ is equivalent.

    Next, we proceed by analyzing $\|A - \hat A \|_{\infty}$. We consider the case where $\lambda_A < 1$, otherwise $A = \hat A$ and \Cref{eq:mat-constr-new-any-alpha-eq1} is trivial. We have that:
    \begin{align*}
        \| A - \hat A \|_{\infty} & \le (1-\lambda_A) \| A \|_{\infty} \\
        & \le 2 (1-\lambda_A) \\
        & \le  2 \frac{\max_{i \in [n]}\hat x^\top A y - e_i^\top A y - \alpha}{\alpha} \\
        & = 2 \frac{(\hat x-x)^\top A y + x^\top A y - \min_{i \in [n]}e_i^\top A y - \alpha}{\alpha} \\
        & \le 2 \frac{(\hat x-x)^\top A y}{\alpha} \tag{$A \in \mathcal{G}_\alpha^x(x,y)$},
    \end{align*}
    which proves the first part of \Cref{eq:mat-constr-new-any-alpha-eq1}. One can follow the same steps to prove the second part of \Cref{eq:mat-constr-new-any-alpha-eq1}.
    
    The proof of \Cref{eq:mat-constr-new-any-alpha-eq2} follows from analogous argument. It is indeed sufficient to define $(\hat A, \hat B) = (\tilde \lambda_A A, \tilde \lambda_B B)$ where:
    \begin{align*}
        \tilde \lambda_A = \min \left\{1, \frac{\alpha}{\max_{i \in [n]}  x^\top A \hat y- e_i^\top A \hat y} \right\}, \quad  \tilde \lambda_B = \min \left\{1, \frac{\alpha}{\max_{j \in [n]} x^\top B \hat y - x^\top B e_j} \right\}.
    \end{align*}
    The proof is thus concluded.
\end{proof}

Next, the following lemma shows how to upper bound the terms that appeared in \Cref{lemma:mat-constr-new-any-alpha}.
We briefly recall that, for any $p \in \Delta^n$,  $\mathcal{S}_\alpha(p)$ is defined as 
\begin{align*}
    & \mathcal{S}_\alpha(p) = \left\{ S \subseteq [n]: \forall i \in S~ p_i > 0 \text{ and } \sum_{i \in S} p_i \le \frac{\alpha}{2} \right\}.
\end{align*}
Also in the following lemma, $x,y,\hat x,\hat y$ are generic strategies on the simplex. 

\begin{lemma}[Key Intermediate Step]\label{lemma:key-intermediate-step}
    Let $\alpha > 0$ and $x,y, \hat x, \hat y \in \Delta_n$.
    The following holds:
    \begin{align}
        & \sup_{A \in \mathcal{G}_\alpha^x(x,y)} \left( (\hat x-x)^\top A y \right)_+ \le 4 \sup_{S \in \mathcal{S}_\alpha(x)} \sum_{i \in S} (\hat x_i - x_i) + 2 \sum_{i: x_i= 0} \hat x_i \label{eq:key-int-step-eq1} \\
        & \sup_{B \in \mathcal{G}_\alpha^y(x,y)} \left( \max_{j \in [n]} (\hat x - x)^\top B (y - e_j) \right)_+ \le 4 \sup_{S \in \mathcal{S}_\alpha(x)} \sum_{i \in S} (\hat x_i - x_i) + 2 \sum_{i: x_i= 0} \hat x_i  \label{eq:key-int-step-eq2} \\
        & \sup_{A \in \mathcal{G}_\alpha^x(x,y)} \left( \max_{i \in [n]}(x-e_i)^\top A (\hat y - y) \right)_+ \le 4 \sup_{S \in \mathcal{S}_\alpha(y)} \sum_{j \in S} (\hat y_j - y_j) + 2 \sum_{j: y_j= 0} \hat y_j  \label{eq:key-int-step-eq3} \\
        & \sup_{B \in \mathcal{G}_\alpha^y(x,y)} \left( x^\top B(\hat y - y) \right)_+ \le 4 \sup_{S \in \mathcal{S}_\alpha(y)} \sum_{j \in S} (\hat y_j - y_j) + 2 \sum_{j: y_j= 0} \hat y_j \label{eq:key-int-step-eq4}.
    \end{align}
\end{lemma}
\begin{proof}
    First, we observe that, due to symmetry, once \Cref{eq:key-int-step-eq1} and \Cref{eq:key-int-step-eq2} hold, \Cref{eq:key-int-step-eq3} and \Cref{eq:key-int-step-eq4} hold as well.

    We now start by proving \Cref{eq:key-int-step-eq1}. We assume w.l.o.g. that the supremum is positive. Furthermore, for any $A \in \mathcal{G}_\alpha(x,y)$, denote by $\star \in [n]$ any index that attains the argmin in $\min_{i \in [n]} (Ay)_i$. Then, we have that:
    \begin{align*}
        \sup_{A \in \mathcal{G}_\alpha^x(x,y)} \left( (\hat x-x)^\top A y \right)_+ & =\sup_{A \in \mathcal{G}_\alpha^x(x,y)} \sum_{i \in [n]} (\hat x_i - x_i) [(Ay)_i - (Ay)_\star] \tag{$\hat x, x \in \Delta_n$} 
        & \\ 
        & \le \left( \sup_{A \in \mathcal{G}_\alpha^x(x,y)} \sum_{i:  x_i >0} (\hat x_i - x_i) [(Ay)_i - (Ay)_\star] \right) +  2 \sum_{i: x_i = 0} \hat x_i  \:.
    \end{align*}
    Let us focus on the remaining optimization problem over $A$. We are interested in solving:\footnote{The sup is attained since it is an optimization problem over a compact domain of a continuous function.}
    \begin{equation}\label{eq:key-int-step-eq5}
    \begin{alignedat}{2}
    & \max_{A \in [-1,1]^{n \times n}} \sum_{i: x_i > 0} (\hat x_i - x_i) [(Ay)_i - (Ay)_\star] \\
    & \text{s.t.~}  \sum_{i: x_i > 0} x_i [(Ay)_i - (Ay)_\star] \le \alpha.
    \end{alignedat}
    \end{equation}
    Let us introduce the following auxiliary variable $c_i = x_i[(Ay)_i-(Ay)_\star] \in[0, 2x_i]$ . Then, \eqref{eq:key-int-step-eq5} is upper bounded by:
    \begin{equation}\label{eq:key-int-step-eq7}
    \begin{alignedat}{2}
    & \max_{(c_1, \dots, c_n)} \sum_{i: x_i > 0} c_i \frac{(\hat x_i - x_i)}{x_i}  \\
    & \text{s.t.~}  \sum_{i: x_i > 0} c_i \le \alpha \text{ and } c_i \in [0, 2x_i].
    \end{alignedat}
    \end{equation}
    This is a fractional knapsack problem. It is well-known that such problems can be solved by a greedy algorithm that sorts the indexes $i$'s according to their value (\ie $\frac{\hat x_i -x_i}{x_i}$) and allocates the maximum value of $c_i$ according to this order until all the budget (\ie $\sum_{i: x_i > 0} c_i \le \alpha$) expires. More formally, any optimal solution $c^\star$ of \eqref{eq:key-int-step-eq7} has the following shape. Let $S \subseteq \{ i \in [n]: x_i > 0 \}$ be any subset such that $$\min_{i \in S}\frac{\hat x_i - x_i}{x_i} \ge \max_{i \in S^\complement}\frac{\hat x_i - x_i}{x_i},$$
    where $S^\complement = \{i \in [n]: x_i > 0 \} \setminus S$. 
    Then $c^\star_i = 2x_i$ for all $i \in S$. Furthermore, there exists at most one $j \in S^\complement$ such that $c^\star_j \in [0, 2x_j]$. For all the remaining elements, we have that $c^\star_i = 0$.
    Hence, the optimal value of \eqref{eq:key-int-step-eq7} can be written (and upper bounded) as follows:
    \begin{align*}
        2 \sum_{i \in S} (\hat x_i - x_i) + \frac{\hat x_j - x_j}{x_j} c^\star_j & \le 4 \sum_{i \in S} ( \hat x_i - x_i ) \\
        & \le 4 \sup_{S \in \mathcal{S}_\alpha(x)} \sum_{i \in S} (\hat x_i - x_i) \tag{Def. of $\mathcal{S}_\alpha(x)$}.
    \end{align*}
    where the first inequality step follows from the fact that, for $j \in S^\complement$ it holds that $c_j \frac{\hat x_j - x_j}{x_j} \le c_i \frac{\hat x_i - x_i}{x_i}$ for any $i \in S$. This concludes the proof of \Cref{eq:key-int-step-eq1}.

    Next, we continue by proving \Cref{eq:key-int-step-eq2}. Again, we assume w.l.o.g. that the sup is greater or equal than $0$. Then, we have that:
    \begin{align*}
        \sup_{B \in \mathcal{G}_\alpha^y(x,y)} \left( \max_{j \in [n]} (\hat x - x)^\top B (y - e_j) \right)_+ & =   \max_{B \in \mathcal{G}_\alpha^y(x,y)}  \max_{\star \in [n]} \sum_{i \in [n]} (\hat x_i - x_i) \left[\sum_{j \in [n]} y_j B_{ij} - B_{i\star}\right] \\
        & \le \max_{B \in \mathcal{G}_\alpha^y(x,y)}  \max_{\star \in [n]} \sum_{i: x_i > 0} (\hat x_i - x_i) \left[\sum_{j \in [n]} y_j B_{ij} - B_{i\star}\right] + 2\sum_{i: x_i = 0} \hat x_i \\
        & = \max_{\star \in [n]} \max_{\substack{B \in \mathcal{G}_\alpha(x,y) \\ \star \in \argmax (\hat x -x)^\top B (y-e_\star)}}  \sum_{i: x_i > 0} (\hat x_i - x_i) \left[\sum_{j \in [n]} y_j B_{ij} - B_{i\star}\right] + 2\sum_{i: x_i = 0} \hat x_i \\
        & \le \max_{\star \in [n]} \max_{\substack{B \in \mathcal{G}_\alpha(x,y)}}  \sum_{i: x_i > 0} (\hat x_i - x_i) \left[\sum_{j \in [n]} y_j B_{ij} - B_{i\star}\right] + 2\sum_{i: x_i = 0} \hat x_i\:.
    \end{align*}
    Hence, fix any $\star \in [n]$ in the outer maximization problem, and restrict our attention to the maximization problem over $B$, that is 
    \begin{equation}\label{eq:key-int-step-eq9}
    \begin{alignedat}{2}
    & \max_{B \in [-1,1]^{n \times n}} \sum_{i: x_i > 0} (\hat x_i - x_i) \left[ \sum_{j \in [n]} y_j B_{ij} - B_{i\star} \right]\\
    & \text{s.t.~}  \sum_{i: x_i > 0} x_i \left[ \sum_{j \in [n]} y_j B_{ij} - B_{ik}\right] \le \alpha \quad \forall k \in [n].
    \end{alignedat}
    \end{equation}
    Since the $\alpha$-Nash constraint holds for all $k \in [n]$, it holds also for $\star$. Hence, we can upper bound \eqref{eq:key-int-step-eq9} with the following problem:
    \begin{equation}\label{eq:key-int-step-eq10}
    \begin{alignedat}{2}
    & \max_{B \in [-1,1]^{n \times n}} \sum_{i: x_i > 0} (\hat x_i - x_i) \left[ \sum_{j \in [n]} y_j B_{ij} - B_{i\star} \right]\\
    & \text{s.t.~}  \sum_{i: x_i > 0} x_i \left[ \sum_{j \in [n]} y_j B_{ij} - B_{i\star}\right] \le \alpha 
    \end{alignedat}
    \end{equation}
    Introducing the variable $c_i = x_i \left[ \sum_{j \in [n]} y_j B_{ij} - B_{i\star}\right]$ and rewriting the objective function as $\sum_{i: x_i > 0} c_i \frac{\hat x_i - x_i}{x_i}$, we obtain again a fractional knapsack problem that can be solved in closed form, as we presented above. Following the same steps as above, the proof is concluded.
\end{proof}

Here, we show how to exploit the previous results to upper bound the Hausdorff distance. Here, we will use the fact that $\hat x$ and $\hat y$ are the empirical mean of $x$ and $y$ after $m$ samples.

\begin{lemma}[Matrix Construction $\to$ Hausdorff Distance]\label{lemma:any-alpha-new-det-h}
    Let $\alpha > 0$ and $x,y, \hat x, \hat y \in \Delta_n$. Specifically, $\hat x$ and $\hat y$ are the maximum likelihood estimators of $x,y$ after $m$ samples.
    Let $f_\alpha(x, \hat x)$ and $f_\alpha(y, \hat y)$ be defined as:
    \begin{align*}
        & f_\alpha(x, \hat x) =  \frac{16}{\alpha} \left(   \sup_{S \in \mathcal{S}_\alpha(x)} \sum_{i \in S} (\hat x_i - x_i) + \sup_{S \in \mathcal{S}_\alpha(\hat x)} \sum_{i \in S} ( x_i - \hat x_i) +  \sum_{i: \hat x_i= 0}  x_i\right) \\
        & f_\alpha(y, \hat y) = \frac{16}{\alpha} \left(   \sup_{S \in \mathcal{S}_\alpha(y)} \sum_{j \in S} (\hat y_j - y_j) + \sup_{S \in \mathcal{S}_\alpha(\hat y)} \sum_{j \in S} ( y_j - \hat y_j) +  \sum_{j: \hat y_j= 0}  y_j\right).
    \end{align*}
    It holds that:
    \begin{align*}
        H(\mathcal{G}_\alpha(x, y), \mathcal{G}_\alpha(\hat x, \hat y) \le f_\alpha(x, \hat x) + f_\alpha(y, \hat y).
    \end{align*}
\end{lemma}
\begin{proof}
    From \Cref{corollary:separation-x-y-err}, we have that:
    \begin{align*}
        H(\mathcal{G}_\alpha(x, y), \mathcal{G}_\alpha(\hat x, \hat y) & \le H(\mathcal{G}_\alpha(x, y), \mathcal{G}_\alpha(\hat x,  y) + H(\mathcal{G}_\alpha(\hat x, y), \mathcal{G}_\alpha(\hat x, \hat y) 
    \end{align*}
    Let us focus on the fist term. From \Cref{lemma:mat-constr-new-any-alpha}, we have that:
    \begin{align*}
        H(\mathcal{G}_\alpha(x, y), \mathcal{G}_\alpha(\hat x,  y) \le 2 \sup_{A \in \mathcal{G}^x_\alpha(x,y)}\left( \frac{(\hat x-x)^\top A y}{\alpha} \right)_+ & + 2 \sup_{A \in \mathcal{G}^x_\alpha(\hat x,y)}\left( \frac{(x-\hat x)^\top A y}{\alpha} \right)_+  +\\
        &+ \sup_{A \in \mathcal{G}^y_\alpha(x, y)}2\left( \frac{\max_{j \in [n]}(\hat x-x)^\top B(y-e_j)}{\alpha} \right)_+ + \\
        & + \sup_{A \in \mathcal{G}^y_\alpha(\hat x, y)}2\left( \frac{\max_{j \in [n]}(x-\hat x)^\top B(y-e_j)}{\alpha} \right)_+
    \end{align*}
    This, in turn, can be upper bounded using \Cref{lemma:key-intermediate-step}. We thus obtain:
    \begin{align*}
        H(\mathcal{G}_\alpha(x, y), \mathcal{G}_\alpha(\hat x,  y) \le \frac{16}{\alpha} \left(   \sup_{S \in \mathcal{S}_\alpha(x)} \sum_{i \in S} (\hat x_i - x_i) + \sup_{S \in \mathcal{S}_\alpha(\hat x)} \sum_{i \in S} ( x_i - \hat x_i) +  \sum_{i: \hat x_i= 0}  x_i\right),
    \end{align*}
    where we used the fact that $\sum_{x_i = 0} \hat x_i  = 0$.
    Similarly, for $H(\mathcal{G}_\alpha(\hat x, y), \mathcal{G}_\alpha(\hat x, \hat y)$, we obtain:
    \begin{align*}
        H(\mathcal{G}_\alpha(\hat x, y), \mathcal{G}_\alpha(\hat x, \hat y)  \le  \frac{16}{\alpha} \left(   \sup_{S \in \mathcal{S}_\alpha(y)} \sum_{j \in S} (\hat y_j - y_j) + \sup_{S \in \mathcal{S}_\alpha(\hat y)} \sum_{j \in S} ( y_j - \hat y_j) +  \sum_{j: \hat y_j= 0}  y_j\right), 
    \end{align*}
    thus concluding the proof.
\end{proof}

The following lemma shows how to upper bound with probabilistic argument the terms that appeared in the previous lemma.

\begin{lemma}[Haussdorf High-Probability upper bound]\label{lemma:any-alpha-new-hauss-high-p}
    Let $m$ be such that:
    \begin{align}\label{eq:any-alpha-new-m-req}
        m \ge 1+ \frac{4 \left(n+\log\left( \frac{24}{\delta} \right)\right)}{\alpha}.
    \end{align}
    With probability at least $1-\delta$ it holds that:
    \begin{align*}
        H(\mathcal{G}_\alpha(x, y), \mathcal{G}_\alpha(\hat x, \hat y) \le \frac{64}{\alpha} \left(\sqrt{\frac{4\alpha \left( n+\log\left( \frac{6}{\delta} \right) \right) }{m}} + \frac{5(n+\log\left( \frac{6}{\delta} \right))}{m} +\frac{3\sqrt{n} \log\left(\frac{6}{\delta} \right)}{m}\right).
    \end{align*}
\end{lemma}
\begin{proof}
    We first fix any $m \ge 1$ and we consider the following good event $$\mathcal{E} = \mathcal{E}_{\textup{miss}} \cap \mathcal{E}_{\textup{set}} \cap \mathcal{E}_{\text{err},x} \cap \mathcal{E}_{\text{err},y}$$ where these events are defined as:
    \begin{align*}
        & \mathcal{E}_{\text{miss}} = \left\{ \sum_{i: \hat x_i=0} x_i - \mathbb{E}\left[ \sum_{i: \hat x_i = 0} x_i \right] \le  \frac{3\sqrt{n} \log\left(\frac{6}{\delta} \right)}{m}  \right\} \bigcap  \left\{ \sum_{j: \hat y_j=0} y_j - \mathbb{E}\left[ \sum_{j: \hat y_j = 0} y_j \right] \le \frac{3\sqrt{n} \log\left(\frac{6}{\delta} \right)}{m}  \right\} \\
        & \mathcal{E}_{\text{set},x} = \bigcap_{S \subseteq[n]} \left\{ \Bigg|\sum_{i \in S} (\hat x_i - x_i) \Bigg| \le \sqrt{\frac{2 \sum_{i \in S} \hat x_i\left(n+\log(\frac{24}{\delta}) \right)}{m}} + \frac{4 (n + \log(\frac{24}{\delta}))}{m-1} \right\} \\ 
        & \mathcal{E}_{\text{set},y} = \bigcap_{S \subseteq[n]} \left\{ \Bigg|\sum_{j \in S} (\hat y_j - y_j) \Bigg| \le \sqrt{\frac{2 \sum_{j \in S} \hat y_j\left(n+\log(\frac{24}{\delta}) \right)}{m}} + \frac{4 (n + \log(\frac{24}{\delta}))}{m-1} \right\} \\ 
        & \mathcal{E}_{\text{err},x} = \bigcap_{S \in \mathcal{S}_{2\alpha}(x)} \left\{ \Bigg|\sum_{i \in S} (\hat x_i - x_i) \Bigg|  \le \sqrt{\frac{4\alpha \left( n+\log\left( \frac{6}{\delta} \right) \right) }{m}} + \frac{4(n+\log\left( \frac{6}{\delta} \right))}{m} \right\}\\
        & \mathcal{E}_{\text{err},y} = \bigcap_{S \in \mathcal{S}_{2\alpha}(y)} \left\{ \Bigg|\sum_{j \in S} (\hat y_j - y_j) \Bigg|  \le \sqrt{\frac{4\alpha \left( n+\log\left( \frac{6}{\delta} \right) \right) }{m}} + \frac{4(n+\log\left( \frac{6}{\delta} \right))}{m}  \right\}
    \end{align*}

    We now continue by showing that $\mathbb{P}(\mathcal{E}) \ge 1-\delta$. In particular, using a union bound, we obtain that $\Prob(\mathcal{E}) = 1- \Prob(\mathcal{E}^\complement) \ge 1 - \left(\Prob\left(\mathcal{E}_{\text{miss}}^\complement\right) + \Prob\left(\mathcal{E}_{\text{set},x}^\complement\right) +  \Prob\left(\mathcal{E}_{\text{set},y}^\complement\right) + \Prob\left(\mathcal{E}_{\text{err},x}^\complement\right) + \Prob\left(\mathcal{E}_{\text{err},y}^\complement\right) \right) $.
    Hence, we continue by upper bounding the complement of each event that compose $\mathcal{E}$. We start from $\mathbb{P}(\mathcal{E}^\complement_{\text{miss}})$. We have that:
    \begin{align*}
        \mathbb{P}(\mathcal{E}^\complement_{\text{miss}}) \le \mathbb{P}\left( \sum_{i: \hat x_i=0} x_i - \mathbb{E}\left[ \sum_{i: \hat x_i = 0} x_i \right] \ge  \frac{\sqrt{n} \log\left(\frac{6}{\delta} \right)}{m} \right) + \mathbb{P}\left( \sum_{i: \hat y_i=0} y_i - \mathbb{E}\left[ \sum_{i: \hat y_i = 0} y_i \right] \ge  \frac{\sqrt{n} \log\left(\frac{6}{\delta} \right)}{m} \right) \le \frac{\delta}{3}.
    \end{align*}
    
    where the last inequality is a direct application of concentration results for the missing mass, \ie \Cref{lemma:missing-mass}.

    Next, we analyze $\Prob(\mathcal{E}_{\text{set},x}^\complement)$ and $\Prob(\mathcal{E}_{\text{set},y}^\complement)$. Using \Cref{lemma:var-aware}, we have that:
    \begin{align*}
        \Prob(\mathcal{E}_{\text{set},x}^\complement) = \mathbb{P}\left(\exists S \subseteq [n]: \Bigg|\sum_{i \in S} (\hat x_i - x_i) \Bigg| > \sqrt{\frac{2 \sum_{i \in S} \hat x_i\left(n+\log(\frac{24}{\delta}) \right)}{m}} + \frac{4 (n + \log(\frac{24}{\delta}))}{m-1} \right) \le \frac{\delta}{6}.
    \end{align*}
    
    Finally, it remains to analyze $\mathbb{P}(\mathcal{E}_{\text{err},x}^\complement)$ and $\mathbb{P}(\mathcal{E}_{\text{err},y}^\complement)$. We present our argument for $\mathbb{P}(\mathcal{E}_{\text{err},x}^\complement)$ and the computations for $\mathbb{P}(\mathcal{E}_{\text{err},y}^\complement)$ follows from analogous arguments. Specifically, using \Cref{lemma:var-aware}, we have that: 
    \begin{align*}
        \mathbb{P}(\mathcal{E}_{\text{err},x}^\complement) = \mathbb{P}\left( \exists S \in \mathcal{S}(2\alpha): \Bigg| \sum_{i \in S}(\hat x_i - x_i) \Bigg|  \ge \sqrt{\frac{4\alpha \left( n+\log\left( \frac{6}{\delta} \right) \right) }{m}} + \frac{4(n+\log\left( \frac{6}{\delta} \right))}{m} \right) \le \frac{\delta}{6}.
    \end{align*}
        Hence, we obtained that $\mathbb{P}(\mathcal{E}) \ge 1-\delta$. 

    At this point, we proceed by upper bounding the Hausdorff distance under the good event $\mathcal{E}$. We start by applying \Cref{lemma:any-alpha-new-det-h} and we analyze $f_\alpha(x, \hat x)$.\footnote{The arguments for upper bounding $f_\alpha(y, \hat y)$ are symmetric.} We recall that:
    \begin{align*}
        f_\alpha(x, \hat x) & =  \frac{16}{\alpha} \left(   \sup_{S \in \mathcal{S}_\alpha(x)} \sum_{i \in S} (\hat x_i - x_i) + \sup_{S \in \mathcal{S}_\alpha(\hat x)} \sum_{i \in S} ( x_i - \hat x_i) +  \sum_{i: \hat x_i= 0}  x_i\right)
    \end{align*}
    At this point, let us focus on $\mathcal{S}_\alpha(\hat x)$. With probability at least $1-\delta$ it holds that:
    \begin{align*}
        \mathcal{S}_\alpha(\hat x) & = \{S \subseteq [n]: \hat x_i > 0~\forall i \in S \text{ and } \sum_{i \in S} \hat x_i \le \tfrac\alpha2 \} \\
        & \subseteq \{S \subseteq [n]:  x_i > 0~\forall i \in S \text{ and } \sum_{i \in S} \hat x_i \le \tfrac\alpha2 \} \tag{$\hat x_i > 0\implies x_i > 0$} \\
        & = \left\{S \subseteq [n]:  x_i > 0~\forall i \in S \text{ and } \sum_{i \in S} \hat x_i \le \tfrac{\alpha}{2} \text{ and }\sum_{i \in S}  x_i \le \tfrac\alpha2 + \Big|\sum_{i \in S} x_i - \hat x_i\Big|  \right\} \\
        & \subseteq \left\{S \subseteq [n]:  x_i > 0~\forall i \in S \text{ and } \sum_{i \in S} \hat x_i \le \tfrac{\alpha}{2} \text{ and }\sum_{i \in S}  x_i \le \tfrac\alpha2 + \sqrt{\frac{2 \sum_{i \in S} \hat x_i (n+\log(\frac{24}{\delta}))}{m}} + \frac{4 (n+\log(\frac{24}{\delta}))}{m-1}  \right\} \\ 
        & = \left\{S \subseteq [n]:  x_i > 0~\forall i \in S \text{ and } \sum_{i \in S} \hat x_i \le \tfrac{\alpha}{2} \text{ and }\sum_{i \in S}  x_i \le \tfrac\alpha2 + \sqrt{\frac{\alpha (n+\log(\frac{24}{\delta}))}{m}} + \frac{4 (n+\log(\frac{24}{\delta}))}{m-1}  \right\} \\
        & \subseteq \left\{S \subseteq [n]:  x_i > 0~\forall i \in S \text{ and } \sum_{i \in S}  x_i \le \tfrac\alpha2 + \sqrt{\frac{\alpha (n+\log(\frac{24}{\delta}))}{m}} + \frac{4 (n+\log(\frac{24}{\delta}))}{m-1}  \right\} \\
        & \subseteq \{S \subseteq [n]:  x_i > 0~\forall i \in S \text{ and } \sum_{i \in S}  x_i \le \alpha   \} \tag{\Cref{eq:any-alpha-new-m-req}}\\
        & = S_{2\alpha}(x).
    \end{align*}
    where (i) the equality steps in the middle of these arguments follows from the fact that, with high probability, $\sum_{i \in S} \hat x_i \le \frac{\alpha}{2}$ implies the additional conditions that are mentioned in the formula, and (ii) the last $\subseteq$ step follows rom the fact that, if $m$ satisfies \Cref{eq:any-alpha-new-m-req}, then:
    \begin{align*}
        \sqrt{\frac{\alpha (n+\log(\frac{24}{\delta})}{m}} + \frac{4 \log(\frac{24}{\delta})}{m-1} \le \frac{\alpha}{2}.
    \end{align*}
    
    Plugging this result within $f_\alpha(x, \hat x)$, we have that, with probability at least $1-\delta$, 
    \begin{align*}
        f_\alpha(x, \hat x) & \le \frac{32}{\alpha} \left( \sup_{S \in \mathcal{S}_{2\alpha}(x)} \Bigg| \sum_{i \in S} \hat x_i - x_i \Bigg| +  \sum_{i: \hat x_i= 0}  x_i\right) \\
        & \le \frac{32}{\alpha} \left(\sqrt{\frac{4\alpha \left( n+\log\left( \frac{6}{\delta} \right) \right) }{m}} + \frac{4(n+\log\left( \frac{6}{\delta} \right))}{m} + \sum_{i: \hat x_i= 0}  x_i  \right) \tag{Def. of $\mathcal{E}$} \\
        & \le \frac{32}{\alpha} \left(\sqrt{\frac{4\alpha \left( n+\log\left( \frac{6}{\delta} \right) \right) }{m}} + \frac{4(n+\log\left( \frac{6}{\delta} \right))}{m} +\frac{3\sqrt{n} \log\left(\frac{6}{\delta} \right)}{m} + \mathbb{E}\left[ \sum_{i: \hat x_i= 0}  x_i  \right]\right).\tag{Def. of $\mathcal{E}$}
    \end{align*}
    It thus remains to control $\mathbb{E}[\sum_{i: \hat x_i=0} x_i]$. It holds that:
    \begin{align*}
        \mathbb{E}\left[\sum_{i: \hat x_i=0} x_i\right] = \sum_{i \in [n]} x_i \mathbb{P}[\hat x_i = 0] = \sum_{i \in [n]} x_i (1-x_i)^m \le n \max_{p \in [0,1]}p(1-p)^m \le \frac{n}{m} .
    \end{align*}
    Following similar steps, we arrived at the following high-probability bound:
    \begin{align*}
    H(G_\alpha(x,y), \mathcal{G}_\alpha(\hat x, \hat y) & \le f_\alpha(x, \hat x) + f_\alpha(y, \hat y)  \\ & \le  \frac{64}{\alpha} \left(\sqrt{\frac{4\alpha \left( n+\log\left( \frac{6}{\delta} \right) \right) }{m}} + \frac{5(n+\log\left( \frac{6}{\delta} \right))}{m} +\frac{3\sqrt{n} \log\left(\frac{6}{\delta} \right)}{m}\right),
    \end{align*}
    which concludes the proof.
\end{proof}

We are now ready to prove \Cref{thm:ub-any-alpha-gsg-new}.

\begin{theorem}\label{thm:ub-any-alpha-gsg-new}
    Let $\alpha > 0$ and let 
    $$m \in {\mathcal{O}}\left( \frac{n+\log(\frac{1}{\delta})}{\alpha \epsilon^2} + \frac{\sqrt{n} \log\left( \frac{1}{\delta} \right)}{\alpha \epsilon} \right).$$
Then, \Cref{alg:WAS} is $(\epsilon, \delta)$-correct for General-Sum Games and its sample complexity $\tau_\delta$ is given by $m$.
\end{theorem}
\begin{proof}
    Let $m_1, m_2, m_3, m_4$ be defined as follows:
    \begin{align*}
        & m_1= 1+ \frac{4 \left(n+\log\left( \frac{24}{\delta} \right)\right)}{\alpha} \\
        & m_2 = \frac{16 (192)^2 (n+\log(6/\delta))}{\alpha \epsilon^2} \\
        & m_3 = \frac{960 (n+\log(6/\delta))}{\alpha \epsilon} \\
        & m_4 = \frac{576 \sqrt{n} \log(6/\delta)}{\epsilon \alpha}. 
    \end{align*}
    Finally, let $m = \lceil \max \{m_1, m_2, m_3, m_4 \}\rceil$. 
    Since $m \ge m_1$, we can apply \Cref{lemma:any-alpha-new-hauss-high-p}, and using the definition of $m$, one can verify that, with probability at least $1-\delta$, \footnote{Here, we note that $m \ge m_i$ imposes that the $i+1$-th error term in the statement of \Cref{lemma:any-alpha-new-hauss-high-p} is bounded by $\tfrac{\epsilon}{3}$.}
    \begin{align*}
        H(\mathcal{G}_\alpha(x,y), \mathcal{G}_\alpha(\hat x, \hat y) ) \le \epsilon.
    \end{align*}
    thus concluding the proof.
\end{proof}

\subsection{Proof of \Cref{thm:ub-any-alpha-zsg-new} ($\alpha > 0$, Zero-Sum Games)}\label{app:anyalpha-zsg}

\paragraph{Proof outline }
The proof of \Cref{thm:ub-any-alpha-zsg-new} is in the same spirit as that of \Cref{thm:ub-any-alpha-gsg-new}. Specifically, first, for any pair of strategies $(x,y)$ and $(\hat x, \hat y)$, we fix a matrix $A \in \mathcal{Z}_\alpha(x,y)$ and we construct a pair of matrices $(\hat A_1, \hat A_2) \in \mathcal{Z}_\alpha(\hat x, y) \times \mathcal{Z}_\alpha(x, \hat y)$ (\Cref{lemma:mat-constr-new-any-alpha-zsg}). Remarkably, $\hat A_1 \approx A $ and $\hat A_2 \approx A$ with error rates that are similar to those that we presented for the General-Sum Setting. Then, in \Cref{lemma:key-intermediate-step-zsg}, we upper bound these errors by heavily exploiting the key intermediate step that we presented for GSG's (\ie \Cref{lemma:key-intermediate-step}). Once this result is obtained, one can follow the same exact proof for General-Sum as for Zero-Sum Games.

\begin{lemma}[Matrix construction, $\alpha > 0$, Zero-Sum Games]\label{lemma:mat-constr-new-any-alpha-zsg}
    Let $\alpha > 0$ and $x, y, \hat x, \hat y \in \Delta_n$. Consider $A \in \mathcal{Z}_\alpha(x,y)$. It holds that:
    \begin{align}
        & \exists \hat A \in \mathcal{Z}_\alpha(\hat x, y): \| A- \hat A \|_{\infty} \le 2 \left( \max \left\{ \frac{(\hat x - x)^\top A y}{\alpha}, \frac{\max_{j \in [n] } (\hat x - x)^\top A (e_j-y) }{\alpha} \right\}  \right)_+
        \label{eq:mat-constr-new-any-alpha-eq1-zsg} \\
        & \exists \hat A \in \mathcal{Z}_\alpha(x, \hat y): \| A - \hat A \|_{\infty} \le 2 \left( \max \left\{ \frac{\max_{i \in [n] } (x  - e_i)^\top A (\hat y - y) }{\alpha}, \frac{x^\top A (y - \hat y) }{\alpha} \right\}  \right)_+. \label{eq:mat-constr-new-any-alpha-eq2-zsg}
    \end{align}
\end{lemma}
\begin{proof}
    We first prove \Cref{eq:mat-constr-new-any-alpha-eq1-zsg}. 

    Let $A \in \mathcal{Z}_\alpha(x,y)$. We define $\hat A = \lambda A$ where $\lambda \in \mathbb{R}$ is given by:
    \begin{align*}
        \lambda = \min\left\{ 1, \frac{\alpha}{\max \{ \max_{i \in [n]} \hat x^\top  Ay - e_i^\top A y , \max_{j \in [n]} \hat x^\top A e_j - \hat x^\top A y\}} \right\}.
    \end{align*}
    In the following, we assume w.l.o.g., that $\lambda \ge 0$. If $\lambda < 0$, one could pick $\hat A = A$ and all the equilibrium constraints would be satisfied. 

    Now, we verify that $A \in \mathcal{Z}_\alpha(\hat x , y)$. First, since $\lambda \in (0, 1]$, we have that $\hat A \in [-1,1]^{n \times n}$. Secondly, for all $i \in [n]$ it holds that:
    \begin{align*}
        \hat x^\top \hat A y - e_i^\top \hat A y & = \lambda (  \hat x^\top A y - e_i^\top A y) \\
        & =  \min\left\{ 1, \frac{\alpha}{\max \{ \max_{i \in [n]} \hat x^\top  Ay - e_i^\top A y , \max_{j \in [n]} \hat x^\top A e_j - \hat x^\top A y\}} \right\} (  \hat x^\top A y - e_i^\top A y) \\
        & \le \alpha,
    \end{align*}
    where, (i) if $\lambda = 1$, we have that $\max_{i \in [n]} \hat x^\top A y - e_i^\top A y \le \alpha $, and (ii) if $\lambda < 1$,  the last step is trivial.
    Following similar arguments, one can prove that $\hat x^\top \hat A e_j - \hat x^\top \hat A y  \le \alpha$ for all $j \in [n]$. Hence, $\hat A \in \mathcal{Z}_\alpha(\hat x, y)$. 


    Next, we proceed by analyzing $\| A - \hat A\|_{\infty}$. We consider the case where $\lambda < 1$, otherwise $A = \hat A$ and \Cref{eq:mat-constr-new-any-alpha-eq1-zsg} is direct. We have that:
    \begin{align*}
        \| A - \hat A \|_{\infty} & \le (1-\lambda) \| A \|_{\infty} \\
        & \le 2(1-\lambda) \\
        & \le 2 \frac{\max \{ \max_{i \in [n]} \hat x^\top  Ay - e_i^\top A y , \max_{j \in [n]} \hat x^\top A e_j - \hat x^\top A y\} - \alpha}{\alpha} \\
        & \le 2 \max \left\{ \frac{\max_{i \in [n] }\hat x^\top A y - e_i^\top A y - \alpha}{\alpha}, \frac{\max_{j \in [n] } \hat x^\top A e_j - \hat x^\top A y- \alpha}{\alpha} \right\} \\
        & \le 2 \max \left\{ \frac{(\hat x - x)^\top A y}{\alpha}, \frac{\max_{j \in [n] } (\hat x - x)^\top A (e_j-y) }{\alpha} \right\} \tag{$A \in \mathcal{Z}_\alpha(x,y)$},
    \end{align*}
    thus concluding the proof of \Cref{eq:mat-constr-new-any-alpha-eq1-zsg}.

    The proof of \Cref{eq:mat-constr-new-any-alpha-eq2-zsg} follows from analogous arguments. One has simply to define $\hat A = \tilde \lambda A$ where $\tilde \lambda$ is given by:
    \begin{align*}
        \tilde \lambda = \min\left\{ 1, \frac{\alpha}{\max \{ \max_{i \in [n]} x^\top A \hat y - e_i^\top A \hat y, \max_{j \in [n]} x^\top A e_j - x^\top A \hat y \}} \right\}.
    \end{align*}
\end{proof}

The following results extend the key intermediate that we presented for General-Sum Games (\ie \Cref{lemma:key-intermediate-step}) to the Zero-Sum Games problem. The proof exploits a connection between the quantities that arise in \Cref{lemma:mat-constr-new-any-alpha-zsg} and those that we already upper-bounded in \Cref{lemma:key-intermediate-step}.

\begin{lemma}[Key Intermediate Step (Zero-Sum Games)]\label{lemma:key-intermediate-step-zsg}
Let $\alpha > 0$ and $x, y, \hat x, \hat y \in \Delta_n$. The following holds:
\begin{align}
    & \sup_{A \in \mathcal{Z}_\alpha(x,y)} (\hat x - x)^\top A y \le 4 \sup_{S \in \mathcal{S}_\alpha(x)} \sum_{i \in S} (\hat x_i - x_i) + 2 \sum_{i: x_i= 0} \hat x_i  \label{eq:key-int-zsgs-eq-1} \\
    & \sup_{A \in \mathcal{Z}_\alpha(x,y)}\max_{j \in [n] } (\hat x - x)^\top A (e_j-y) \le 4 \sup_{S \in \mathcal{S}_\alpha(x)} \sum_{i \in S} (\hat x_i - x_i) + 2 \sum_{i: x_i= 0} \hat x_i \label{eq:key-int-zsgs-eq-2} \\
    & \sup_{A \in \mathcal{Z}_\alpha(x,y) }{\max_{i \in [n] } (x  - e_i)^\top A (\hat y - y) } \le 4 \sup_{S \in \mathcal{S}_\alpha(y)} \sum_{j \in S} (\hat y_j - y_j) + 2 \sum_{j: y_j= 0} \hat y_j \label{eq:key-int-zsgs-eq-3} \\
    & \sup_{A \in \mathcal{Z}_\alpha(x,y) } x^\top A (y - \hat y) \le 4 \sup_{S \in \mathcal{S}_\alpha(y)} \sum_{j \in S} (\hat y_j - y_j) + 2 \sum_{j: y_j= 0} \hat y_j  \label{eq:key-int-zsgs-eq-4}
\end{align}
\end{lemma}
\begin{proof}
    We start by proving \Cref{eq:key-int-zsgs-eq-1}. Specifically, we have that:
    \begin{align*}
        \sup_{A \in \mathcal{Z}_\alpha(x,y)} (\hat x - x)^\top A y & \le \sup_{A \in \mathcal{G}^x_\alpha(x,y)} (\hat x - x)^\top A y \\
        & \le 4 \sup_{S \in \mathcal{S}_\alpha(x)} \sum_{i \in S} (\hat x_i - x_i) + 2 \sum_{i: x_i= 0} \hat x_i \tag{\Cref{lemma:key-intermediate-step}}, 
    \end{align*}
    where in the first inequality we have simply dropped the constraints $x^\top A e_j - x^\top A y \le \alpha$ for all $j \in [n]$.
    We can follow the same argument (\ie simply drop the constraints) also for proving \Cref{eq:key-int-zsgs-eq-3}.

    Next, we continue with \Cref{eq:key-int-zsgs-eq-2}. We have that:
    \begin{align*}
        \sup_{A \in \mathcal{Z}_\alpha(x,y)}\max_{j \in [n] } (\hat x - x)^\top A (e_j-y) \le \sup_{B \in \mathcal{G}_\alpha^y(x,y)} \max_{j \in [n] } (\hat x - x)^\top B (y-e_j) \\
        \sup_{S \in \mathcal{S}_\alpha(x)} \sum_{i \in S} (\hat x_i - x_i) + 2 \sum_{i: x_i= 0} \hat x_i \tag{\Cref{lemma:key-intermediate-step}},
    \end{align*}
    where, in the first step, we have dropped the constraints $x^\top A y - e_i^\top A y \le \alpha$ and we applied a change of variable $B = -A$. We can use the same argument (\ie dropping constraints and change of variable) also for proving \Cref{eq:key-int-zsgs-eq-4}.
\end{proof}

We can now prove \Cref{thm:ub-any-alpha-zsg-new}

\begin{theorem}\label{thm:ub-any-alpha-zsg-new}
    Let $\alpha > 0$ and let 
    $$m \in {\mathcal{O}}\left( \frac{n+\log\left(\frac{1}{\delta}\right)}{\alpha \epsilon^2} + \frac{\sqrt{n} \log\left( \frac{1}{\delta} \right)}{\alpha\epsilon}  \right).$$
Then, \Cref{alg:WAS} is $(\epsilon, \delta)$-correct for Zero-Sum Games, and its sample complexity $\tau_\delta$ is given by $m$.
\end{theorem}
\begin{proof}
    Combining \Cref{lemma:mat-constr-new-any-alpha-zsg} and \Cref{lemma:key-intermediate-step-zsg}, we can upper bound the Hausdorff distance $H(\mathcal{Z}_\alpha(x,y), \mathcal{Z}_\alpha(\hat x, \hat y))$ with the same quantities that we obtained in General Sum Games (\ie those appearing in \Cref{lemma:any-alpha-new-det-h}). Hence, we can proceed as in the proof of General-Sum Games, and we obtain the same result.
\end{proof}

\section{Auxiliary Lemmas}

In this section, we provide statements of auxiliary lemma that are used throughout our analysis.

\subsection{Concentrations and Probabilities}

\begin{lemma}\label{lemma:support}
    Consider $D \in \mathbb{N}_{>0}$ and $p \in \Delta_{D}$. Let $\hat p_m$ be the maximum likelihood estimate of $p$ obtained with $m \ge 1$ samples. 
    Consider a threshold level $\beta > 0$ and let $\mathcal{I}(\beta) = \{ i \in [D]: x_i \ge \beta \}$. 
    If $m \ge \frac{\log\left(\frac{D}{\delta}\right)}{\log\left(\frac{1}{1-\beta}\right)}$, then, with probability at least $1-\delta$, it holds that:
    \begin{align*}
        \forall i \in \mathcal{I}(\beta), N_i(m) \ge 1.
    \end{align*}
    As a corollary, let $p_{\min} \in \mathbb{R}$ as any real number that satisfies $p_{\min} > 0$ and   $p_{\min} \le \min_{i \in [D]: p_i > 0} p_i$. Then, if $m \ge \frac{\log\left( \frac{D}{\delta} \right)}{\log\left( \frac{1}{1-p_{\min}} \right)}$ it holds that 
    \begin{align*}
        \mathbb{P}\left( \supp(p) \ne \supp(\hat p_m) \right) \le \delta. 
    \end{align*}
\end{lemma}
\begin{proof}
    With probabilistic arguments, we have that:
    \begin{align*}
        \mathbb{P}\left( \exists i \in \mathcal{I}(\beta): N_i(m) = 0 \right) & \le \sum_{i \in \mathcal{I}(\beta)} \mathbb{P}(N_i(m) = 0) \tag{Boole's inequality} \\
        & \le \sum_{i: x_i \ge \beta} (1-x_i)^m \\
        & \le D (1-\beta)^m \\
        & \le \delta \tag{$m \ge \frac{\log \left( \frac{D}{\delta} \right)}{\log\left( \frac{1}{1-\beta} \right)}$},
    \end{align*}
    where, in the second inequality, we have used the fact that, for any $i$, $N_i(m) = 0$ if and only if, for all the samples $m$, the outcome $i$ is never observed.

    The corollary is a direct application of this claim by setting $\beta = p_{\min}$.
\end{proof}

\begin{lemma}[Theorem A.2 in \citet{rajaraman2020toward}]\label{lemma:missing-mass}
    Consider $D \in \mathbb{N}_{>0}$ and a distribution $p \in \Delta_D$. Consider a sequence $P_1, \dots, P_m$ of i.i.d. samples drawn from $p$. For all $i \in [D]$, let $N_i(p) = \sum_{t=1}^m \bm{1} \left\{ P_t = i\right\}$. Then, for $\delta \in (0,1)$, it holds that:
    \begin{align*}
        \mathbb{P}\left( 
        \sum_{i \in [D]} p_i \bm{1} \left\{ N_i(m) = 0 \right\} - \mathbb{E}  \left[  \sum_{i \in [D]} p_i \bm{1}\left\{ N_i(m) = 0 \right\} \right] \ge \frac{{3 \sqrt{D} \log\left(\frac{1}{\delta} \right)}}{n} 
        \right) \le \delta.
    \end{align*}
\end{lemma}

\begin{lemma}\label{lemma:l1-high-prob}
    Consider $D \in \mathbb{N}_{>0}$ and $p \in \Delta_{D}$. Let $\hat p$ be the maximum likelihood estimator of $p$ after $m$ i.i.d. samples. Then, for any $\delta \in (0,1]$, it holds that:
    \begin{align*}
        \mathbb{P}\left( \| p - \hat p_m\|_1 > \sqrt{\frac{2\log(1/\delta)+2D \log(6m)}{m}} \right) \le \delta. 
    \end{align*}
\end{lemma}
\begin{proof}
    
By Pinsker's inequality we have:
\[
\|\hat{p}_m-p\|_1 \leq \sqrt{2\textup{KL}(\hat{p}_m,p)}
\]
where $\textup{KL}(p_1,p_2)$ denotes the Kullback–Leibler divergence between distributions $p_1$ and $p_2$.

Now, applying Proposition 1 in \citet{jonsson2020planning}, we obtain
\[
\Prob\left[\forall m \in \mathbb{N}^*,\: m\textup{KL}(\hat{p}_m,p) \leq \log\left(\frac{1}{\delta}\right) + (D-1)\log \left(e\left(1+\frac{m}{D-1}\right) \right)  \right] \geq 1-\delta.
\]
Thus, we obtained that:
\[
\Prob\left[\|\hat{p}_m-p\|_1 \leq \sqrt{\frac{2}{m} \left( \log\left(\frac{1}{\delta}\right) + (D-1)\log \left(e\left(1+\frac{m}{D-1}\right) \right)\right)  } \right] \geq 1-\delta.
\]
At this point, we observe that:
\[
\sqrt{\frac{2}{m} \left( \log\left(\frac{1}{\delta}\right) + (n-1)\log \left(e\left(1+\frac{m}{D-1}\right) \right)\right) } \leq  \sqrt{\frac{2\log(1/\delta) + 2D \log(6m)}{m}}
\]
Thus concluding the proof.
\end{proof}

\begin{lemma}\label{lemma:kl-divergence-aux-lemma}
    Let $n > 1$ and $v \in \{-1, 1\}^{n-1}$ such that $\sum_{k=1}^{n-1} v_k = 0$. Furthermore, let $\tfrac{1}{2}> \beta > \gamma > 0$, 
     $\mathbb{P} = \left( \frac{\beta + \gamma\beta v_1}{n-1}, \dots, \frac{\beta + \gamma \beta v_{n-1}}{n-1}, 1-\beta\right)$, and $\mathbb{Q} = \left( \frac{\beta}{n-1}, \dots, \frac{\beta}{n-1}, 1-\beta \right)$.
    Then, it holds that:
    \begin{align*}
        \textup{KL}(\mathbb{P}, \mathbb{Q}) \le 2 \beta \gamma^2.
    \end{align*}    
\end{lemma}
\begin{proof}

    We have that: 
    \begin{align*}
        \textup{KL}(\mathbb{P}, \mathbb{Q}) & = \sum_{i \in [n-1]: v_i = 1} \frac{\beta + \gamma\beta }{n-1} \log\left(1+\gamma \right) + \sum_{i \in [n-1]: v_i = -1} \frac{\beta - \gamma\beta }{n-1} \log\left(1-\gamma \right) \\
        & = \frac{1}{2} \beta (1+\gamma) \log(1+\gamma) + \frac{1}{2} \beta (1-\gamma) \log(1-\gamma) \\
        & = \frac{\beta}{2} \left( \log(1-\gamma^2) + \gamma \log(1+\gamma) - \gamma \log(1-\gamma) \right) \\
        & \le \frac{\beta}{2} \left( \gamma \log(1+\gamma) - \gamma \log(1-\gamma) \right) \tag{$\log(1-\gamma^2) \le 0$} \\
        & \le \frac{\beta}{2} \left( \gamma^2 + \gamma\left(\frac{1}{1-\gamma} - 1\right) \right) \tag{$\log(1+x) \le x$ and  $-\log(1-x) \le \frac{1}{1-x} -1$ for $x \in (0,1)$} \\
        & \le 2\beta \gamma^2.
    \end{align*}
This concludes the proof.
\end{proof}

\begin{lemma}[Variance Aware Error Control]\label{lemma:var-aware}
    Let $D \in \mathbb{N}_{>0}$ and $p \in \Delta_D$. Let $\hat p $ be the maximum likelihood estimator of $p$ after $m$ i.i.d. samples where $m > 2$. Let $\beta \in (0,1)$ and let $\mathcal{S}(\beta) = \{ S \subseteq [D]: \forall i \in S, p_i > 0 \text{ and} \sum_{i \in S} p_i \le \beta \}$.

    Then, for $\delta \in (0,1)$ it holds that:
    \begin{align}\label{eq:bernst-ci}
        & \mathbb{P}\left( \exists S \in \mathcal{S}(\beta): \Bigg| \sum_{i \in S} ( p_i - \hat p_i) \Bigg|  > \sqrt{\frac{4\beta \left( D+\log\left( \frac{1}{\delta} \right) \right) }{m}} + \frac{4(D+\log\left( \frac{1}{\delta} \right))}{m}\right) \le \delta.
    \end{align}
    Furthermore, for $\delta \in (0,1)$ it holds that:
    \begin{align}\label{eq:emp-bernst-ci}
        \mathbb{P}\left(\exists S \subseteq [n] > \sqrt{\frac{2 \sum_{i \in S} \hat p_i\left(D+\log(\frac{4}{\delta}) \right)}{m}} + \frac{4 (D + \log(\frac{4}{\delta})}{m-1}\right) \le \delta.
    \end{align}
\end{lemma}
\begin{proof}
    We start by proving \Cref{eq:bernst-ci}.
    Fix a set $S \in \mathcal{S}(\beta)$, then $\sum_{i \in S} \hat p_i$ is the empirical mean of a Bernoulli r.v. with mean $\mu = \sum_{i \in S} p_i$. Hence, for $\delta' \in (0,1)$, a standard application of Bernstein's inequality \citep{boucheron2003concentration} yields 
    \begin{align*}
        \mathbb{P}\left( \Bigg| \sum_{i \in S} ( p_i - \hat p_i) \Bigg|  > \sqrt{\frac{4\beta \log\left( \frac{1}{\delta'} \right)  }{m}} + \frac{4\log\left( \frac{1}{\delta'} \right)}{m}\right)  \le \delta',
    \end{align*}
    where we have used that $\Var(\sum_{i \in S} p_i) \le \sum_{i \in S} p_i \le \beta$.

    Now, for any value of $\beta$, the cardinality of $\mathcal{S}(\beta)$ is bounded by $2^D$. Hence, for any $\delta' \in (0,1)$ we have that: 
    \begin{align*}
        \mathbb{P}\left( \exists S \in \mathcal{S}(\beta): \Bigg| \sum_{i \in S} ( p_i - \hat p_i) \Bigg|  > \sqrt{\frac{4\beta \log\left( \frac{1}{\delta'}  \right) }{m}} + \frac{4\log\left( \frac{1}{\delta'} \right)}{m}\right) \le 2^D \delta'.
    \end{align*}
    Picking $\delta' = \frac{\delta}{2^D}$ concludes the proof.

    \Cref{eq:emp-bernst-ci} follows from similar arguments, but we replaced Bernstein's inequality with its empirical version \cite{maurer2009empirical}. In particular, the empirical Bernstein's inequality, replace the variance in the Bernstein inequality with the sample variance. In the following, we show how the sample variance of a Bernoulli r.v. with mean $q$ can be upper bouded with $2 \hat q$. Let   
    Denote by $Z_1, \dots, Z_m \sim q$ the $m$ independent samples from $q$. Then:
    \begin{align*}
        \hat{\Var}(\hat q) & = \frac{1}{m(m-1)} \sum_{i < j} (Z_i - Z_j)^2 \\ 
        & = \frac{1}{m-1} \sum_{i=1}^m (Z_i - p)^2 \\ &
        = \frac{m}{m-1} \hat q (1-\hat q) \\
        & \le 2 \hat q \tag{For $m \ge 2$}
    \end{align*}
    From this follows the result of \Cref{eq:emp-bernst-ci}.
\end{proof}

The following change of measure arguments follows from the same arguments of the ones used in \citet{metelli2023towards}. 

\begin{lemma}[Change of measure argument]\label{lemma:ours-change-of-measure}
    Let $\alpha \ge 0$ and consider any $(\epsilon, \delta)$-correct algorithm. Let  $(x^0, y^0)$ and $(x^1, y^1)$ be the strategy profiles of two different problem instances. Then, the following holds:
    \begin{align}
        & H(\mathcal{G}_\alpha(x^0, y^0), \mathcal{G}_\alpha(x^1, y^1)) \ge 2 \epsilon \implies \delta \ge \frac{1}{4} \exp\left( -\textup{KL}(\mathbb{P}_{x^0, y^0}, \mathbb{P}_{x^1, y^1}) \right) \label{eq:com-eq-1}\\
        & H(\mathcal{Z}_\alpha(x^0, y^0), \mathcal{Z}_\alpha(x^1, y^1)) \ge 2 \epsilon \implies \delta \ge \frac{1}{4} \exp\left( -\textup{KL}(\mathbb{P}_{x^0, y^0}, \mathbb{P}_{x^1, y^1}) \right) \label{eq:com-eq-2}
    \end{align}
    where $\mathbb{P}_{x, y} = \prod_{t=1}^{\tau_\delta} p_{x}(X_t) p_y(Y_t)$, with $(X_t,Y_t)$ denoting the joint action sampled at time $t$ and $p_q(\cdot)$ denoting the density function of $q \in \Delta_n$.
\end{lemma}

\begin{lemma}[Change of measure with several instances]\label{lemma:ours-change-of-measure-several-instances}
Let $\alpha \ge 0$ and $\{ (x^k, y^k) \}_{k=1}^K$ be $K$ distinct instances. Let $(x^0, y^0)$ be a reference instance. It holds that:
\begin{align*}
    & \forall i,j \in [K], i \ne j, H(\mathcal{G}_\alpha(x^i,y^i), \mathcal{G}_\alpha(x^j,y^j)) \ge 2\epsilon \implies \delta \ge 1-\frac{1}{\log K} \left( \frac{1}{K}\sum_{k \in [K]} \textup{KL}(\mathbb{P}_{x^i, y^i}, \mathbb{P}_{x^0, y^0}) + \log2  \right) \\ 
    & \forall i,j \in [K], i \ne j, H(\mathcal{Z}_\alpha(x^i,y^i), \mathcal{Z}_\alpha(x^j,y^j)) \ge 2\epsilon \implies \delta \ge 1-\frac{1}{\log K} \left( \frac{1}{K}\sum_{k \in [K]} \textup{KL}(\mathbb{P}_{x^i, y^i}, \mathbb{P}_{x^0, y^0}) + \log2  \right),
\end{align*}
where $\mathbb{P}_{x, y} = \prod_{t=1}^{\tau_\delta} p_{x}(X_t) p_y(Y_t)$, and $p_q(\cdot)$ denotes the density function of $q \in \Delta_n$.
\end{lemma}

\subsection{Others}

\begin{lemma}\label{lemma:nash-support}
Let $A\in[-1,1]^{n\times n}$, $y\in\Delta_n$, $x\in\Delta_n$.
Then it holds that:
\begin{align}
& \text{(i)}~ x^\top A y \le e_i^{\top}Ay\ \ \forall i\in[n]
\qquad\Longleftrightarrow\qquad
\text{(ii)}\ 
\begin{cases}
e_i^{\top}Ay = e_j^{\top}Ay & \forall\, i,j\in \supp(x),\\
e_i^{\top}Ay \le e_j^{\top}Ay & \forall\, i\in \supp(x),\ \forall\, j\notin \supp(x).
\end{cases} \label{eq:support-nash-1}
\\
& \text{(i)}~ x^\top A e_j \le  x^\top A y \le e_i^\top Ay\ \ \forall i,j\in[n]
~\Longleftrightarrow~
\text{(ii)}\ 
\begin{cases}
x^{\top}Ay = x^{\top}Ae_j =e_i^TAy & \forall\, i \in \supp(x),j\in \supp(y),\\
x^\top A e_j \le x^{\top}Ay \le e_i^\top A y & \forall\, i\notin \supp(x),\ \forall\, j\notin \supp(y).
\end{cases} \label{eq:support-nash-3} 
\end{align}
\end{lemma}

\begin{proof}
This is a well-known result (e.g.~\citet[Section~3.6]{gintis2000game}). For completeness, we prove \Cref{eq:support-nash-1}. The proof of \Cref{eq:support-nash-3} follows analogous arguments. We first show that (i) $\Rightarrow$ (ii).
Assume $x^\top A y \le e_i^{\top}Ay$ for all $i \in [n]$.
Rewrite $x^\top A y$ as follows:
\[
x^\top A y 
= \sum_{k \in [n]} x_k \sum_{j \in [n]} y_j A_{ij}
= \sum_{k\in \supp(x)} x_k \sum_{j \in [n]} y_j A_{ij},
= \sum_{k\in \supp(x)} x_k \left( e_k^\top A y \right)
\]
since $x_k=0$ for $k\notin \supp(x)$. Note that $\sum_{k \in S}x_i =1$.
Then, since $x_k>0$ for $k\in \supp(x)$ and each $e_k^{\top}Ay \ge x^\top A y$ (by assumption), it must follows that $e_k^\top A y = x^\top A y$ for all $k \in \supp(x)$. Hence, we have proved the first condition of (ii) and it remains to prove that for all $j \notin \supp(x)$ and all $i \in \supp(x)$, we have that $e_i^\top A y \le e_j^\top A y$.

Fix any $j \notin \supp(x)$. By assumption (i), we have $x^\top A y \le e_j^\top A y$.
Combining this with the equality $e_i^\top A y = x^\top A y$ for all $i \in \supp(x)$, we obtain
\[
e_i^\top A y \le e_j^\top A y
\qquad \forall\, i \in \supp(x),\ \forall\, j \notin \supp(x),
\]
which concludes the proof of $(i) \Rightarrow (ii)$.

We conclude the proof by showing that (ii) $\Rightarrow$ (i).
Assume (ii) holds. Let $c$ denote the common value $e_i^{\top}Ay$ for $i\in \supp(x)$. Then:
\[
x^{\top}Ay = \sum_{i\in \supp(x)} x_i\, \left( e_i^{\top}Ay\right) = \sum_{i\in S} x_i\, c = c.
\]
Hence, $x^\top A y \le e_i^\top A y$ for all $i \in \supp(x)$.

For any $j \notin \supp(x)$, (ii) gives $e_j^{\top}Ay \ge c=x^\top A y$, which concludes the proof.
\end{proof}

\begin{lemma}\label{lemma:tech-lemma}
Let $c_1, c_2,K \in \mathbb{R}_{> 0}$ and consider 
$$
\bar{T} = \inf \left\{ t \in \mathbb{N}: \frac{\log(c_1/\delta) + n \log(c_2t)}{t} \le K^2 \right\}.
$$
Then:
\begin{align*}
    \bar{T} \le 2 \left( 2 + \frac{4\log(c_1/\delta)}{K^2} + \frac{4n}{K^2} \log\left( \frac{2nc_2}{K^2} \right)\right).
\end{align*}
\end{lemma}
\begin{proof}
First, it is easy to see that:
    \begin{align*}
        \bar{T} & \le \max \left\{ \inf \left\{t \in \mathbb{N}: \frac{\log(c_1/\delta)}{t} \le \frac{K^2}{2} \right\}, \inf \left\{t \in \mathbb{N}: \frac{n\log(c_2t)}{t} \le \frac{K^2}{2} \right\}  \right\} \\ 
        & \le 2 \left( \inf \left\{t \in \mathbb{N}: \frac{\log(c_1/\delta)}{t} \le \frac{K^2}{2} \right\} + \inf \left\{t \in \mathbb{N}: \frac{n\log(c_2t)}{t} \le \frac{K^2}{2} \right\} \right) \\
        & \le 2 \left( 1 + \frac{4\log(c_1/\delta)}{K^2} + \inf \left\{t \in \mathbb{N}: \frac{n\log(c_2t)}{t} \le \frac{K^2}{2} \right\} \right) \\ 
        & \coloneqq 2 \left( 1 + \frac{4\log(c_1/\delta)}{K^2} + \tilde{T} \right).
    \end{align*}
    
    It remains to analyze the term: $$\tilde{T} =\inf \left\{t \in \mathbb{N}: \frac{n\log(c_2t)}{t} \le \frac{K^2}{2} \right\} = \left\{ t \in \mathbb{N}: t \ge \frac{2n \log(c_2t)}{K^2} \right\}.$$
    By definition, for $\tilde{T}-1$, we have that:
    \begin{align}\label{eq:tech-lemma}
        \tilde{T} - 1 < \frac{2n \log(c_2(\tilde{T}-1))}{K^2}. 
    \end{align}
    Since $\log(x) \le \sqrt{x}$, we have that:
    \begin{align*}
        \sqrt{\tilde{T} - 1} < \frac{2n\sqrt{c_2}}{K^2},
    \end{align*}
    which leads to:
    \begin{align*}
        \tilde{T} - 1 < \left( \frac{2n\sqrt{c_2}}{K^2} \right)^2.
    \end{align*}
    Plugging this result within \Cref{eq:tech-lemma}, we obtain:
    \begin{align*}
        \tilde{T} < \frac{4n}{K^2} \log\left( \frac{2nc_2}{K^2} \right) + 1,
    \end{align*}
    which leads to the desired result.
\end{proof}

\begin{lemma}[General Decomposition Error]\label{lemma:tech-decomposition-lemma}
    Let $(\mathcal{Y}, d)$ be a metric space and let $\mathcal{X}$ be any set.
    Consider a compact-valued correspondence $F: \mathcal{X} \times \mathcal{X} \mapstoto \mathcal{Y}$, and let $x_1, x_2, \hat{x}_1, \hat x_2 \in \mathcal{X}$. It holds that:   
    \begin{align*}
        H_d(F(x_1,x_2), F(\hat x_1, \hat x_2)) \le H_d(F(x_1,x_2), F(\hat x_1, x_2)) + H_d(F(\hat x_1,  x_2), F( \hat x_1, \hat x_2)).
    \end{align*}
\end{lemma}
\begin{proof}
    For any $y \in \mathcal{Y}$, let $c(y) \in \argmin_{c \in F(\hat x_1, x_2)} d(y,c)$. Note that, by Heine-Cantor, this minimum is attained since $d(y, \cdot)$ is a continuous function and $F$ is compact-valued. Then, recall that:
    \begin{align}\label{eq:decomposition-lemma-eq-1}
        H_d(F(x_1,x_2), F(\hat x_1, \hat x_2) & = \max \left\{ \sup_{y \in F(x_1,x_2)} \inf_{\hat y \in F(\hat x_1, \hat x_2)} d(y, \hat y), \sup_{\hat y \in F(\hat x_1,\hat x_2)} \inf_{ y \in F( x_1, x_2)} d(y, \hat y)  \right\} 
    \end{align}
    Now, focus on the first term. We have that:
    \begin{align*}
        \sup_{y \in F(x_1,x_2)} \inf_{\hat y \in F(\hat x_1, \hat x_2)} d(y, \hat y) & \le \sup_{y \in F(x_1,x_2)} \inf_{\hat y \in F(\hat x_1, \hat x_2)} \left( d(y, c(y)) + d(c(y), \hat y) \right) \tag{Triangular inequality} \\
        & = \sup_{y \in F(x_1,x_2)} \left( d(y, c(y)) + \inf_{\hat y \in F(\hat x_1, \hat x_2)}   d(c(y), \hat y) \right) \\
        & \le \sup_{y \in F(x_1,x_2)}  d(y, c(y)) +  \sup_{y \in F(x_1,x_2)} \inf_{\hat y \in F(\hat x_1, \hat x_2)}   d(c(y), \hat y)  \\
        & \le H_d(F(x_1, x_2), F(\hat x_1, x_2)) + \sup_{y \in F(x_1,x_2)} \inf_{\hat y \in F(\hat x_1, \hat x_2)}   d(c(y), \hat y) \tag{Def. of $c(y)$} \\
        & \le H_d(F(x_1, x_2), F(\hat x_1, x_2)) + H_d(F(\hat x_1, x_2), F(\hat x_1, \hat x_2)) \tag{$c(y) \in F(\hat x_1, x_2)$}
    \end{align*}
    Similarly, by using the triangular inequality with $c(\hat y)$, one has that:
    \begin{align*}
        \sup_{\hat y \in F(\hat x_1,\hat x_2)} \inf_{ y \in F( x_1, x_2)} d(y, \hat y) \le H_d(F(x_1, x_2), F(\hat x_1, x_2)) + H_d(F(\hat x_1, x_2), F(\hat x_1, \hat x_2)).
    \end{align*}
    Plugging these results within \Cref{eq:decomposition-lemma-eq-1}, concludes the proof.
\end{proof}

\begin{corollary}[Decomposition Error in Inverse Game Theory]\label{corollary:separation-x-y-err}
    Let $x,y,\hat{x},\hat{y} \in \Delta_n$. Then, for $\alpha \ge 0$, it holds that:
    \begin{align*}
        & H(\mathcal{G}_{\alpha}(x,y), \mathcal{G}_\alpha(\hat{x}, \hat{y}))\le H(\mathcal{G}_{\alpha}(x,y), \mathcal{G}_{\alpha}(\hat x,y)) + H(\mathcal{G}_{\alpha}(\hat x,y), \mathcal{G}_{\alpha}(\hat x, \hat y)) \\
        & 
        H(\mathcal{Z}_{\alpha}(x,y), \mathcal{Z}_\alpha(\hat{x}, \hat{y}))\le H(\mathcal{Z}_{\alpha}(x,y), \mathcal{Z}_{\alpha}(\hat x,y)) + H(\mathcal{Z}_{\alpha}(\hat x,y), \mathcal{Z}_{\alpha}(\hat x, \hat y))
    \end{align*}
\end{corollary}
\begin{proof}
    This is a direct application of \Cref{lemma:tech-decomposition-lemma}.
\end{proof}

\begin{lemma}[Packing argument \cite{metelli2023towards}, Lemma E.6] \label{lemma:pack-arg}
Let
\(
\mathcal{V}
\;=\;
\Bigl\{ v \in \{-1,1\}^D : \sum_{j=1}^D v_j = 0 \Bigr\}.
\)
Then the $\tfrac{D}{16}$-packing number of $\mathcal{V}$ with respect to the metric
\(
d(v,v') \;=\; \sum_{j=1}^D \lvert v_j - v'_j \rvert
\)
is lower bounded by $2^{D/5}$.    
\end{lemma}

\end{document}